\documentclass[pra,superscriptaddress,aps,showkeys]{revtex4}

\usepackage{amsmath,epsfig,amsfonts,graphicx,hyperref}
\usepackage{subfigure}

\graphicspath{{./figs/}}

\usepackage{graphicx}
\usepackage[american]{babel}
\usepackage{amssymb}
\usepackage{amsfonts}
\usepackage{color}
\usepackage[normalem]{ulem}

\newenvironment{proof}[1][Proof]{\textit{#1:} }{ $\square$}
\newcommand{\figscaleee}{0.38}

\providecommand{\keywords}[1]
{
 \small 
 \textbf{\textit{Keywords---}} #1
 }
\begin{document}
	\newcommand{\be}{\begin{equation}}
	\newcommand{\ee}{\end{equation}}
	\newtheorem{corollary}{Corollary}[section]
	\newtheorem{remark}{Remark}[section]
	\newtheorem{definition}{Definition}[section]
	\newtheorem{theorem}{Theorem}[section]
	\newtheorem{proposition}{Proposition}[section]
	\newtheorem{lemma}{Lemma}[section]
	\newtheorem{help1}{Example}[section]

\title{The Lefever-Lejeune nonlinear lattice: convergence dynamics and the structure of equilibrium states}

\author{N. I. Karachalios, P. Kyriazopoulos \footnote{ Corresponding Author.}  and K. Vetas}
\affiliation{Department of Mathematics, University of the Aegean, Karlovassi, 83200  Samos, Greece}

\begin{abstract}
We consider the Lefever-Lejeune nonlinear lattice, a spatially discrete propagation-inhibition model describing the growth of vegetation densities in dry-lands. We analytically identify parametric regimes distinguishing between decay (associated with spatial extinction of vegetation patches) and potentially non-trivial time-asymptotics.  To gain insight on the convergence dynamics, a stability analysis of spatially uniform states is performed, revealing the existence of a threshold for the discretization parameter which depends on the lattice parameters, below which their destabilization occurs and spatially non-uniform equilibrium states may emerge. Direct numerical simulations justified that the analytical stability criteria and parametric thresholds effectively describe the above transition dynamics and revealed the rich structure of the equilibrium set. Connections with the continuous sibling Lefever-Lejeune partial differential equation are also discussed.
\end{abstract}

\keywords{
Lefever-Lejeune, lattice dynamics, decay estimates, pattern formation, dryland vegetation
}

\maketitle

\section{Introduction}

Nonlinear spatially discrete and continuous in time systems, namely, nonlinear lattices, play a significant role in the understanding of physical systems. Fascinating nonlinear phenomena, as the energy equipartition in nonlinear systems and energy localization, have been effectively described by nonlinear lattices as the discrete Klein-Gordon and  the discrete Nonlinear Schr\"{o}dinger equations and their various extended variants.
These fundamental phenomena have been proved to be the underlying mechanisms for the emergence of stationary and travelling inherently discrete localized waveforms in solids, condensed matter, optical fibers and wave guides, even explaining the self-trapping of vibrational energy in proteins and DNA double strand denaturation, \cite{Eil,Braun2004, Panos_book,reviewsC,peyrard}.  

Spatial discretizations of reaction-diffusion systems define another important class of nonlinear lattices, which may exhibit a rich  structure of their equilibrium set, pattern formation and invasion dynamics, even spatiotemporal chaos. These systems have been used as effective models to describe a multitude of phenomena, from phase separation in binary alloys and glasses,  the excitation of travelling pulses in myelinated nerve axons, to pattern formation in cellular networks.  Characteristic examples are the discrete Allen-Cahn type equations with monostable or bistable nonlinearities, the discrete Cahn-Hiliard, and the discrete Swift-Hohenberg equations, \cite{Cook1969, Zinner1992, Zinner1993,Cahn1995,Chow1996,Chow1995MalParI,Chow1995MalParII,Chen2002,Carpio2003,KJ1987,Abell2000,Kusdiantara2017}.

In the present paper we consider a strongly nonlinear lattice, claiming its relevance with yet another exciting theme: the nonlinear physics of ecosystems, and particularly, the vegetation pattern formation process \cite{EhudBook}. The model is the discrete Lefever-Lejeune equation (DLL):
\begin{eqnarray}
\label{eq1}
\dot{U}_n+\frac{\gamma_1}{h^4}U_n\Delta_d^2U_n+ \frac{\gamma_2}{h^2}U_n\Delta_dU_n-\frac{\gamma_3}{h^2}\Delta_dU_n -f(U_n)=0,
\end{eqnarray}
endowed with the initial condition 
\begin{eqnarray}
\label{inc}
U_n(0)=U^0_n,
\end{eqnarray}
and supplemented with suitable boundary conditions, that will be discussed below. In Eq.~(\ref{eq1}),  $U_n(t)$ is the unknown function occupying the lattice site $n$, the parameters $\gamma_i>0$, $i=1,2,3$, while
$h>0$ stands for the lattice spacing.  The nonlinearity 
\begin{eqnarray}
\label{eq2}
f(U)=\alpha U+\beta U^2-U^3,\;\;\alpha,\beta\in\mathbb{R}.
\end{eqnarray} 
The linear operator 
\begin{eqnarray}
\label{eq3}
\left\{\Delta_dU\right\}_{n\in\mathbb{Z}}=U_{n+1}-2U_n+U_{n-1},
\end{eqnarray}
is the one-dimensional discrete Laplacian. Then, $\Delta^2_dU:=\Delta_d[\Delta_dU]$, defines the associated {\em discrete biharmonic operator}, i.e.,
\begin{eqnarray}
\label{eq4}
\left\{\Delta_d^2U\right\}_{n\in\mathbb{Z}}=U_{n+2}-4U_{n+1}+6U_n-4U_{n-1}+U_{n-2}.
\end{eqnarray} 
In the above discrete set-up, the lattice \eqref{eq1} can be viewed as a  discretization of the Lefever-Lejeune (LL) partial differential equation (PDE),  which in a non-dimensional form reads as
\begin{eqnarray}
\label{eq5}
U_t+\gamma_1UU_{xxxx}+\gamma_2UU_{xx}-\gamma_3U_{xx}-f(U)=0.
\end{eqnarray}
To value the lattice \eqref{eq1}, let us recall some information on the continuous LL model. Equation \eqref{eq5} is a spatially continuous propagation-inhibition model
describing the growth of vegetation density in dry-lands, and is the formal continuum limit of the lattice \eqref{eq1}, as $h\rightarrow 0$. In such resource poor environments, spatial patterns of vegetation are observed,
and LL attempts to explain their formation attributing it to a short-range cooperative and long-range competitive spatial mechanism. It should be remarked, that the  original LL model \cite{LL1997}, is a spatially non-local integral-differential equation which involves a continuous redistribution-kernel convoluted with
a density dependent nonlinearity, encapsulating the dispersal and spatial interactions of individuals. Although the kernel-based models \cite{Borgogno2009} are considered  as more 
realistic, since they capture accurately enough global spatial-interactions involving kernel shapes that are common in nature, the Lefever-Lejeune PDE is a 
biharmonic approximation which is commonly preferred. 
The reason is two-fold: it successfully derives the qualitative behavior of plant 
community systems, and offers a simpler template for numerical investigations and mathematical analysis. 
In the differential form, the short-range corporative interplay among plants is expressed by a linear diffusion term and the non-linear local growth term, 
while non-linear biharmonic and Laplacian diffusion term with negative coefficient imprint the long-range competition for resources.

Numerical and analytical studies in one and two spatial dimensions, 
have revealed the pattern forming potential of the spatially continuous LL equation \cite{LL1997}.
Besides the existence of Turing periodic patterns, 
the LL produces localized solutions such as isolated spots of vegetation, or groups of spots confined 
by the homogeneous zero vegetation \cite{Tlidi2008}. Furthermore, the self-replication capabilities of the equation have been investigated, 
showing that for a particular regime of the aridity parameter in two spatial dimensions, a single circular vegetation patch can destabilize leading to
an elliptical deformation, followed by patch multiplication \cite{Bordeu2016}. This patch splitting phenomenon evolves in time
until the system reaches a self-organized hexagonal pattern.

In the above physical context, reaction-diffusion lattice  systems can be used to model vegetation patterns dynamics, in a situation where space is viewed
as a collection of patches (or cells). Each patch hosts a portion of the population, and so, the entire population can be represented by a lattice vector. According to that point of view, an ODE equation considered on a patch involves apart of the coupling among vector components, dispersal, migration or other effects.  One of the first works where such a model has been used, particularly a first order quasilinear lattice of the form of \eqref{eq1} with $\gamma_1=0$ and $f(U)=U(1-U)$, is \cite{Ares2003}. Therein, numerical results on pattern formation where produced in a random version of the model for the action of diffusion (assuming a random walk pattern).

Regarding such a lattice approach, to highlight the connections and differences between the DLL and its continuous counterpart, it is important to recall the physical meaning of the parameters involved. The parameter $\alpha =1-\mu$, where $\mu$ represents the mortality to growth rate ratio. 
Therefore, $\mu$ and in turn, $\alpha$,  can be interpreted as a measure of the environment's aridity which characterizes the productivity of the system.
In the local nonlinearity $f$,  the parameter $\beta=\Lambda-1$, where $\Lambda$ represents the cooperation effect influencing the local reproduction. 
This effect is considered weak for  $\beta \leq 0$ ($\Lambda\leq 1$) and strong for $\beta>0$ ($\Lambda>1$). 
The higher order derivative nonlinearities and the linear diffusion terms of strengths $\gamma_i$, are  modelling  the long and  short-range interaction effects, respectively,  between the vegetation  patches. Particularly, $\gamma_3=\frac{1}{2}l^2$,
with $l$ representing the ratio of facilitative to inhibitory interaction zones. The latter are the spatial distances beyond which cooperation and competition becomes negligible.
%

Evidently, with the presence of  the discreteness effect, the DLL \eqref{eq1} seems to be particularly relevant to describe the interaction mechanisms between the vegetation densities of the patches. In the discrete system \eqref{eq1}, the strength of spatial interactions  depends explicitly on the distance between the patches $h$, which appears in the coefficients  of the discrete spatial-operators. For instance, one of the most interesting effects to highlight, which can't be modelled by the continuous counterpart, is the non-homogenous dependence of  the spatial interactions.
It becomes clear that as $h$ decreases, the strength of long-range (next neighbor) competition  increases much faster than the short-range (next neighbor) facilitation. 
Similarly, for $h\ll1$, both spatial interactions are weakened. Due to such an interplay between discreteness and nonlinearity incorporating the above dependencies, we may expect novel effects in the dynamics. 

Up to our knowledge, the present paper is a first attempt to consider the DLL equation, investigating in a rather systematical manner, well-posedness and dynamical features of the model.  The presentation of the results has as follows. In the analytical considerations, the  DLL equation \eqref{eq1} will be supplemented with vanishing boundary conditions in an infinite lattice, which corresponds to the case of an infinite dimensional dynamical system, and with Dirichlet and periodic boundary conditions,  cases which correspond to  a finite dimensional one.  Some extra care is given in Section II, to discuss the system's set-up with the above conditions, due to the presence of the higher order discrete operators; we also present some of their properties, while further details are given in the Appendix \ref{App}. The local existence in the sequence phase spaces $\ell^2$ and $\ell^1$ is established with fixed point arguments. The latter case in the non-reflexive Banach space $\ell^1$  highlights an unusual feature of the local in time semiflow, to be strongly continuous. This feature is a result of the Schur property of $\ell^1$, where weak convergence coincides with strong (in norm) convergence. 

In Section III, we prove uniform bounds for the global existence of solutions. We distinguish between two cases, both being physically significant. In the first case, we identify several parametric regimes for the parameters $\alpha$ and $\beta$, for the global stability of the trivial steady-state. This case is of particular importance since it can be associated with parametric conditions for the spatial extinction of vegetation densities, and consequently, with the emergence of desertification \cite{sys0,hiD2}. The second one, which is the generic case of uniform in time bounds is associated with the existence of an attracting set, for all the cases of the boundary conditions considered, and the potential convergence to non-trivial equilibrium states. Yet the Schur property, may imply some interesting observations revolving around the uniform compactness of a restricted semiflow on $\ell^1$. It should be warned that the derivation of the above  estimates is heavily depending on the properties of the discrete phase spaces, and it is an interesting question in what extend they are valid in the case of the LL-PDE--see Remark \ref{DtC}.
 
To elucidate further the dynamics, in section IV-A, we perform a detailed linear stability analysis of spatially uniform states, revealing the crucial role of discreteness in their  destabilization, and the  potential emergence of spatially non-uniform states.  Remarkably, we found a critical threshold for the discretization parameter $h$, depending on the amplitude of the unstable homogenous state and the other parameters of the lattice, below which the above destabilization occurs.  A first interesting outcome of the linear stability analysis is the full coincidence of the derived criteria for the stability of the trivial solution, with those derived by the analytical energy estimates for its global stability; in some cases, such a coincidence is in excellent quantitative agreement. A second important outcome is that for initial conditions which are harmonic perturbations of homogeneous states, we were able to derive analytical conditions depending on the lattice parameters, which dictate the period and shape of the potential resulting steady-state.

The analytical arguments are corroborated with direct numerical simulations, whose results are presented in Section IV-B.  First, although we are unable to prove that the system possesses a gradient structure, in all the numerical experiments we identify convergence to an equilibrium. Investigating numerically the analytical stability criteria, we found that generically, they effectively describe the qualitative behaviour of the DLL system regarding the existence of thresholds distinguishing the convergence  to  spatially homogeneous, from  the convergence to spatially  non-homogeneous states and pattern formation.  Furthermore, in the latter case, the numerical simulations revealed that varying the discretization parameter, convergence towards geometrically distinct profiles occurs, and that the equilibrium set even in the one-dimensional model, possesses a richer structure than its continuous limit; a wealth of steady-states is included, from spatially-periodic, mosaics of periodic states and localized states in a periodic background, to localized single or multi-spike equilibria. Although a detailed bifurcation analysis is not within the scopes of the present work, the numerical results  glimpse on the role of the productivity gradient and other parameters of the lattice for a dynamical transition between such different classes of equilibrium states.

The last section summarizes our results and comment on some potential future ideas for extending them to other relevant nonlinear lattice  models for dryland vegetation. 
\section{Functional set-up and local existence of solutions}
This section is devoted to the functional set up of the problem according to the implemented boundary conditions, the description of various properties of the involved linear and nonlinear operators, concluding with the discussion of local existence of solutions. Some interesting mathematical implications appear when the problem assumes initial data in a non-reflexive Banach space.
\subsection{Boundary conditions and phase spaces}
Equation (\ref{eq1}), will be supplemented  with either vanishing conditions in the case of the infinite lattice, and periodic or Dirichlet boundary conditions, which give rise to a finite dimensional system. The former case involves the standard infinite dimensional sequence spaces, while the latter cases are associated with their relevant finite dimensional subspaces.  In each case of boundary conditions, we will recall the definition and properties of the relevant phase spaces for each case of boundary conditions, with additional details given in the Appendix \ref{App}, as well as, properties of the discrete Laplacian and biharmonic operator in these functional set-ups.
\paragraph{Vanishing boundary conditions in an infinite lattice.} In the case of the infinite lattice with vanishing boundary conditions 
\begin{eqnarray}
\label{vanv}
\lim_{|n|\rightarrow\infty}U_n=0,
\end{eqnarray}
the problem will be considered in the standard infinite dimensional sequence spaces, $\ell^p$, $1\leq p\leq\infty$, for which the definition and their key inclusion properties are included in the Appendix \ref{App}. The discrete Laplacian as a linear operator  $\Delta_{d}: {\ell}^2 \to {\ell}^2$, is selfadjoint, as it satisfies 
\begin{eqnarray}
\label{lp6}
(\Delta_{d}U,U)_{\ell^2}&=&-\sum_{n\in\mathbb{Z}}|U_{n+1}-U_n|^2\leq 0,\\
\label{lp7}
(\Delta_{d}U,W)_{\ell^2}&=&-\sum_{n\in\mathbb{Z}}(U_{n+1}-U_n)(W_{n+1}-W_n)=(U,\Delta_{d}W)_{\ell^2}, \quad U,\;W\in {\ell}^2,
\end{eqnarray}
while for its conitnuity
\begin{eqnarray}
\label{lp8}
||\Delta_d U||_{\ell^2}^2\leq 4||U||_{\ell^2}^2,
\end{eqnarray}
see \cite{DNLS2005}. Generically, the operator $\Delta_d:\ell^p\rightarrow \ell^p$, for $1\leq p\leq\infty$ is continuous, that is there exists a constant $C>0$, such that
\begin{eqnarray}
\label{genlp8}
||\Delta_d U||_{\ell^p}\leq C||U||_{\ell^p},\;\;\mbox{for all}\;\; U\in\ell^p.
\end{eqnarray}

The discrete biharmonic operator $\Delta_d^2:\ell^2\rightarrow\ell^2$ is also continuous, satisfying due to \eqref{lp8}, the inequality:
\begin{eqnarray}
\label{lp9}
||\Delta_d^2U||_{\ell_2}^2=||\Delta_d[\Delta_dU]||_{\ell_2}^2\leq 4||\Delta_dU||^2_{\ell^2}\leq 16||U||^2_{\ell^2}.
\end{eqnarray}
Note that for \eqref{lp9}, in order to apply \eqref{lp8}, we have used that $\Delta_dU\in \ell^2$ in the case of the infinite lattice, which is a direct consequence of the fact that $U\in\ell^2$.  Similarly, when $1\leq p\leq\infty$, we have that 
\begin{eqnarray}
\label{genlp9}
||\Delta_d^2 U||_{\ell^p}\leq C||U||_{\ell^p},\;\;\mbox{for all}\;\; U\in\ell^p.
\end{eqnarray}

\paragraph{Dirichlet boundary conditions.} To define a finite dimensional lattice dynamical system from Eq.~\eqref{eq1}, we assume that an arbitrary number of $N+1$ nodes are occupying equidistantly the interval $\Omega=[-L,L]$, with lattice spacing $h=2L/N$. Accordingly, the discrete spatial coordinate is $x_n=-L+nh$, for $n= 0,1,2,\ldots,N$, and in Eq.~\eqref{eq1}, the function $U_n(t)=U(x_n,t)$.  In the case of Dirichlet boundary conditions some extra care should be paid due to the presence of the discrete biharmonic operator $\Delta_d^2$. We believe that it is useful to discuss the set-up herein for completeness as well as for future considerations.  This set-up is motivated by the analogies with the continuous counterpart \eqref{eq5}, where the homogeneous Dirichlet boundary conditions of the first kind $U(-L)=U(L)=0$ and $U_{x}(-L)=U_{x}(L)=0$, as well as, of the second kind $U(-L)=U(L)=0$ and $U_{xx}(-L)=U_{xx}(L)=0$, should be imposed to construct a well-defined system.  Then, the discrete Dirichlet boundary conditions of the second kind are analogously defined as follows:
\begin{align}
\label{D1a}
U_0&=U_{N}=0,\\
\label{D2a}
\Delta_d U_n&=U_{n+1}-2U_n+U_{n-1}=0\;\text{for} \;n=0\;\mbox{and}\;n=N.
\end{align}
\begin{figure}[tbh!]
	\label{fig:D1}
	\includegraphics{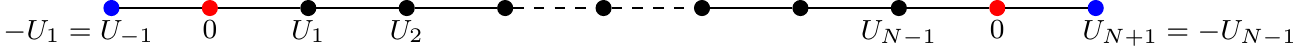}	
	\caption{Implementation of the Dirichlet boundary conditions of the second kind in the lattice, where the Dirichlet Laplacian satisfies $\Delta_d=0$  and gives rise to the antisymmetric conditions \eqref{D1}-\eqref{D2}.}
\end{figure}
The first condition \eqref{D1a} implies initially that   $U\in\ell^2_0$, the finite dimensional subspace of $\ell^2$ (see Appendix \ref{App}). However, the second condition \eqref{D2a} imposes conditions on the extra points $U_{-1}$ and $U_{N+1}$;  $U_{-1}$ is supposed to be positioned at $x_{-1}=L-h$ and $U_{N+1}$ at  $x_{N+1}=L+h$, i.e., outside $[-L,L]$. For instance, \eqref{D2a} implies the antisymmetric conditions:
\begin{eqnarray}
\label{D1}
&&\text{for}\quad n=0,\;\; U_{-1}=-U_{1},\;\;\mbox{and}\\
\label{D2}
&&\text{for}\quad n=N,\;\; U_{N+1}=-U_{N-1}.
\end{eqnarray}
The above implementation is illustrated in the cartoon of Fig. 1. Since we have $N-1$ equations at the interior points $x_1\dots x_{n-1}$, two options arise for the exterior points, as suggested by the conditions \eqref{D1}-\eqref{D2}, in order to construct a well-defined system in the presence of the discrete biharmonic operator $\Delta_d^2$. 

The first one may not pre-assume for \eqref{D1}-\eqref{D2}, that they are locked to zeros.  With such an option, we may proceed to the well-definition of $\Delta_d^2$:  in the neighbouring points $U_1, U_2$, of the boundary point $U_0$, it takes the values:
\begin{align}
\Delta_d^2 U_1&= U_{3}-4U_{2}+6U_1-4U_{0}+U_{-1}=U_{3}-4U_{2}+5U_1,\\
\Delta_d^2 U_2&= U_{4}-4U_{3}+6U_2-4U_{1}+U_{0}=U_{4}-4U_{3}+6U_2-4U_1.
\end{align}
In the interior points $U_n$ for  $n=3,\dots, N-3$, we have by its definition \eqref{eq4}:
\begin{equation}
\Delta_d^2 U_n=U_{n+2}-4U_{n+1}+6U_n-4U_{n-1}+U_{n-2}.
\end{equation}
In the neighbouring points $U_{N-2}, U_{N-1}$ of $U_N$, it takes the values: 
\begin{align}
\Delta_d^2 U_{N-2}& = U_{N-4}-4U_{N-3}+6U_{N-2}-4U_{N-1}+U_{N} = U_{N-4}-4U_{N-3}+6U_{N-2}-4U_{N-1}\\
\Delta_d^2 U_{N-1} &= U_{N-3}-4U_{N-2}+6U_{N-1}-4U_{N}+U_{N+1} = U_{N-3}-4U_{N-2}+5U_{N-1}.
\end{align}
In the first option, the operators $\Delta_d$, and $\Delta_d^2$ have the matrix formulation:
\[
\Delta_d =
\begin{pmatrix} 
-2 & 1           &    & &  &    \\ 
1  &  -2        &1& &&   \\  
&  1         &-2&1 & &   \\  
&   &\ddots&\ddots &\ddots& \\  
&   && 1&-2 & 1\\  
&         &    &&1 &  -2 
\end{pmatrix}_{(N-1)\times(N-1)},\,
\Delta^2_ d=
\begin{pmatrix} 
5& -4  &  1 & &  &   & &&\\ 
-4  &  6 &-4&1 &   &   &&&\\  
1  &  -4 &6&-4 &1   &   &&&\\  
&\ddots &\ddots  &\ddots&  \ddots&\ddots&& &\\  
&   &1 &  -4 &6&-4 &1    &&\\  
&   &   &\ddots&\ddots &\ddots&\ddots &\ddots&\\  
&&   &   &1&-4 &6 & -4&1\\       
&&   &   && 1&-4 & 6&-4\\  
&&   &          &    &&1 &  -4 &5
\end{pmatrix}_{(N-1)\times(N-1)}\,,
\]
where the empty entries are zeros.  

The second, alternative option, is to pre-assume zero values for the antisymmetric conditions \eqref{D1}-\eqref{D2}. 
This option naturally corresponds to the Dirichlet conditions of the first kind. It is graphically illustrated in the cartoon of Fig. 2.
\begin{figure}[tbh!]
	\label{fig:D2}
	\includegraphics{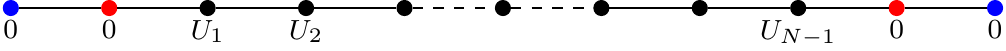}	
	\caption{Implementation of the Dirichlet boundary conditions of the second kind when the antisymmetric conditions \eqref{D1}-\eqref{D2} are locked to zeros, and actually, gives rise to the Dirichlet boundary conditions of the first kind.}
\end{figure}
In this second option, the operator $\Delta_d$ preserves its previous matrix format, while for $\Delta_d^2$  the matrix form becomes: 
\[
\Delta^2_ d=
\begin{pmatrix} 
6 & -4  &  1 & &  &   & &&\\ 
-4  &  6 &-4&1 &   &   &&&\\  
1  &  -4 &6&-4 &1   &   &&&\\  
&\ddots &\ddots  &\ddots&  \ddots&\ddots&& &\\  
&   &1 &  -4 &6&-4 &1    &&\\  
&   &   &\ddots&\ddots &\ddots&\ddots &\ddots&\\  
&&   &   &1&-4 &6 & -4&1\\       
&&   &   && 1&-4 & 6&-4\\  
&&   &          &    &&1 &  -4 &6
\end{pmatrix}_{(N-1)\times(N-1)}
\]
Note that the $(N-1)\times(N-1)$ matrix operators act on a vector $U=(0,U_1,U_2,\dots,U_{N-2},U_{N-1},0)\in\ell^2_0$. Furthermore, it is important to stress the following:  obviously the relations \eqref{lp6}-\eqref{lp7}, are valid in the case of $\Delta_d:\ell^2_{0}\rightarrow\ell^2$, as well as the continuity property \eqref{lp8}, for both of the above options (see also Lemma \ref{LA1} in Appendix \ref{App}). Additionally,  the discrete Poincar\'{e} inequality \cite{DKG2019}, holds:
\begin{eqnarray}
\label{crucequiv}
\mu_1\sum_{n=0}^{K+1}|U_n|^2\leq
\frac{1}{h^2}(-\Delta_dU,U)_{\ell^2}\leq \frac{4}{h^2} \sum_{n=0}^{K+1}|U_n|^2,\;\; \mu_1=\frac{4}{h^2}\sin^2\left(\frac{\pi h}{2l}\right)=\frac{4}{h^2}\sin^2\left(\frac{\pi h}{4L}\right),
\end{eqnarray}
where $l=2L$ the length of the symmetric interval $[-L,L]$.
\paragraph{Periodic boundary conditions.} For the continuous model \eqref{eq5}, the periodic boundary conditions should be imposed for the function $U(x,t)$ and its derivatives up to the third order, that is, $\partial^j_xU(x,-L)=\partial^j_xU(x, L)$, for $j=1,2,3$.  
In analogy to the continuous case the discrete periodic boundary conditions take the form:
\begin{eqnarray}
\label{sysper}
U_0&=&U_{N},\nonumber\\
U_{-1}-U_1&=&U_{N-1}-U_{N+1},\nonumber\\ 
U_{-1}-2U_0+U_{1}&=&U_{N-1}-2U_{N}+U_{N+1},\\
U_{2}-2U_{1}+2 U_{-1}-U_{-2}& =&U_{N+2}- 2U_{N+1}+2U_{N-1}-U_{N-2}.\nonumber
\end{eqnarray}
For the derivation of the system \eqref{sysper},  we have used the central difference approximation for derivatives up to the third order. From the system \eqref{sysper}, we deduce that $U_{i}=U_{N+i}$ for $i=-2,-1,0,1,2$, and eventually that $U\in\ell^2_{\mathrm{per}}$, the space of periodic sequences of period $N$ defined in Appendix \ref{App}.

Relations \eqref{lp6}-\eqref{genlp8} are also valid in the case of $\Delta_d:\ell^2_{\mathrm{per}}\rightarrow\ell^2$, see \cite{CNS2019} and Lemma \ref{LA1} in Appendix \ref{App}. It should be also remarked that the inequality \eqref{genlp9}
is also valid, see Lemma \ref{LA2} in Appendix \ref{App}. 
\subsection{Local existence and properties of the solution operators}
\paragraph{Local existence of solutions for $\ell^p$-type initial data, $1\leq p\leq\infty$.} For local existence of solutions we may apply the generalized Picard-Lindel\"{o}f Theorem \cite[Theorem 3.A, pg. 78]{zei85a}. We need to show that the nonlinear operator 
\begin{eqnarray}
\label{eq9}
\left\{\mathcal{G}[U]\right\}_{n\in\mathbb{Z}}:=\mathcal{G}[U_n]=-\frac{\gamma_1}{h^4}U_n\Delta_d^2U_n- \frac{\gamma_2}{h^2}U_n\Delta_dU_n+\frac{\gamma_3}{h^2}\Delta_dU_n+f(U_n),
\end{eqnarray}
is locally Lipschitz, as stated in the following lemma.
\begin{lemma}
	\label{L1}
Let $Z$ be either, $\ell^p$, $\ell^p_{\mathrm{per}}$ or $\ell^p_0$, for $1\leq p\leq\infty$. The operator $\mathcal{G}:Z\to\ell^p$ is bounded on bounded sets of $Z$, and locally Lipschitz continuous.
\end{lemma}
\begin{proof}
For brevity, we shall only present the proof for the higher order term $\left\{\mathcal{L}_1[U]\right\}_{n\in\mathbb{Z}}=U_n\Delta_d^2U_n$ on $Z=\ell^p$,  for some $p\in [1,\infty)$, since for the other terms of $\mathcal{L}$ and the other cases of $Z$, the arguments are almost identical. Let $U\in \bar{B}(0,R)\subset\ell^p$ a closed ball of $\ell^p$ centered at $0$ and of radius $R$. Then, 
\begin{eqnarray}
||\mathcal{L}_1[U]||^p_{\ell^p}=\sum_{n\in\mathbb{Z}^N}|U_n|^p\,|\Delta_d^2U_n|^p\leq ||U||_{\ell^{\infty}}^p||\Delta_d^2U||_{\ell_p}^p\leq C||U||^{p+1}_{\ell^p}\leq CR^{p+1},
\end{eqnarray}
fore some generic constant $C>0$. Note that we have used the inclusion \eqref{lp7}, for $\ell^p\subset\ell^{\infty}$ and inequality \eqref{genlp9}. For the Lipschitz continuity, for $U,W\in  \bar{B}(0,R)$, 
\begin{eqnarray}
\label{nonlipl2}
||\mathcal{L}_1[U]-\mathcal{L}_1[W]||^p_{\ell^p}&=&\sum_{n\in\mathbb{Z}}|U_n\Delta^2_dU_n-W_n\Delta^2_dW_n+U_n\Delta^2_dW_n-U_n\Delta^2_dW_n|^p\nonumber\\
&\leq& C||U||_{\ell^{\infty}}^p||\Delta_d^2(U-W)||_{\ell^p}^p+C||\Delta_d^2W||_{\ell^{\infty}}^p||U-W||_{\ell^p}^p\nonumber \\
&\leq&
C||U||_{\ell^p}^p||\Delta_d^2(U-W)||_{\ell^p}^p+C||\Delta_d^2W||_{\ell^p}^p||U-W||_{\ell^p}^p\nonumber\\
&\leq&
C||U||_{\ell^p}^p||(U-W)||_{\ell^p}^p+C||W||_{\ell^p}^p||(U-W)||_{\ell^p}^p\leq CR^p||U-W||_{\ell^p}^p,
\end{eqnarray}
for which, we have used again the inequality  \eqref{genlp9}, together with the inclusion $\ell^p\subset\ell^{\infty}$.
\end{proof}

We proceed by rewriting the system \eqref{eq1} as 
\begin{eqnarray*}
\label{semifl1}
\dot{U}=\mathcal{G}[U],
\end{eqnarray*}
and we will implement the aforementioned  Picard-Lindel\"{o}f Theorem on the integral formula
\begin{eqnarray}
\label{mildsl1}
U(t)=U^0+\int_{0}^{t}\mathcal{G}[U(s)]ds.
\end{eqnarray}
Then, as a consequence of Lemma \ref{L1}, we have the following local existence  result.
\begin{theorem} 
\label{thloc}	
Assume that $\gamma_i>0$, $i=1,2,3$ and $\alpha,\beta\in\mathbb{R}$, and let $U^0\in Z$, arbitrary. There exists some $T^*(U^0)>0$ such that the initial value problem \eqref{eq1}-\eqref{inc}, has a unique solution $U\in C^1([0,T],Z)$  for all $0<T<T^*(U^0)$. In addition, the following alternatives hold: Either $T^*(U^0)=\infty$ (global existence) or $T^*(U^0)<\infty$ and $\lim_{t\uparrow T^*(U^0)}||U(t)||_{\ell^2}=\infty$ (collapse). Furthermore the solution $U$ depends continuously on the initial condition $U^0\in Z$, with respect to the norm of $C([0,T],Z)$.
\end{theorem}

For all $U^0\in\ell^p$ and $t\in [0,T^*(U^0))$, we may define the map
\begin{eqnarray}
\label{wds1}
&&\phi_t:\ell^p\rightarrow\ell^p,\;\;1\leq p\leq\infty,\nonumber\\
&&\;\;\;\;\;\;U^0\rightarrow \phi_t(U^0)=U(t).
\end{eqnarray}
Due to Theorem \ref{thloc}, we have that $\phi_t(U^0)\in  C^1([0,T^*(U^0)),\ell^p)$. 
\paragraph{Strong continuity of the semiflow in $\ell^1$.} 
When $p=1$, the semiflow $\phi_t$ defined in \eqref{wds1} possesses the following remarkable property.
\begin{theorem}
	\label{rth}
	Let as assume a sequence $U^{m,0}$, $m\in\mathbb{N}$, converging weakly to the initial condition $U^0$ in $\ell^1$, that is
	\begin{eqnarray}
	\label{weak1}
	U^{m,0}\rightharpoonup U^0\;\;\mbox{in $\ell^1$},\;\;\mbox{as $m\rightarrow\infty$}. 
	\end{eqnarray}
Then if $T<T^*(U^0)$, we have the strong convergence $\phi_t(U^{m,0})\rightarrow \phi_t(U^0)$ in 	 $C([0,T],\ell^1)$.
\end{theorem}
\begin{proof}
From the Theorem of Banach and Steinhaus  \cite[Proposition 21.23(b), pg. 258]{zei85b}, 
\begin{eqnarray}
\label{BST}
\mbox{$U^{m,0}$ is bounded and}\;\;||U^0||\leq\liminf_{m\rightarrow\infty}||U^{m,0}||_{\ell^1}.
\end{eqnarray}
Thus, from \eqref{BST}, if $||U^{0,m}||\leq M$, then $||U^0||_{\ell^1}\leq M$, and the solutions $\phi_t(U^{m,0})=U^m(t)$ and $\phi_t(U^0)=U(t)$, are well defined for all $t\in [0, T^*)$, where $T^*$ is suitably constructed from Theorem \ref{thloc}.  Now we may consider an arbitrary interval $[0,T]$, with $T<T^*$.  By the integral formula \eqref{mildsl1} and Lemma \ref{L1},  we have that
\begin{eqnarray*}
||\phi_t(U^{m,0})-\phi_t(U^0)||_{\ell^1}&\leq& ||U^{m,0}-U^0||_{\ell^1}+\int_0^t||\mathcal{G}[\phi_s(U^{m,0})]-\mathcal{G}[\phi_s(U^0)]||_{\ell^1}ds\\
&\leq&||U^{m,0}-U^0||_{\ell^1}+C(M)\int_0^t||\phi_s(U^{m,0})-\phi_s(U^0)||_{\ell^1}ds,
\end{eqnarray*}	
for some $C(M)>0$, and all $t\in [0,T]$. Then, applying Gronwall's lemma in the above inequality, we get that
\begin{eqnarray}
\label{Gron}
||\phi_t(U^{m,0})-\phi_t(U^0)||_{\ell^1}\leq ||U^{m,0}-U^0||_{\ell^1}\mathrm{e}^{TC(L)},\;\;\forall t\in [0, T].
\end{eqnarray}
This is the point where an important and non-trivial feature of $\ell^1$, namely the {\em Schur property} comes into play: weak and norm sequential convergence in $\ell^1$ coincide \cite[Definition 2.3.4 \& Theorem 2.3.6, pg. 32]{Kalton}. Therefore, the weak convergence \eqref{weak1}, implies the norm (strong) convergence $\lim_{m\rightarrow\infty}||U^{m,0}-U^0||_{\ell^1}=0$. Passing to the limit as $m\rightarrow\infty$ in \eqref{Gron}, we find that $\lim_{m\rightarrow\infty}||\phi_t(U^{m,0})-\phi_t(U^0)||_{\ell^1}=0$ for all $t\in [0,T]$ and the claim is proved. 
\end{proof}	

Using the inclusion relation $\ell^1\subset\ell^2$ and the Schur property, we have the following corollary.
\begin{corollary}
\label{CorLim1}
The restriction $\hat{\phi}_t:\ell^1\rightarrow \ell^2$ of 	the semiflow $\phi_t:\ell^2\rightarrow\ell^2$ is strongly continuous in the sense of Theorem \ref{rth}.
\end{corollary}
It should be remarked that the reflexive spaces $\ell^p$ for $p>1$ do not have the Schur property, and that if a reflexive Banach space has it, is finite dimensional \cite[Corollary 2.3.8, pg. 37]{Kalton}. In the next section, we will show that the above properties have some interesting applications when the solutions $\hat{\phi}_t(U^0)$, for $U^0\in\ell^1$ are uniformly bounded.	
\section{Global existence regimes: Extinction conditions and uniform bounds}
In this section, we identify parametric regimes for the parameters $\alpha, \beta\in\mathbb{R}$, associated with extinction (in the sense $\lim_{t\rightarrow\infty}||U(t)||_{\ell^p}=0$, $1<p\leq\infty$) or with uniform bounds for which the solution may not essentially vanish ($||U(t)||_{\ell^p}<M$, for some constant $M$, for $1<p\leq\infty$). When necessary, we shall distinguish between the infinite and the finite dimensional system. 

\paragraph{Stability of the zero solution-extinction.} We start with  the following extinction results in the case where the system is supplemented with the vanishing initial conditions \eqref{vanv}. 
\begin{proposition}
	\label{PInf1}
Consider the system \eqref{eq1} supplemented with either case of boundary conditions (vanishing \eqref{vanv}, periodic 
or Dirichlet), 
 and let $U^0\in\ell^2$, an arbitrary initial condition. We assume that $\alpha=-\tilde{\alpha}<0$ and $\beta=-\tilde{\beta}<0$. Then, there exists 
$-\tilde{\beta}_{\mathrm{thresh}}(\gamma_1, \gamma_2,h)<0$ (depending only on $\gamma_1$, $\gamma_2$, $h$), such that, if $-\tilde{\beta}<-\tilde{\beta}_{\mathrm{thresh}}$, then  $\lim_{t\rightarrow\infty}||U(t)||_{\ell^p}=0$, for all $2<p\leq\infty$. 
\end{proposition}
\begin{proof}
We multiply Eq.~\eqref{eq1} in the $\ell^2$-inner product, to get the balance equation:
\begin{eqnarray}
\label{eq23}
\frac{1}{2}\frac{d}{dt}||U||^2_{\ell^2}-\frac{\gamma_3}{h^2}\left<\Delta_dU, U\right>_{\ell^2}+||U||_{\ell^4}^4 +\tilde{\alpha}||U||^2_{\ell^2}+\tilde{\beta}\left<U^2, U\right>_{\ell^2}=-\frac{\gamma_1}{h^4}\left<U\Delta^2_dU, U\right>_{\ell^2}-\frac{\gamma_2}{h^2}\left<U\Delta_dU, U\right>_{\ell^2}.
\end{eqnarray}	
The terms of the right-hand side can be estimated as follows:
\begin{eqnarray}
\label{eq11}
\left|\left<U\Delta^2_dU, U\right>_{\ell^2}\right|\leq||U^2||_{\ell^2}||\Delta^2_dU||_{\ell^2}. 
\end{eqnarray}
Note, that due to the embedding relation \eqref{eq7}, $||U^2||_{\ell^2}=||U||_{\ell^4}^2\leq||U||^2_{\ell^2}$. Then, by using the inequality \eqref{lp9}, we get that
\begin{eqnarray}
\label{eq13}
\left|\left<U\Delta^2_dU, U\right>_{\ell^2}\right|\leq 4||U||^3_{\ell^2}.
\end{eqnarray}
The same argument, if applied to the term with $U\Delta_dU$, implies that
\begin{eqnarray}
\label{eq14}
\left|\left<U\Delta_dU, U\right>_{\ell^2}\right|\leq||U^2||_{\ell^2}||\Delta_dU||_{\ell^2}\leq 2||U||^3_{\ell^2}.
\end{eqnarray}
Then, by using equation \eqref{lp6}, and by inserting the estimates \eqref{eq13} and \eqref{eq14} into \eqref{eq23}, we arrive at the differential inequality
\begin{eqnarray}
\label{eq24}
\frac{1}{2}\frac{d}{dt}||U||^2_{\ell^2}+\frac{\gamma_3}{h^2}\sum_{n\in\mathbb{Z}}|U_{n+1}-U_n|^2+||U||_{\ell^4}^4+\tilde{\alpha}||U||^2_{\ell^2}+(\tilde{\beta}-\tilde{\beta}_{\mathrm{thresh}})||U||_{\ell^2}^3<0,
\end{eqnarray}
for some $\tilde{\beta}_{\mathrm{thresh}}(\gamma_1, \gamma_2,h)>0$.  Hence, assuming that $\tilde{\beta}>\tilde{\beta}_{\mathrm{thresh}}$ it follows that 
\begin{eqnarray*}
||U(t)||^2_{\ell^2}\leq \mathrm{e}^{-2\tilde{\alpha}t}||U_0||^2_{\ell^2},
\end{eqnarray*}	
implying that $\lim_{t\rightarrow\infty}||U(t)||_{\ell^2}^2=0$. Since from \eqref{lp7}, $||U||_{\ell^p}\leq ||U||_{\ell^2}$, for all $2<p\leq\infty$, the latter implies the decay in all the relevant $\ell^p$-norms.
\end{proof}

A second option concerns extinction in the infinite lattice when $\beta>0$ and $\alpha$ being negative and sufficiently small.

\begin{proposition}
	\label{PInf2}
Consider the system \eqref{eq1} supplemented with either case of boundary conditions (vanishing \eqref{vanv}, periodic 
or Dirichlet 
), and let $U^0\in\ell^2$, an arbitrary initial condition. We assume that $\alpha=-\tilde{\alpha}<0$ and $\beta>0$. Then, there exists $-\tilde{\alpha}_{\mathrm{thresh}}(\gamma_1, \gamma_2,\beta,h)<0$
	such that, if $-\tilde{\alpha}<-\tilde{\alpha}_{\mathrm{thresh}}$, then  $\lim_{t\rightarrow\infty}||U(t)||_{\ell^p}=0$, for all $2<p\leq\infty$. 
\end{proposition}
\begin{proof}
For $\beta>0$ and $\alpha=-\tilde{\alpha}<0$,	we are viewing the balance equation \eqref{eq23} as
	\begin{eqnarray}
	\label{eq23n}
	\frac{1}{2}\frac{d}{dt}||U||^2_{\ell^2}-\frac{\gamma_3}{h^2}\left<\Delta_dU, U\right>_{\ell^2}+||U||_{\ell^4}^4 +\tilde{\alpha}||U||^2_{\ell^2}=-\frac{\gamma_1}{h^4}\left<U\Delta^2_dU, U\right>_{\ell^2}-\frac{\gamma_2}{h^2}\left<U\Delta_dU, U\right>_{\ell^2}+\beta\left<U^2, U\right>_{\ell^2},
	\end{eqnarray}	
and we manipulate the three terms of its right-hand side alternatively to Proposition \ref{PInf1}, as follows:
	\begin{eqnarray}
	\label{eq11n}
	\frac{\gamma_1}{h^4}\left|\left<U\Delta^2_dU, U\right>_{\ell^2}\right|&\leq&\frac{\gamma_1}{h^4}||U^2||_{\ell^2}||\Delta^2_dU||_{\ell^2}=\frac{\gamma_1}{h^4}||U||_{\ell^4}^2||\Delta^2_dU||_{\ell^2}\leq \frac{4\gamma_1}{h^4}||U||_{\ell^4}^2||U||_{\ell^2}\nonumber\\
	&\leq& \frac{1}{\epsilon^2}||U||_{\ell^4}^4+c_1(\gamma_1,  h,\epsilon^{-2})||U||_{\ell^2}^2,\\
	\label{eq12n}
	\frac{\gamma_2}{h^2}\left|\left<U\Delta_dU, U\right>_{\ell^2}\right|&\leq&\frac{\gamma_2}{h^2}||U^2||_{\ell^2}||\Delta_dU||_{\ell^2}=\frac{\gamma_2}{h^2}||U||_{\ell^4}^2||\Delta_dU||_{\ell^2}\leq \frac{2\gamma_2}{h^2}||U||_{\ell^4}^2||U||_{\ell^2}\nonumber\\
	&\leq& \frac{1}{\epsilon^2}||U||_{\ell^4}^4+c_2(\gamma_2,  h,\epsilon^{-2})||U||_{\ell^2}^2,\\
\label{eq13n}
\beta\left|\left<U^2, U\right>_{\ell^2}\right|&\leq&\beta||U^2||_{\ell^2}||U||_{\ell^2}=\beta||U||_{\ell^4}^2||U||_{\ell^2}\nonumber\\
&\leq& \frac{1}{\epsilon^2}||U||_{\ell^4}^4+c_3(\beta,h,\epsilon^{-2})||U||_{\ell^2}^2.
\end{eqnarray}	
For the estimates \eqref{eq11n}-\eqref{eq13n},  we have used Young's inequality $ab\leq \epsilon^pa^p+ b^q/\epsilon^q$, for $a,b\geq 0$ when $1/p+1/q=1$, with the choices $p=q=2$ and some suitable fixed $\epsilon>0$. This time, by inserting the estimates  \eqref{eq11n}-\eqref{eq13n} and \eqref{eq14} into \eqref{eq23n}, we derive the differential inequality
	\begin{eqnarray}
	\label{eq24n}
	\frac{1}{2}\frac{d}{dt}||U||^2_{\ell^2}+\frac{\gamma_3}{h^2}\sum_{n\in\mathbb{Z}}|U_{n+1}-U_n|^2+\delta(\epsilon)||U||_{\ell^4}^4+(\tilde{\alpha}-\tilde{\alpha}_{\mathrm{thresh}})||U||^2_{\ell^2}<0,\;\;\tilde{\alpha}_{\mathrm{thresh}}=c_1+c_2+c_3,\;\;\delta(\epsilon)\geq 0.
	\end{eqnarray}
	Therefore, assuming that $\tilde{\alpha}>\tilde{\alpha}_{\mathrm{thresh}}$ it follows that 
	\begin{eqnarray*}
		||U(t)||^2_{\ell^2}\leq \mathrm{e}^{-2(\tilde{\alpha}-\tilde{\alpha}_{\mathrm{thresh}})t}||U_0||^2_{\ell^2},
	\end{eqnarray*}	
	implying again that $\lim_{t\rightarrow\infty}||U(t)||_{\ell^2}^2=0$.
\end{proof}

In the case of the Dirichlet boundary conditions, we have the following alternatives for the stability of the zero solution, due to the validity of the discrete Poincar\'{e} inequality \eqref{crucequiv}.
\begin{proposition}
	\label{PInf3}
	Consider the system \eqref{eq1} supplemented with Dirichlet boundary conditions, 
and let $U^0\in\ell^2_0$, an arbitrary initial condition. We assume that $\alpha>0$ and $\beta>0$. Then, there exists a constant $\gamma(\gamma_1,\gamma_2,\beta,h)>0$ and
	$$\gamma_{3,\mathrm{thresh}}=\frac{1}{\mu_1}\left(\alpha+\gamma\right)>0,$$ such that, if $\gamma_3>\gamma_{3,\mathrm{thresh}}$, then  $\lim_{t\rightarrow\infty}||U(t)||_{\ell^p}=0$, for all $2<p\leq\infty$. 
\end{proposition}
\begin{proof}
	For $\beta>0$ and $\alpha>0$, by using the discrete Poincar\'{e} inequality \eqref{crucequiv}, we estimate the second term of \eqref{eq23} from below:
	\begin{eqnarray*}
	\frac{1}{2}\frac{d}{dt}||U||^2_{\ell^2}+\gamma_3\mu_1||U||_{\ell^2}^2+||U||_{\ell^4}^4 -\alpha||U||^2_{\ell^2}=-\frac{\gamma_1}{h^4}\left<U\Delta^2_dU, U\right>_{\ell^2}-\frac{\gamma_2}{h^2}\left<U\Delta_dU, U\right>_{\ell^2}+\beta\left<U^2, U\right>_{\ell^2},
	\end{eqnarray*}	
Manipulating the three terms of the right-hand side as in Proposition \ref{PInf2}, we get an inequality of the form
	\begin{eqnarray}
	\label{eq24nn}
	\frac{1}{2}\frac{d}{dt}||U||^2_{\ell^2}+\left(\gamma_3\mu_1-\gamma\right)||U||_{\ell^2}^2+\frac{1}{2}||U||_{\ell^4}^4<0,
	\end{eqnarray}
With the assumption  $\gamma_3>\gamma_{3,\mathrm{thresh}}$, we may set $\Gamma= \gamma_3\mu_1-\gamma>0$. Consequently, it follows that 
	\begin{eqnarray*}
		||U(t)||^2_{\ell^2}\leq \mathrm{e}^{-2\Gamma t}||U_0||^2_{\ell^2},
	\end{eqnarray*}	
from which, we again deduce the global stability of the zero-solution.
\end{proof}

Combining the estimation procedures of Proposition \ref{PInf1} and \ref{PInf3}, we may identify an extinction regime when $\alpha>0$ and $\beta<0$.
\begin{proposition}
	\label{PInf4}
	Consider the system \eqref{eq1} supplemented with Dirichlet boundary conditions, 
and let $U^0\in\ell^2_0$, an arbitrary initial condition. We assume that $\alpha>0$ and $\beta=-\tilde{\beta}<0$. Then, there exists
\begin{eqnarray} 
\label{m1crit}
	\hat{\gamma}_{3,\mathrm{thresh}}=\frac{\alpha}{\mu_1}>0,
	\end{eqnarray} 
	such that, if $\gamma_3>\hat{\gamma}_{3,\mathrm{thresh}}$, and  $-\tilde{\beta}<-\tilde{\beta}_{\mathrm{thresh}}$, where $-\tilde{\beta}_{\mathrm{thresh}}$ is defined in Proposition \ref{PInf1}, then $\lim_{t\rightarrow\infty}||U(t)||_{\ell^p}=0$, for all $2<p\leq\infty$. 
	\end{proposition}
\begin{proof}
In this case, the differential inequality 
\begin{eqnarray}
	\label{eq24nnn}
	\frac{1}{2}\frac{d}{dt}||U||^2_{\ell^2}+\left(\gamma_3\mu_1-\alpha\right)||U||_{\ell^2}^2+\frac{1}{2}||U||_{\ell^4}^4+(\tilde{\beta}-\tilde{\beta}_{\mathrm{thresh}})||U||_{\ell^2}^3<0,
	\end{eqnarray}
is the one which implies the decaying estimate 	
	\begin{eqnarray*}
		||U(t)||^2_{\ell^2}\leq \mathrm{e}^{-2\Gamma_1 t}||U_0||^2_{\ell^2},\;\;\Gamma_1=\gamma_3\mu_1-\alpha>0,
	\end{eqnarray*}	
and the global stability of the zero-solution.
\end{proof}
\paragraph{Uniform in time estimates.}  In the case of the infinite lattice, we may identify the following parametric regimes for uniform in time, but not decaying bounds. 
\begin{proposition}
	\label{PInf5}
	Consider the system \eqref{eq1} supplemented with either case of boundary conditions (vanishing \eqref{vanv}, periodic  or Dirichlet), and let $U^0\in\ell^2$, an arbitrary initial condition. We assume that $\alpha>0$ and $-\tilde{\beta}<-\tilde{\beta}_{\mathrm{thresh}}$, where $-\tilde{\beta}_{\mathrm{thresh}}$ is defined in Proposition \ref{PInf1}. Then, there exists $R_0>0$, such that 
  $\limsup_{t\rightarrow\infty}||U(t)||_{\ell^p}\leq R_0$, for all $2<p\leq\infty$. 
	\end{proposition}
\begin{proof}
When $\alpha>0$ and $-\tilde{\beta}<-\tilde{\beta}_{\mathrm{thresh}}$, working for the nonlinear terms as in Proposition \ref{PInf1}, we arrive in the following cubic differential inequality:
\begin{eqnarray}
\label{eq26nn}
\frac{d}{dt}||U||^2_{\ell^2}-2\alpha||U||^2_{\ell^2}+2(\tilde{\beta}-\tilde{\beta}_{\mathrm{thresh}})||U||_{\ell^2}^3<0.
\end{eqnarray}
We set $||U||^2_{\ell^2}=\chi(t)$, and $B=\tilde{\beta}-\tilde{\beta}_{\mathrm{thresh}}$, and the inequality \eqref{eq26nn}, can be rewritten as
\begin{eqnarray}
\label{eq27nn}
\dot{\chi}\leq 2\alpha \chi-2B\chi^{\frac{3}{2}}.
\end{eqnarray}
With the change of variables $u=\chi^{-\frac{1}{2}}$, we have $\dot{\chi}=-2\chi^{\frac{3}{2}}\dot{u}$. Then, since $\chi(t)\geq 0$, the inequality \eqref{eq27nn} for $\chi$ becomes the following inequality for $u$:
\begin{eqnarray}
\label{eq28nn}
\dot{u}\geq -\alpha u+B.
\end{eqnarray}
Multiplying \eqref{eq28nn} with the integrating factor $\mathrm{e}^{\alpha t}$, we get that $\frac{d}{dt}[u\mathrm{e}^{\alpha t}]\geq B\mathrm{e}^{\alpha t}>0$, for all $t\geq 0$. Integrating in the interval $[0,t]$ for arbitrary $t\geq 0$, we find that
\begin{eqnarray*}
u(t)\geq u(0)\mathrm{e}^{-\alpha t}+\frac{B}{\alpha}\left(1-\mathrm{e}^{-\alpha t}\right),
\end{eqnarray*}	
which if written in terms of $\chi(t)$, implies that
\begin{eqnarray*}
||U(t)||_{\ell^2}\leq \left[\frac{\mathrm{e}^{-\alpha t}}{||U_0||_{\ell^2}}+\frac{B}{\alpha}\left(1-\mathrm{e}^{-\alpha t}\right)\right]^{-1}.
\end{eqnarray*}	
Passing to the limit as $t\rightarrow\infty$, we get that
$\limsup_{t\rightarrow\infty}||U(t)||_{\ell^2}=\alpha/B=R_0$.
\end{proof}

In the case of the finite lattice, we may have  unconditional, with respect to the sign of parameters $\alpha$ and $\beta$, uniform boundedness of solutions.
\begin{proposition}
	\label{PInf6}
	Consider the system \eqref{eq1} supplemented with either periodic  or Dirichlet) boundary conditions, and let $U^0\in\ell^2$, an arbitrary initial condition. Let $\alpha,\beta\in\mathbb{R}$. There exists $R>0$, such that 
	$\limsup_{t\rightarrow\infty}||U(t)||_{\ell^p}\leq R$, for all $2<p\leq\infty$. 
	\end{proposition}
\begin{proof}
	When $\alpha,\beta\in\mathbb{R}$, we may work to handle the nonlinear terms as in Proposition \ref{PInf2}, and derive the quartic differential inequality:
	\begin{eqnarray}
	\label{eq30nn}
	\frac{d}{dt}||U||^2_{\ell^2}+||U||_{\ell^4}^4\leq 2(|\alpha|+\tilde{\alpha}_{\mathrm{thresh}})||U||^2_{\ell^2}.
	\end{eqnarray}
Moreover, due to the equivalence of norms \eqref{eq22}, for $p=2$ and $q=4$, we have that 
\begin{eqnarray}
\label{eq30nnn}
\frac{1}{N}||U||_{\ell^2}^4\leq ||U||_{\ell^4}^4,
\end{eqnarray}
and \eqref{eq30nn} becomes
\begin{eqnarray}
\label{eq31nnn}
\frac{d}{dt}||U||^2_{\ell^2}\leq 2(|\alpha|+\tilde{\alpha}_{\mathrm{thresh}})||U||^2_{\ell^2}-\frac{1}{N}||U||_{\ell^2_{\mathrm{per}}}^4.
\end{eqnarray}
Now, we set $||U||^2_{\ell^2}=\chi(t)$, $A=2(|\alpha|+\tilde{\alpha}_{\mathrm{thresh}})$ and $B=1/N$, to rewrite the inequality \eqref{eq31nnn} for $\chi$:
	\begin{eqnarray}
	\label{eq31nn}
	\dot{\chi}\leq A\chi-B\chi^2.
	\end{eqnarray}
In the case of \eqref{eq31nnn},	the change of variables $u=\chi^{-1}$, transforms it to the inequality for $u$:
	\begin{eqnarray*}
	\label{eq32nn}
	\dot{u}\geq -A u+B,
	\end{eqnarray*}
which can be integrated as in Proposition \ref{PInf5} to deduce the estimate	
	\begin{eqnarray*}
		||U(t)||_{\ell^2}\leq \left[\frac{\mathrm{e}^{-A t}}{||U_0||^2_{\ell^2}}+\frac{B}{A}\left(1-\mathrm{e}^{-A t}\right)\right]^{-\frac{1}{2}}.
	\end{eqnarray*}	
Then, passing to the limit as $t\rightarrow\infty$, we find that
$\limsup_{t\rightarrow\infty}||U(t)||_{\ell^2}=\sqrt{A/B}=R$.
\end{proof}
\begin{remark}
\label{DtC}	
The derivation of the decay or uniform in time estimates of Propositions \ref{PInf1}-\ref{PInf6} is strictly depending on the properties of the discrete ambient spaces, particularly on the inclusions \eqref{eq7}, and the continuity properties of the discrete operators \eqref{genlp8} and \eqref{genlp9}, which allow for the handling of the highly nonlinear terms. It is not obvious that these estimates remain valid in the case of the continuous limit (the LL-PDE), since the estimation procedure gives rise to a rapid proliferation of higher order terms; for the latter is unclear if they can be estimated in a similar manner as in the discrete counterpart, due to the restrictions posed by the Sobolev embeddings, even in the $1D$-spatial domain.
\end{remark}
\paragraph{Comments on the nontrivial convergence dynamics.}
In the case of the finite dimensional dynamical system (associated to the case of Dirichlet or periodic boundary conditions), all the above propositions imply the existence of an absorbing set $\mathcal{B}$ in $Z=\ell^2_0,\ell^2_{\mathrm{per}}$, for the semiflow $\phi_t:Z\rightarrow Z$,  and the existence of a global attractor $\mathcal{Y}=\omega(\mathcal{B})$ in $Z$, where $\omega(\mathcal{B})$ denotes the $\omega$-limit set of $\mathcal{B}$, \cite{RT,cazh}. Particularly, under the assumptions of Propositions \ref{PInf1}-\ref{PInf4}, we have that  $\mathcal{Y}=\left\{0\right\}$, both for the finite and the infinite dimensional system.  On the other hand, under the assumptions of Propositions \ref{PInf5}-\ref{PInf6}, the attractor might be non-trivial, i.e.,  $\mathcal{Y}\neq \left\{0\right\}$. Furthermore, in the case of the infinite dimensional dynamical system, $\mathcal{Y}$ is a weak atrractor, that is compact in the weak-topology of $\ell^2$. 
However, although it is not possible to prove directly that the semiflow $\phi_t:\ell^2\rightarrow \ell^2$ is uniformly compact which would imply the compactness of the attractor in the norm-topology of $\ell^2$, for its restriction $\hat{\phi}_t:\ell^1\rightarrow\ell^2$, due to the Schur property, we have the following result.
\begin{proposition}
	\label{CORLIMA}
Assume that the parameters of the lattice $\eqref{eq1}$ satisfy the conditions of Propositions \ref{PInf5} or \ref{PInf6}. Let $\Psi^m$, $m\in\mathbb{N}$ be a convergent sequence in $\ell^1$, $\Psi^m\rightarrow \Psi$, as $m\rightarrow\infty$.   Then, there exists $T_{\mathrm{crit}}>0$ which is independent of $m$, such that for all $t\geq T_{\mathrm{crit}}$, $\hat{\phi}_{t}(\Psi^m)\rightarrow \hat{\phi}_t(\Psi)$ as $m\rightarrow\infty$, strongly in $\ell^2$, where $\hat{\phi}_t(\Psi)\in \mathcal{B}$,  and uniformly for all $t\in [T_{\mathrm{crit}}, T]$, for arbitrary $T<\infty$.	
\end{proposition} 	
\begin{proof} 
Let $\mathcal{O}$ a bounded set in $\ell^1$. Since $\ell^1\subset\ell^2$, the set $\mathcal{O}$ is also bounded in $\ell^2$. As $\mathcal{B}$ is an absorbing set in $\ell^2$, attracts all the bounded sets in $\ell^2$, and accordingly $\mathcal{O}$. That is, there exist a time of entry $T(\mathcal{O})>0$, such that $\hat{\phi}_t(\mathcal{O})\subset\mathcal{B}$, for all $t\geq T(\mathcal{O})$.  

Now, since $\Psi^m$ is converging to $\Psi$ in $\ell^1$, is bounded in $\ell^1$, and $\left\{\Psi^m\right\}_m\subset\mathcal{O}\subset\ell^1\subset\ell^2$ for some bounded set $\mathcal{O}$, which may depend on $\Psi^m$, but is independent of $m$ for all $m\in\mathbb{N}$. Then
\begin{eqnarray}
\label{LIM1}
\hat{\phi}_{t}(\Psi^m)\in \mathcal{B},\;\;\mbox{for $t>T(\mathcal{O}):=T_{\mathrm{crit}}$}.
\end{eqnarray} 
Therefore, since $\ell^2$ is reflexive, \eqref{LIM1} implies that $\left\{\hat{\phi}_{t}(\Psi^m)\right\}_m$ is weakly relatively compact in $\ell^2$, for all $t\geq T_{\mathrm{crit}}$. Using also that  $\mathcal{B}$ is closed and convex in $\ell^2$, we have that for some subsequence $m'$:
\begin{eqnarray}
\label{LIM2}
\hat{\phi}_{t}(\Psi^{m'})\rightharpoonup\Phi\in\mathcal{B},\;\;\mbox{weakly as $m'\rightarrow\infty$,\;\; for all $t\in  [T_{\mathrm{crit}}, T]$.}
\end{eqnarray}
Consider next, an arbitrary closed interval $S=[0,T]$, $T<\infty$. Since for any $U^0\in\ell^1$, the orbits $\hat{\phi}_t(U^0)$ are uniformly bounded, we may extend the strongly continuity results of Theorem \ref{rth} and Corollary \ref{CorLim1} for any arbirtrary $t\in S$ (i.e., $T^*(U^0)=\infty$): We have also  that $\Psi\in\ell^1\subset\ell^2$, and hence
\begin{eqnarray}
\label{LIM3}
\hat{\phi}_t(\Psi^m)\rightarrow\hat{\phi}_t(\Psi),\;\;\mbox{as $m\rightarrow\infty$, strongly in $\ell^2$, for any $t\in S$.}
\end{eqnarray}
We proceed by showing that the convergence \eqref{LIM3} holds uniformly for all $t\in S$. This will be a consequence of the fact that the sequence $\hat{\phi}_{t}(\Psi^{m})$ is equicontinuous.  For this purpose, setting $U^m(t)=\hat{\phi}_{t}(\Psi^{m})$, we consider the sequence
\begin{eqnarray}
\label{wseq}
w_m(t)=\left<U^m(t), \chi\right>_{\ell^2},\;\;\mbox{for all $t\in S$ and $\chi\in\ell^2$}.
\end{eqnarray}
Due to Lemma \ref{L1}, we may easily deduce from $\eqref{eq1}$ (solved by $U^m$ when written in the form $\dot{U}^m=\mathcal{G}[U^m]$), that $\dot{U}^m$ is also uniformly bounded for all $t\in S$. Then, by applying the mean value theorem:
\begin{eqnarray}
\label{wseq2}
|w_m(t_1)-w_m(t_2)|&=&\left|\left<U^m(t_1)-U^m(t_2),\chi\right>_{\ell^2}\right|=\left|\left<\dot{U}^m(\theta),\chi\right>_{\ell^2}\right|\;|t_1-t_2|\nonumber\\
&\leq&\sup_{\theta\in S}||\dot{U}(\theta)||_{\ell^2}||\chi||_{\ell^2}|t_1-t_2|\nonumber\\
&\leq& c|t_1-t_2|,\;\;\;\;\;\;\mbox{for some $\theta\in (t_1,t_2)\subset S$},
\end{eqnarray}
where $c$ is a constant independent of $m$. Since \eqref{wseq2} holds for any $\chi\in \ell^2$, we may apply \eqref{wseq2} for the uniformly bounded $\chi=U^m(t_1)-U^m(t_2)$, to derive:
\begin{eqnarray*}
||\hat{\phi}_{t_1}(\Psi^{m})-\hat{\phi}_{t_2}(\Psi^{m})||_{\ell^2}\leq C|t_1-t_2|,
\end{eqnarray*}
where $C$ is yet independent of $m$.  Hence, the above inequality implies the equicontinuity of $\hat{\phi}_{t}(\Psi^{m})$.  Accordingly, the uniform convergence \eqref{LIM3}, for all $t\in S$, follows from the Arzel\`{a}-Ascoli Theorem.

On the one hand, strong convergence in $\ell^2$ implies weak convergence in $\ell^2$, while the limiting relation \eqref{LIM3} holds for the whole sequence $\hat{\phi}_t(\Psi^m)$, and thus, for all of each subsequences. Therefore, combining the limiting relations \eqref{LIM2} and \eqref{LIM3}, we conclude that $\Phi=\hat{\phi}_t(\Psi)$ and that the convergence 
$\hat{\phi}_{t}(\Psi^{m})\rightarrow \hat{\phi}_t(\Psi)\in\mathcal{B}$, is strong as $m\rightarrow\infty$,  uniformly for all $t\in  [T_{\mathrm{crit}}, T]$, as claimed.
\end{proof}

In the specific case where the sequence $\Psi^m$ of $\ell^1$ converges to an equilibrium $U_s$ of the lattice $\eqref{eq1}$, and the equilibrium is such that $U_s\in\ell^1$, we have the following 
\begin{corollary}
	\label{CORLIMB}
Let the assumptions of Proposition \ref{CORLIMA}, be satisfied, with the additional hypothesis that $\Psi=U^s\in\ell^1$ is an equilibrium for the system $\eqref{eq1}$. Then $U^s\in\mathcal{B}$, and $\hat{\phi}_{t}(\Psi^{m})\rightarrow U^s\in\mathcal{B}$, strongly as $m\rightarrow\infty$, and uniformly for all $t\in [T_{\mathrm{crit}}, T]$, for arbitrary $T<\infty$.		
\end{corollary}
\begin{proof}
The result is a consequence of the fact that if $\Psi=U^s\in\ell^1$ is an equilibrium, then $\hat{\phi}_{t}(U^s)=U^s$, for all $t\geq 0$.	
\end{proof}

An interesting question not considered herein, could be the extension of the so-called ``tails estimates method'' \cite{tails1, tails2, tails3} in order to prove the asymptotic compactness of the flow $\phi_t:\ell^2\rightarrow\ell^2$ in the case of the infinite lattice.  Such an extension seems to be non-trivial due to the presence of the higher-order nonlinear terms (see also comments on \cite{DNLS2005,tails4}). It is also particularly relevant (due to the presence of the local nonlinearity $f$ and the discrete Laplacian-term) to investigate possible extensions of  \L{}ojasiewicz inequality-type arguments \cite{DKG2019, ConAH} in order to establish convergence to non-trivial equilibrium, at least for certain parametric regimes; for small $\gamma_1,\gamma_2$, one could view DLL as a perturbation of a gradient system.
 \section{Linear Stability and Numerical insights on the dynamics}
 In Propositions \ref{PInf5} and \ref{PInf6}, we identified parametric regimes where the  convergence dynamics of the system may be non-trivial. Although we have not derived a Lyapunov function, motivated by the results on the continuous model, we may assume that the solutions converge to non-trivial steady-states. With such a motivation, to investigate the convergence dynamics of the flow, we perform in this section a linear stability analysis for the simplest class of steady states, namely, the spatially homogeneous. The aim is to derive their instability criteria under which richer dynamics may emerge. 
 
 Prior to the above stability analysis, it is important to discuss further the physical meaning of the key parameters of the continuous LL model, as we expect that the instability criteria for the discrete one should involve these parameters, and importantly, the discretization parameter $h$. 
 
 The parameter $\alpha $ can be interpreted as a measure of the environment's aridity which characterizes the productivity of the system. To highlight its significance, we may consider two extreme cases: $\alpha<0$ close to $-1$ is associated to the non-survival
 chances for  a parched environment, while values of $\alpha$ much greater than zero  are relative to a humid - high productivity environment.  
 In subsection \ref{section:numerics}, an intermediate situation is investigated in a neighborhood of  $\alpha=0$. 
 This is a threshold value below which zero vegetation is asymptotically possible, since the environment is arid enough,
 while above that threshold vegetation certainly survives, at least for the uncoupled system.
 These roughly oriented parameter regimes are independent of the spatial interactions;
 in other words rely only on the source-nonlinearity $f$.  However,  the transition from one regime to the other, reveals a variety of spatiotemporal behaviors
 which are triggered by the nonlinear spatial coupling, given that $f$ satisfies specific conditions. This is already known for the continuous system and 
 similar conclusions can be made for the discrete case as the linear stability analysis indicates in the following section.  
 Also, we note that in $f$, the parameter $\beta$, controls the cooperation effect influencing the local reproduction. This effect is considered weak for 
 $\beta \leq 0$  and strong for $\beta>0$.  We consider the latter case where a bi-stability region of $\alpha$ is formed below zero.
 Regarding the parameters $\gamma_1$ and $\gamma_2$, in all simulations we use the values appearing in the continuous LL as provided in \cite{LL1997}, 
 namely, $\gamma_1=0.125$,  $\gamma_2=0.5$.  The parameter   $\gamma_3=\frac{1}{2}l^2$ will be varied,  recalling that $l$ represents the ratio of facilitative to inhibitory interaction ranges. 
 For example, a long-range competition short-range activation hypothesis requires that $l<1$.
%

In the presence of these effects, how the spatially discrete coupling can affect the long term dynamics? Summarizing the analytical arguments, we will further explore this question by numerical simulations.
 
 \subsection{Linear stability analysis}
Let $U_s$ denote a uniform steady-state, that is  $U_n=U_s$ for all $n\in \mathbb{Z} $, with $f(U_s)=0$.
Three distinct such uniform steady-states may  exist:
the trivial state $U^0_s=0$ and two non-trivial states
\begin{equation}
U_s^{\pm}=\frac{\beta\pm\sqrt{\beta^2+4\alpha}}{2}\,.
\end{equation}
Notice, that $U^0_s$ exists for all $\alpha,\;\beta$, while $U_s^{\pm}$ exist only when $\beta^2+4\alpha\geq0$.
In particular, 
\begin{itemize}
	\item[(i)] $\beta>0$ implies $U_s^{\pm}>0$ for $\alpha\in \left[-\left(\frac{\beta}{2}\right)^2,0\right)$  and $U_s^{+}>0>U_s^{-}$ for $\alpha\geq 0$. 
	\item[(ii)] $\beta\leq 0$ implies  $U_s^{+}>0>U_s^{-}$ when $\alpha>0$ and $0>U_s^{+}>U_s^{-}$ otherwise.
\end{itemize}
Here,  we are interested only on the ecologically realistic equilibria, which must be non-negative, i.e. only the cases where $U_s\geq 0$ are relevant. 
Consider the vector
\begin{equation}
\label{eq:pert}
U(t)= U_s +\hat{U}(t),
\end{equation}  
which is a perturbation of $U_s$ by $\hat{U}(t)=\{\hat{U}_n\}_{n\in \mathbb{Z}}$.
By substituting (\ref{eq:pert}) in equation (\ref{eq1}), and keeping only linear terms, we get the following linear system of coupled ODEs:
\begin{equation}\label{eq:line_syst}
\frac{d}{dt} \hat{U}_n = \mathcal{A} \hat{U} _n+f'(U_s)  \hat{U}_n.
\end{equation}
In the linearized equation \eqref{eq:line_syst}, the operator $\mathcal{A}$ is given by
$$\mathcal{A} =-U_s\frac{\gamma_1}{h^4}\Delta_d^2- (\gamma_2U_s-\gamma_3)\frac{1}{h^2}\Delta_d\,.$$
We proceed, by seeking solutions of the system (\ref{eq:line_syst}), possessing the form $\hat{U}_n(t)=\exp(\lambda t+i k h n)$.  We obtain for the parameter $\lambda$,  the $k$-dependent eigenvalues:
\begin{equation}
\label{eq:eigenvalues}
\lambda(k)=f'(U_s) - (\gamma_2U_s-\gamma_3) \frac{2(\cos(k h)-1)}{h^2}-\gamma_1U_s \frac{4(\cos(k h)-1)^2}{h^4}\,.
\end{equation}
\paragraph{Linear stability analysis of the trivial steady state.}
For the trivial steady-state $U^0_s=0$, we get:
\begin{equation}\label{eq:U0_eigen}
\lambda_k^0 = \alpha-\gamma_3\frac{4}{h^2}\sin^2\left(\frac{k h}{2}\right).
\end{equation}
Therefore, if $\alpha<0$, the  trivial steady-state $U^0_s=0$ is linearly stable for the local dynamics induced by $f$, and it remains stable in non-uniform perturbations, since $\lambda_k^0<0$ for all $k\in \mathbb{R}$.  On the other hand, when $\alpha>0$,  the trivial steady-state is unstable for the uncoupled system. However, large enough values of $\gamma_3$ may linearly stabilize this state when coupling is present. Particularly, in the case of the  Dirichlet boundary conditions this observation leads to the following inference.
\begin{proposition}
	\label{stabt}
	Consider the lattice \eqref{eq1} supplemented with the Dirichlet boundary conditions and $\alpha>0$. The steady-state $U_s^0=0$ is linearly stable if $\gamma_3>\dfrac{\alpha}{\mu_1}=\hat{\gamma}_{3,\mathrm{thresh}}$, where $\mu_1$ is the first eigenvalue of the discrete Laplacian.
\end{proposition}
\begin{proof}
Recall that  $N-1$ is the number of interior nodes of the finite interval $[-L, L]$ occupied by the lattice \eqref{eq1}. Setting
$k=\frac{j\pi}{hN}:=k^D_j$ in (\ref{eq:U0_eigen}), and using the $N-1$ eigenvalues of the discrete Laplacian operator  $\mu_j=\frac{4}{h^2} \sin^2\left(\frac{j\pi}{2N}\right)$, we get that 
\begin{equation}
\lambda_{k_j}^0=\alpha-\gamma_3 \mu_j,\; j\in\{1,2,\dots,N-1\}.
\end{equation}
Since $\mu_1<\mu_2<\cdots<\mu_{N-1}$, if
$\gamma_3>\frac{\alpha}{\mu_1}$, then  $\lambda_{k_j}^0<0$ for all $j=1,2,\dots,N-1$.
\end{proof}
\begin{remark}
\label{RemTS} 
It is interesting to recover in Proposition \ref{stabt} the threshold value $\hat{\gamma}_{3,\mathrm{thresh}}$ on the parameter $\gamma_3$ for the linear stability of $U^0_s=0$, as it was found in Proposition \ref{PInf4} for its global stability (i.e., for all initial data) in the case $\beta<0$. It is also important to highlight the physical relevance of Proposition \ref{PInf3}: While the linear stability of $U^0_s=0$ is guaranteed under the condition $\gamma_3>\hat{\gamma}_{3,\mathrm{thresh}}$ for $\beta<0$, for this state to become a global attractor in the case $\beta>0$, it is required at that $\gamma_3>\gamma_{3,\mathrm{thresh}}$-the threshold value derived in Proposition \ref{PInf3}.  Comparing the threshold values $\hat{\gamma}_{3,\mathrm{thresh}}$ and $\gamma_{3,\mathrm{thresh}}$, it is obvious that $\gamma_{3,\mathrm{thresh}}>\hat{\gamma}_{3,\mathrm{thresh}}$. Thus it is natural to assume larger values for $\gamma_3$ in order to achieve the global stability of  $U^0_s=0$, than its local (linear) stability.
\end{remark}
\paragraph{Linear stability analysis of the non-trivial steady-state.}
For the positive state $U_s=U_s^+>0$, it is easy to verify that it is stable in the absence of the non-linear next neighbor term (i.e. when $\gamma_2=0$),
independently of the values of the rest of the parameters.
However, excluding this case, there may exist bands of $k$ for which eigenvalues $\lambda_k$  become positive, and in turns,  $U_s^+$  will lose its stability.
In what follows, we will investigate how its  destabilization depends on the parameters of the lattice. 

We observe from \eqref{eq:eigenvalues}, that a  first crucial necessary condition for the potential instability of $U_s^+$ is 
\begin{equation}\label{Condition1}
C_1:=(\gamma_3-\gamma_2U_s)<0\,,
\end{equation}
and that a second necessary condition is
\begin{equation}\label{Condition2}
C_2:=(\gamma_2 U_s-\gamma_3)^2+4\gamma_1U_sf'(U_s)>0\,,
\end{equation}
since otherwise $\lambda(k)$  cannot change sign.
We continue under the hypothesis that (\ref{Condition1}) and (\ref{Condition2}) are valid;
clearly, it is required that $\gamma_2\neq 0$, and that $\gamma_3$ is sufficiently small. To further analyse the  behaviour of $\lambda(k)$ as a function of $k$, we consider its first derivative with respect to $k$:
\begin{equation}
\label{deig}
\frac{d}{dk} \lambda(k)=h\sin(k h)\left[U_s\frac{\gamma_1}{h^4} 8(\cos(k h)-1)+2\frac{ (\gamma_2U_s-\gamma_3)}{h^2} \right]\,.
\end{equation}
Zeroes of $\frac{d}{dk} \lambda(k)$ will provide the possible positions of local extreme-points of $\lambda$. 
We are particularly interested for the maxima, since we need to determine the critical parameter values 
for which the tops of the curve $\{(k,\lambda(k))|k\in \mathbb{R}\}$ cross the horizontal axis. From the first term of the product \eqref{deig}, we note that if 
\begin{equation}
\sin(k h)=0\Rightarrow k=k^{(S)}_j:=\frac{j\pi}{h},\;\;j \in\mathbb{Z}.
\end{equation}
The values $k^{(S)}_j$ define the first set on $k$'s for the local extrema of $\lambda(k)$. The second set is defined by the zeroes of the second term of the product  \eqref{deig}, i.e.,  the roots of the equation
$$ U_s\frac{\gamma_1}{h^2} 4(\cos(k h)-1)+ (\gamma_2U_s-\gamma_3) =0.$$
These zeroes denoted by $k^{(C)}_j$, are equivalently satisfy the equation: 
\begin{equation}\label{eq:eigCos}
\cos(k h)=(\gamma_3-\gamma_2U_s)\frac{h^2}{4\gamma_1U_s} +1\,.
\end{equation}
With the above preparations, we may proceed to the proof of the following result on the instability of $U^{+}_s$.
\begin{proposition}\label{prop: destabilization_wrt_h}
Consider the positive steady state $U^{+}_s$, and assume that the parameters of the lattice \eqref{eq1} satisfy conditions (\ref{Condition1}) and (\ref{Condition2}). 
Then, for 
\begin{equation}\label{eq:critical_h}
h_c={\sqrt{ \frac{2(C_1-\sqrt{C_2})}{f'(U_s^{+})} }},
\end{equation}  
and for all $h<h_c$ there exists a union $J$ of periodically repeated bands $J_0$ with empty intersection, 
such  that $\lambda(k)>0$ for all $k\in J$, i.e., $U_s^+$ is linearly unstable. 
\end{proposition}
\begin{proof}
With the sets $k^{(S)}_j$ and $k^{(C)}_j$ in hand, to identify where the local maxima and minima of $\lambda(k)$ occur, we shall use the second derivative of  $\lambda(k)$ with respect to $k$:
\begin{eqnarray}
\label{2nddeig}
\frac{d^2}{dk^2} \lambda(k)=h^2\left(\cos(k h)\left[U_s\frac{\gamma_1}{h^4} 8(\cos(k h)-1)+2\frac{ (\gamma_2U_s-\gamma_3)}{h^2} \right]-U_s\frac{\gamma_1}{h^4} 8\sin^2(k h)\right).
\end{eqnarray}
From condition (\ref{Condition1}), we may define
\begin{equation}\label{eq:h_1}
h_1:=\sqrt{8\frac{\gamma_1U_s}{-C_1}}.
\end{equation}
In the rest of the proof, we set  $U^s=U_s^+$,  and we shall distinguish between the cases $ h\geq h_1$ and $0<h<h_1$.
\begin{itemize}
\item{\em Case $h\geq h_1$}:
When $ h>h_1$,  for $j=2m$ even, we deduce from \eqref{2nddeig},  that $k^{(S)}_{2m}$ are the positions of local minima, and for $j=2m+1$ odd,  that $k^{(S)}_{2m+1}$ are the positions of local maxima. It is not difficult to check that all the local maxima located at $k^{(S)}_{2m+1}$  have the same (global) value. This is not the case for all the local minima located at  $k^{(S)}_{2m}$. 
Also, when $h=h_1$, we find the following:  $k^{(C)}_j$ coincide with $k^{(S)}_{2j+1}$. Furthermore,  for $h>h_1$ (\ref{eq:eigCos}) has no solutions, and thus, no additional extreme points arise from equation (\ref{eq:eigCos}).
\item{\em Case $0<h<h_1$}: When $0<h<h_1$, for all $j\in\mathbb{N}$,  $k^{(S)}_j$ are the positions of local minima. Furthermore, this is the only case for which solutions of (\ref{eq:eigCos}) exist, and its solutions 
$k^{(C)}_j$  for $j\in\mathbb{Z}$, define the positions of local maxima given by the formula:
\begin{equation}
\lambda\left(k^{(C)}_j\right)=f'(U_s)+\frac{(\gamma_3-\gamma_2U_s)^2}{4\gamma_1U_s}\,. 
\end{equation}
Clearly, $\lambda(k^{(C)}_j)$ do not depend on $h$ and are positive thanks to (\ref{Condition2}).
\end{itemize}
We restrict now to the case  $h>h_1$. In this regime for $h$, we observe that  $\lambda(k^{(S)}_j h)$ with odd $j$ (where the maxima occur) is a decreasing function of $h$. 
Hence we may expect that destabilization  of $U^{+}_s>0$, may emerge at some critical value $h_c>h_1$. In fact, $h_c$ can be found by  the dispersion relation  
\begin{equation}\label{eq:dispersalA}
\lambda(k^{(S)}_j h)=0.
\end{equation}
After some algebra on \eqref{eq:dispersalA}, we get the equation
\begin{equation}\label{eq:dispersal_for_h}
f'(U_s)h^4-4C_1h^2-16U_s\gamma_1=0.
\end{equation}
Using (\ref{Condition1}), we find that the discriminant of the quartic equation \eqref{eq:dispersal_for_h} in $h$, is  $D=16(\gamma_2U_s-\gamma_3)^2+4^3f'(U_s)U_s\gamma_1=16C_2>0$.  A sign-analysis, implies  that  both solutions of equation (\ref{eq:dispersal_for_h}) for $h^2$ are positive, namely, 
\begin{equation}\label{eq:critical-h_s}
h_{\pm}^2=\frac{4C_1\pm\sqrt{D}}{2f'(U_s)}>0\,.
\end{equation}
Next, solving for $h$ and keeping the rest of the parameters fixed so that conditions (\ref{Condition1}) and (\ref{Condition2}) hold, we observe that the positive $h_{\pm}$  given in (\ref{eq:critical-h_s}), satisfy the following relations: 
\[h_-^2-h_1^2 = \frac{2C_1^2+8\gamma_1U_sf'(U_s)-2C_1\sqrt{C_2}}{C_1f'(U_s)} = \frac{2(C_2-C_1\sqrt{C_2})}{C_1f'(U_s)} >0,\]
and 
\[h_+^2-h_1^2 = \frac{2(C_2+C_1\sqrt{C_2})} { C_1f'(U_s)}  <0\,.\]
Therefore,  $h_{-}$ and $h_{+}$ are ordered as
\begin{eqnarray}
\label{ord}
h_+< h_1<h_- .
\end{eqnarray}
From \eqref{ord} we conclude that $h_{-}=h_c$, below which the maxima $\lambda(k^{(S)}_j h)$ become positive, and thus,  $U_s^{+}$ becomes unstable. 
\end{proof}

We remark that Proposition \ref{prop: destabilization_wrt_h} remains valid in the case of Dirichlet or periodic boundary conditions with the modification that the  relevant $k_j^{(S)}$ and $k_j^{(D)}$-sets are finite and determined by the choice of boundary conditions, and thus, destabilization also requires $k_j$ to lie in the instability set $J$.

Figure \ref{fig: Lin_stab} visualizes Proposition \ref{prop: destabilization_wrt_h} for the set of parameters $ \gamma_{1}=0.125, \gamma_{2}= 0.5, \gamma_{3}= 0.005, \alpha=0.02, \beta=0.1$, depicting the graphs of $\lambda(k)$ varying $h$.  
\begin{figure}[th!]
	\includegraphics[scale=0.8]{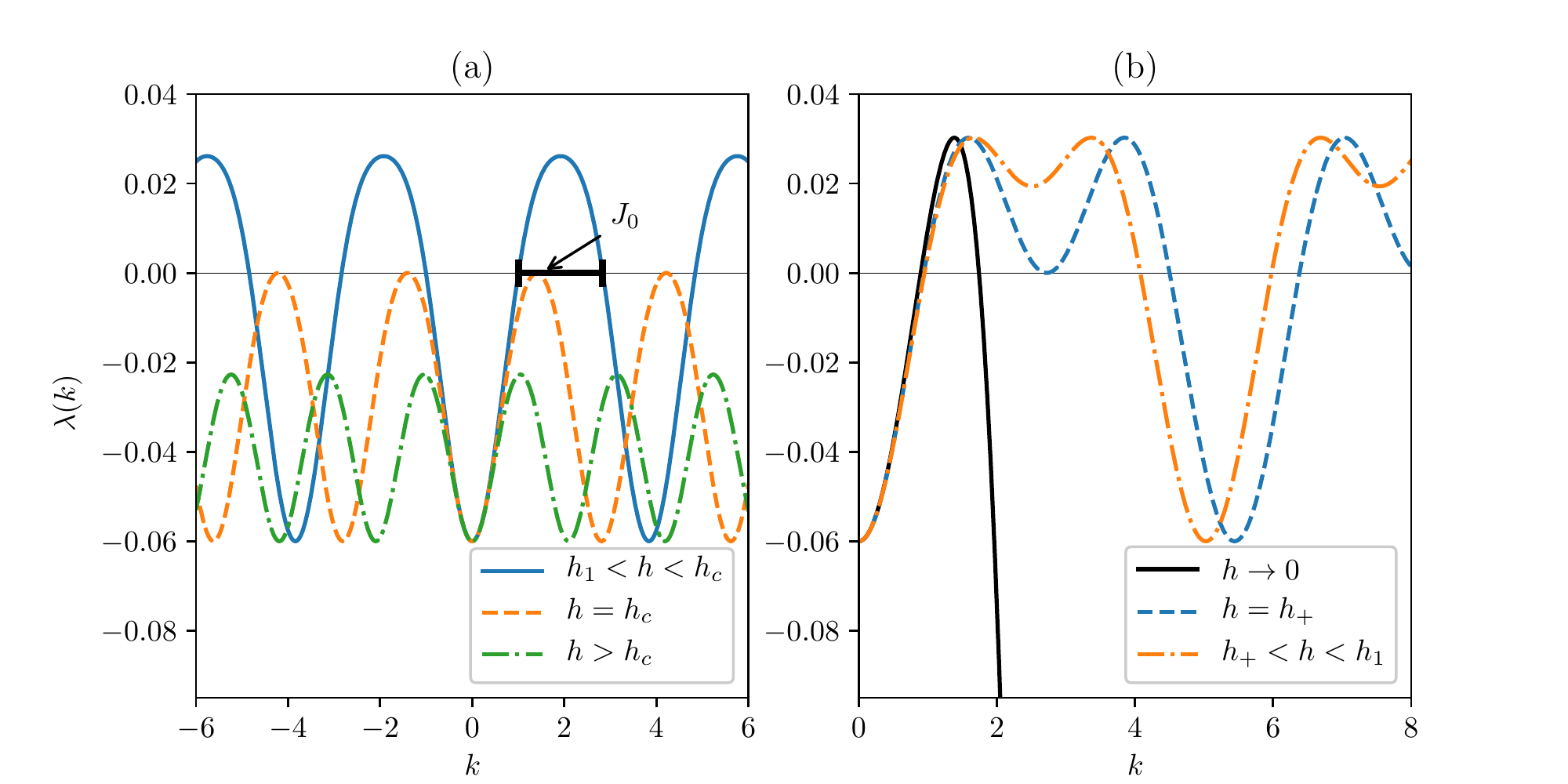}
	\caption{Plots of the graphs of the eigenvalue functions $\lambda(k)$ for different  values of the discretization parameter $h$.   Parameters: $ \gamma_{1}=0.125, \gamma_{2}= 0.5, \gamma_{3}= 0.005, \alpha=0.02, \beta=0.1$. On panel (a): The dotted dashed (green) curve corresponds to the case $h>h_c$. The dashed (orange) curve corresponds to the instability threshold $h=h_c$. The continuous (light blue) curve corresponds to the case that $h_1<h<h_c$, with $J_0$ denoting the principal instability band.
	On panel (b): The dashed (light blue) line corresponds to $h=h^+$ which is given in \eqref{eq:critical-h_s}. The solid black line depicts the case where $h$ approaches zero ($h\ll h+$) and  the dashed-dotted line (orange) corresponds to $h_+<h<h_1$.}
\label{fig: Lin_stab}
\end{figure}
As discussed in the proof of  Proposition \ref{prop: destabilization_wrt_h} and illustrated in Figure \ref{fig: Lin_stab}, 
for $h_1\leq h<h_{-}$, the principal maximum value of $\lambda$ is located at $k^{(S)}_1=\pi/h$, while for $h<h_1$, $k^{(S)}_1$ becomes
the position of a positive local minimum. The principal instability band $J_0$ of frequencies $k$, for which $\lambda(k)>0$, is initially an interval formed around $k^{(S)}_1$, 
and splits when $h$ becomes less or equal to $h_+>0$ given in \eqref{eq:critical-h_s},  
shaping two intervals $J_0^A$, $J_0^B$ symmetric with respect to $k^{(S)}_1$ so that  $J_0= J_0^A\cup J_0^B$.
Finally, letting $h\to 0$ in \eqref{eq:eigenvalues} implies $\lambda(k)=f'(U_s)-C_1k^2-\gamma_1U_s k^4$, which is the characteristic polynomial of the continuous system.

\paragraph{From spatially uniform to spatially periodic equilibria.}
Of specific physical significance, relevant to pattern formation of vegetation patches are spatially periodic equilibria which can be considered as perturbations of a spatially uniform state. It is of primary interest whether such spatial structures emerge in the convergence dynamics of the system \eqref{eq1}. 
They can be approximated by the simplest periodic ansatz
\begin{eqnarray}
\label{spans}
U_P(n)=U_s+\epsilon\cos(khn),\;0<\epsilon\ll1,
\end{eqnarray}
and may emerge when $\lambda(k)>0$, as suggested by the previous linear stability analysis. 
Here, $P$ will denote the period of the discrete function \eqref{spans}, which has to be a positive integer satisfying $U_P(n)=U_P(n+P)$. 
Therefore, $khP=2\pi m$  for some $m\in\mathbb{Z}$.%

Now, for integers $m\in\mathbb{Z}$ and $P\in\mathbb{Z}^+$, a rational multiple of $2\pi/h$, namely,
\begin{eqnarray}
\label{spans1}
k(m,P)=\frac{m}{P}\frac{2\pi}{h},
\end{eqnarray}
provides the frequency of a spatial oscillation given that  $k(m,P)$ belongs to the instability set $J$.
In the case where $m\in\mathbb{Z}$ and $P\in\mathbb{Z}^+$ have no common factors, $P$ is called the fundamental period of \eqref{spans}, 
and  $m$ is called the envelope,  determining the shape of \eqref{spans}.
Note that for a lattice supplemented with periodic boundary conditions, the number of nodes has to be a multiple of the principal period $P$.
The case $|m/P|=1$ and $m=0$ correspond to the so called constant, or trivial modes. 
For each admissible $(m,P)$ (in the sense that the corresponding $k$ belongs to the instability set) 
with $P$ being a fundamental period,  a representative of $m$ can be traced in the interval $(0,P)$ (and the corresponding $k$ in $J_0$). 
Moreover,  for each $m \in (0,P)$ there is an additional envelope  $m'=P-m \in (0,P)$, which we may call the symmetric envelope, 
that gives identical discrete sinusoids, since $\cos(\frac{2\pi m}{P} n)= \cos(\frac{2\pi m'}{P} n)$. 

For instance, when $P>2$ is a prime, we can get up to $(P-1)/2$ non-trivial distinct modes with such a period. 
In order, to count the number of distinct modes with a particular fundamental period $P$ in general, 
we must take into account only $m$ for which $P$ and $m$ are relatively prime, 
and exclude symmetric $m$-values which provide identical configurations. 

\subsection{Numerical results}\label{section:numerics}
The numerical simulations on the dynamics of the  DLL lattice  (\ref{eq1}) will consider the following cases of initial conditions: 
\begin{itemize}
\item Positive spatially periodic initial data of the form:
\begin{equation}\label{eq:IC2}
U_n^0=U_n(0) =A \cos(k x_ n)+ B\,.
\end{equation}
The initial condition \eqref{eq:IC2} is relevant in investigating the dynamics of periodic structures of the form \eqref{spans}. It is compliant with the periodic boundary conditions but also with the Dirichlet, for  suitable choices of $k$. 
\item  Symmetric box-shaped localized concentrations
centered around $x=0$. More precisely, for a symmetric sub-interval $W\subset [-L,L]$  of with $w$,  and $A>0$, the initial condition has the form:
\begin{equation}\label{eq:IC1}
U_n^0=U_n(0) = \left\{
\begin{array}{cc}
A, &\text{if}\quad x_n \in W,\\
0, &\text{if}\quad x_n \notin W\,.
\end{array}
\right.
\end{equation}
where $A$ is the initial amplitude. Note that the box-shaped initial condition can even be an impulse located at $x=0$ when the distance $h$ between the nodes of the lattice is large enough. The initial condition \eqref{eq:IC1} is compliant with all the types of boundary conditions (vanishing, Dirichlet and periodic).
\end{itemize} 
\subsubsection{Dirichlet boundary conditions}
In the case of the DLL lattice \eqref{eq1} supplemented with Dirichlet boundary conditions, the numerical experiments will examine the  global  and linear stability criteria derived in sections III and IV, for the trivial steady state $U^0_s=0$, particularly in the light of Remark \ref{RemTS}.   We recall that the trivial steady state $U^0_s=0$, is of particular physical significance since it is associated with the spatial extinction of vegetation patches, and thus, desertification dynamics.

With the above motivation, we examine the evolution of the spatially extended initial condition \eqref{eq:IC2}, with fixed $k=\frac{\pi}{L}$, varying the amplitudes $A=B$. The half-length interval is $L=30$ and the discretization parameter is $h=1.5$.  The parameter $\alpha=0.1$,  and for the above values of $h$ and $L$, we find from the formula \eqref{crucequiv} for $\mu_1$ and \eqref{m1crit} for $\hat{\gamma}_{3,\mathrm{thresh}}$, that 
\begin{eqnarray*}
\mu_1\approx 0.002,\;\; \hat{\gamma}_{3,\mathrm{thresh}}=\frac{\alpha}{\mu_1}\approx 36.5.
\end{eqnarray*}	
We investigate two examples for $\beta>0$, and $\beta<0$, respectively.
Figure \ref{fig: beta>0} depicts the evolution of the $\ell^2$-norm of the solution when $\beta=0.1$ for three different values of $\gamma_3$. The three dotted-dashed (red) curves correspond to the evolution of the initial condition (\ref{eq:IC2}) for $A=B=0.2$ (upper starting curve), $A=B=0.05$ (middle starting curve) and $A=B=0.005$ (bottom starting curve), when $\gamma_3=36<\hat{\gamma}_{3,\mathrm{thresh}}$.  For this choice of $\gamma_3$, the criterion of Proposition \ref{stabt} on the linear stability of $U^0_s=0$ and  of Proposition \ref{PInf3} on its global stability, is violated. Here is where the results of Proposition \ref{PInf6} on the existence of non-trivial attracting sets come into play. The numerical results show that the global attractor is a non-trivial equilibrium possessing the form of a a hump-shaped non-negative state, shown in the inset of Fig. \ref{fig: beta>0}. Furthermore, the form of the graphs of the $\ell^2$-norm justify the qualitative relevance of the functional form of the estimates derived in Proposition \ref{PInf6}, from the solutions of the Bernoulli-type inequalities; the graphs seem to represent logistic-type integral curves as being solutions of Bernoulli type ODE's.
\begin{figure}[tb!]
	\includegraphics[scale=0.8]{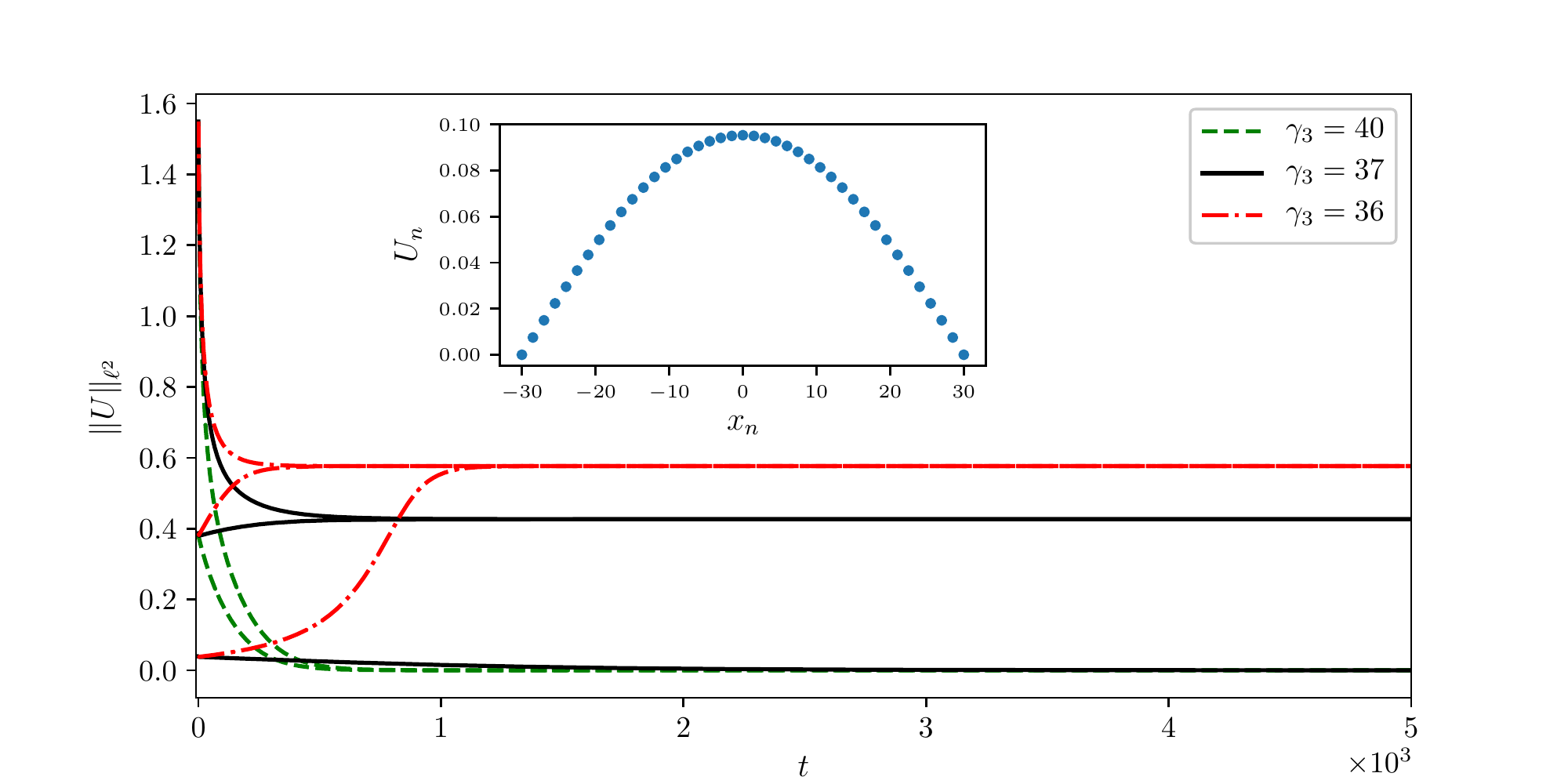}
	\caption{Dynamics for $\alpha>0$, $\beta>0$: Evolution of the $\ell^2$-norm of the solutions of the lattice \eqref{eq1} supplemented with Dirichlet boundary conditions, when starting  from  initial conditions \eqref{eq:IC2} for various cases of $A=B$, $\beta=0.1$, varying $\gamma_3$. Other parameters: $\gamma_{1}=0.125, \gamma_{2}= 0.5, \alpha=0.1$, $L=30$ and $h=1.5$.  The set of the dotted-dashed (red) curves corresponds to the case $\gamma_3=36$ (upper curve for $A=B=0.2$, middle curve for $A=B=0.05$ and bottom curve for $A=B=0.005$). The set of the continuous (black) curves corresponds to the case $\gamma_3=37$ (upper curve for $A=B=0.2$, middle curve for $A=B=0.05$ and bottom curve for $A=B=0.005$). The couple of the dashed (green) curves corresponds to the case $\gamma_3=40$ (upper curve for $A=B=0.2$ and bottom curve for $A=B=0.05$). The inset depicts the non-negative equilibrium attracting all trajectories when $\gamma_3=37$.}
	\label{fig: beta>0}
\end{figure}
\begin{figure}[tbh!]
	\includegraphics[scale=0.8]{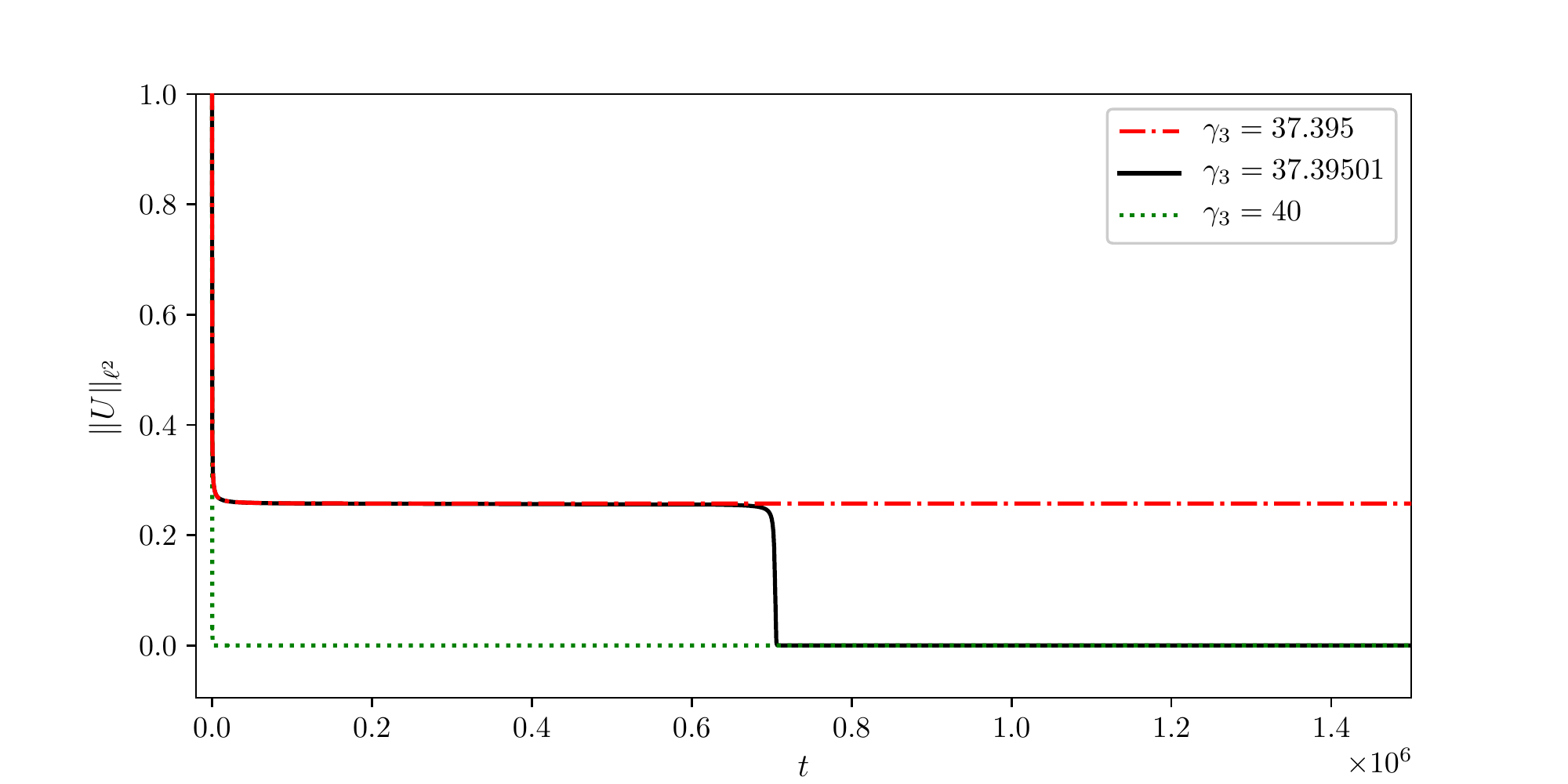}
	\caption{Dynamics for $\alpha>0$, $\beta>0$: Evolution of the $\ell^2$-norm of the solutions of the lattice \eqref{eq1} supplemented  with Dirichlet boundary conditions, when starting from the initial condition \eqref{eq:IC2} for $A=B=0.28$, with $\|U^0\|_{\ell^2} \approx 2.169$, and varying $\gamma_3$. 
Dotted curve (green) $\gamma_3=40> \gamma_3^*$.
Solid curve (black)  $\gamma_3=37.3950$ (approximately $\gamma_3^*$).
Dashed-dotted  curve (red)  $\gamma_3=37.395<\gamma_3^*$.
Other parameters: $\gamma_{1}=0.125, \gamma_{2}= 0.5, \alpha=0.1$, $\beta=0.1$, $L=30$ and $h=1.5$. }
	\label{fig: beta>trans}
\end{figure}
Increasing $\gamma_3=37>\hat{\gamma}_{3,\mathrm{thresh}}$, the $\ell^2$-norm evolution is portrayed by the continuous (black) curves. The upper starting curve corresponds to the dynamics of the initial condition \eqref{eq:IC2} with $A=B=0.2$, the middle  starting curve to $A=B=0.05$ and the bottom curve to $A=B=0.005$. The threshold value $\hat{\gamma}_{3,\mathrm{thresh}}$ seems to be quantitatively sharp for the linear stability of $U^0_s=0$ and significantly qualitative relevant. The lower continuous curve corresponds to the evolution of small $\ell^2$-norm initial data; the initial condition in this case seems to be in the domain of attraction of $U^0_s=0$. As $\gamma_3=37$ is slightly larger than $\hat{\gamma}_{3,\mathrm{thresh}}$, the steady-state $U_0^s=0$ should be only linearly asymptotically stable and not globally nonlinearly stable. Thus, initial data of larger norm should be out of the domain of attraction of $U^0_s=0$ possessing non-trivial dynamics. This is the case represented by the middle and upper continuous curves, converging to a non-trivial equilibrium similar to the one shown in the inset.  
However, for the same couple of initial conditions we numerically identified a threshold value $\gamma_3^*\approx 37.39501>\hat{\gamma}_{3,\mathrm{thresh}}$ 
($ 37.395<\gamma_3^*< 37.39501$)  with the following property. For $\gamma_3>\gamma_3^*$, all the initial conditions  converge to the trivial equilibrium $U^0_s=0$. 
This is the case depicted by the couple of the dashed (green) curves, where the upper curve corresponds to $A=B=0.2$ 
and the bottom curve to $A=B=0.05$, when $\gamma_3=40>\gamma_3^*$. 

The transition from bistable dynamics to monostable dynamics  was further investigated for $\gamma_3 > \hat{\gamma}_{3,\mathrm{thresh}}$,
as it is shown in Figure \ref{fig: beta>trans}. In particular, for an initial condition with relative large norm $\|U^0\|_{\ell^2}\approx 2.169$, when $\gamma_3=37.395$ 
the trajectory converges to the positive non-trivial equilibrium, while for $\gamma_3=37.39501$
it stays close to that equilibrium for a long time but eventually diverges widely from it, converging to the zero state.

Actually, repeated numerical experiments for the same fixed values $\alpha,\beta,\gamma_1,\gamma_2>0$, and
varied $\gamma_3$-values and norm $||U^0||_{\ell^2}$  of the initial data, 
revealed the following dynamical scenarios:
\begin{enumerate}
	\item When  $\gamma_3<\hat{\gamma}_{3,\mathrm{thresh}}$, the trivial equilibrium $U^s_0=0$ is unstable,  and a non-trivial positive equilibrium attracts all the trajectories. 
	\item When  $\gamma_3>\hat{\gamma}_{3,\mathrm{thresh}}$, the trivial equilibrium $U^s_0=0$ is linearly stable. Moreover, there exists a threshold value 
	$\gamma_3^*>\hat{\gamma}_{3,\mathrm{thresh}}$ such that:
\begin{enumerate}
\item if $\hat{\gamma}_{3,\mathrm{thresh}}<\gamma_3<\gamma_3^*$, then
for sufficiently small norm $||U^0||_{\ell^2}$  of the initial data,  the trivial steady state $U^s_0=0$ attracts their trajectories, 
while for initial conditions of larger norm, a non-trivial positive equilibrium attracts their trajectories (bistable case).
\item if $\gamma_3>\gamma_3^*$, then $U^s_0=0$ attracts all the trajectories for all initial conditions (monostable case)
\end{enumerate}		
\end{enumerate}

Regarding scenario 1, it should be noted that for $0<\gamma_3\ll \hat{\gamma}_{3,\mathrm{thresh}}$,  the convergence dynamics may involve equilibria with a complex spatial structure than the positive hump portrayed in the inset of Fig. \ref{fig: beta>0}. Such equilibria  are similar to those that will be illustrated in the next paragraph, discussing the dynamics of the lattice supplemented with the periodic boundary conditions. Regarding scenario 2, it elucidates further the physical significance of Proposition \ref{PInf3}, on the global (uniform with respect to all initial data) stability of $U^s_0$: whatever the magnitude of the norm $||U^0||_{\ell^2}$ of the initial condition is, there exists a {\em universal extinction threshold} $\gamma_{3,\mathrm{thresh}}>0$, such that if $\gamma_3>\gamma_{3,\mathrm{thresh}}$, the trivial state $U^s_0=0$ is the globally attracting state; it represents an extreme scenario where desertification is inevitable if the parameter $\gamma_3$ exceeds such a universal threshold.  
\begin{figure}[bh!]
	\includegraphics[scale=0.8]{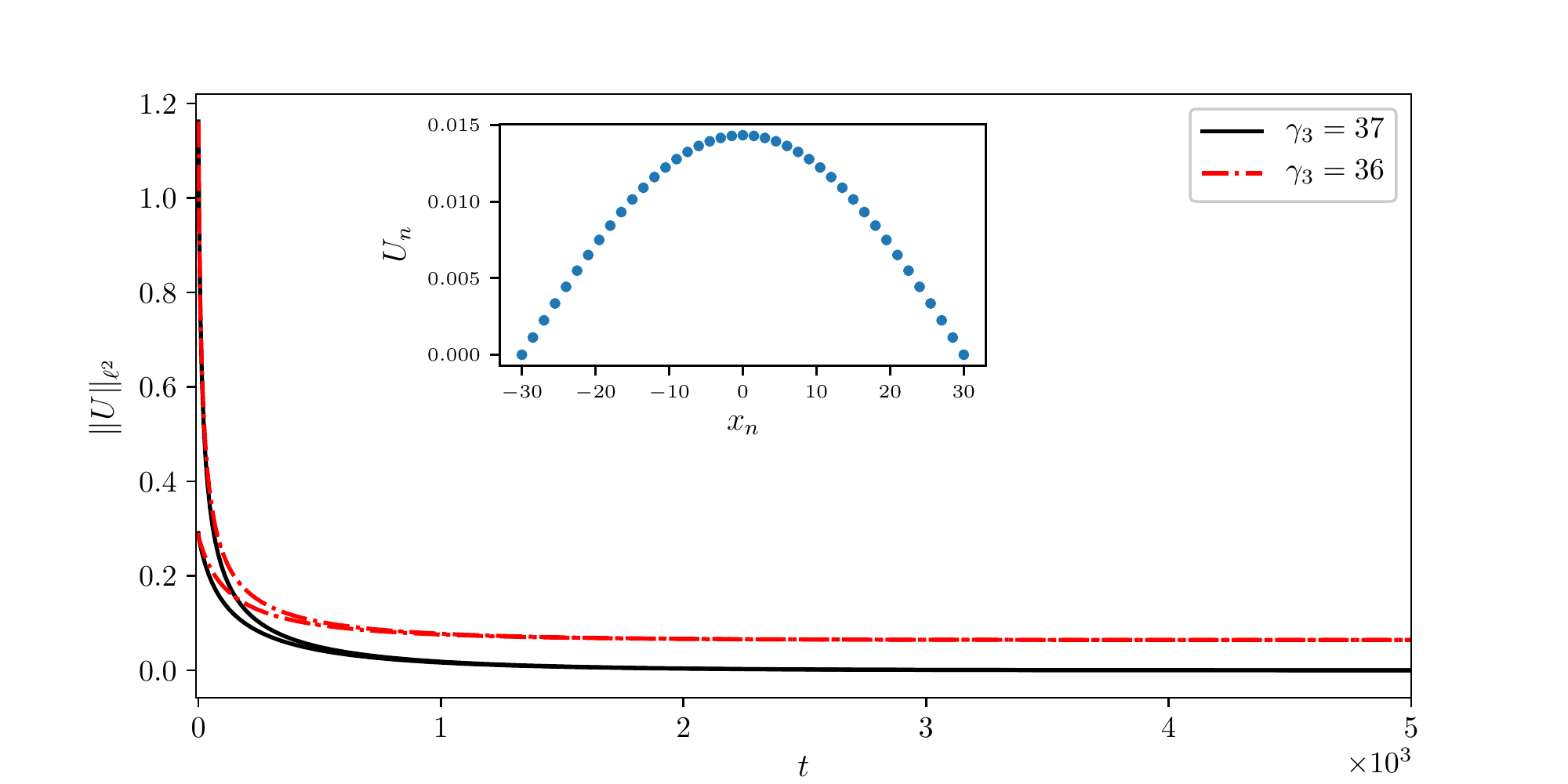}
	\caption{Dynamics for $\alpha>0$, $\beta<0$: Evolution of the $\ell_2$-norm of the solutions of the lattice \eqref{eq1} supplemented  with Dirichlet boundary conditions, when starting  from  initial conditions \eqref{eq:IC2} for various cases of $A=B$, $\beta=-0.1$, varying $\gamma_3$. Other parameters: $\gamma_{1}=0.125, \gamma_{2}= 0.5, \alpha=0.1$, $L=30$ and $h=1.5$.  The set of the dotted-dashed (red) curves corresponds to the case $\gamma_3=36$ (upper curve for $A=B=0.15$, bottom curve for $A=B=0.03$). The set of the continuous (black) curves corresponds to the case $\gamma_3=37$ (upper curve for $A=B=0.15$, bottom curve for $A=B=0.03$).}
	\label{fig: beta<0}
\end{figure}

Figure \ref{fig: beta<0},  depicts the evolution of the $\ell^2$-norm of the solution when $\beta=-0.1$ for two different cases of $\gamma_3$.  The initial condition is yet \eqref{eq:IC2} and the rest of the parameters are fixed as above. The numerical results illustrate the sharpness of the analytical predictions of Proposition \ref{PInf4}, on the global stability of $U^0_s=0$, which complement the  analytical predictions for the regime $\alpha>0$ and  $\beta<0$, beyond linearization.  The  couple of continuous (black) curves show the dynamics when $\gamma_3=37>\hat{\gamma}_{3,\mathrm{thresh}}$ (upper starting curve for $A=B=0.15$ and lower curve for $A=B=0.03$). We observe that the trivial steady state $U^0_s=0$ attracts the trajectories for both cases of the initial conditions. The couple of the dotted-dashed (red) curves show the dynamics when $\gamma_3=36<\hat{\gamma}_{3,\mathrm{thresh}}$ (upper starting curve for $A=B=0.15$ and lower curve for $A=B=0.03$). This is the case of the instability of $U^0_s=0$, and the trajectories are attracted by the positive steady-state portrayed in the inset. 

We remark that the same dynamics as discussed above are observed when the box-initial data are used (not shown here) for both regimes of $\alpha>0$, $\beta>0$ and $\alpha>0$, $\beta<0$.

\subsubsection{Periodic boundary conditions}
In the case of the periodic boundary conditions, simulations are performed for fixed values of $\beta=0.1$, $L=60$,  $\gamma_3=0.005$, and various values of $h$.
The parameter $\alpha$  will take either negative or positive values.  


When $\alpha>0$ is fixed, the results of the numerical simulations are in full agreement with the analytical predictions of the linear stability analysis: for sufficiently large $h$, the positive uniform equilibrium $U^+_s>0$ attracts all the trajectories, while lower values of $h$ cause its destabilization,  giving rise to spatially-periodic equilibria. Remarkably, the threshold value $h_c$ proved to be sharp in distinguishing the above behaviors.   

For $\alpha= 0.02$, we find that $U^+_s=0.2$ and $h_c\approx 2.236$.  Therefore, when $h<h_c$ the uniform state $U^+_s$ should be unstable: the maximum value of $\lambda(k)$ is attained at $k_1=\frac{\pi}{h}$. Admissible choices of $(m,P)$ are such that $k = \frac{2m}{P}\frac{\pi}{h}\in J_0$,  and for these admissible choices of $(m,P)$, spatially periodic equilibria  should emerge with distinct profiles. Figure \ref{fig:A3} portrays the profiles of the equilibrium states for an initial condition \eqref{eq:IC2} with $A=B=U_s^+$ (which is a harmonic perturbation of $U_s^+$ of the form \eqref{spans}),  when $h=2<h_c$.  The principal instability band is found $J_0=(1.107,2.034)$.  For $h=2$, the maximum value of $\lambda(k)$ is attained at $k_1=\pi/2\approx 1.57$ which coincides with $k(m,P)$ in \eqref{spans1} for the pair $(m,P)=(1,2)$, and if $k_1$ is inserted in the initial condition, the solution converges to the equilibrium shown in the left panel (a); the mode $\cos(n \pi )$ grows,  developing  a mosaic-solution (recurring every second point).  For this same value of $h=2$, we observe the remarkable changes between the profiles of the equilibria, when $(m,P)=(2,5)$ with $k(m,P) \approx 1.256$ shown in the middle panel (b), and $(m, P)=(4,9)$ with $k\approx 1.396$ shown in the right panel (c). Note that since the number of the lattice nodes is $N=60$, the equilibrium of panel (b) is $5$-periodic as $P=5$ is a divisor of $N$. This is not the case for the equilibrium of panel (c) since $P=9$ is not a divisor of $N$. Therefore, 
the system does not admit 9-periodic solutions and instead a mosaic solution consisting of 5 and 2 -periodic solutions is reached, that satisfies the boundary conditions.
\begin{figure}[bp!]
	\begin{tabular}{ccc}
		(a)&(b)&(c)\\
		\includegraphics[scale=\figscaleee]{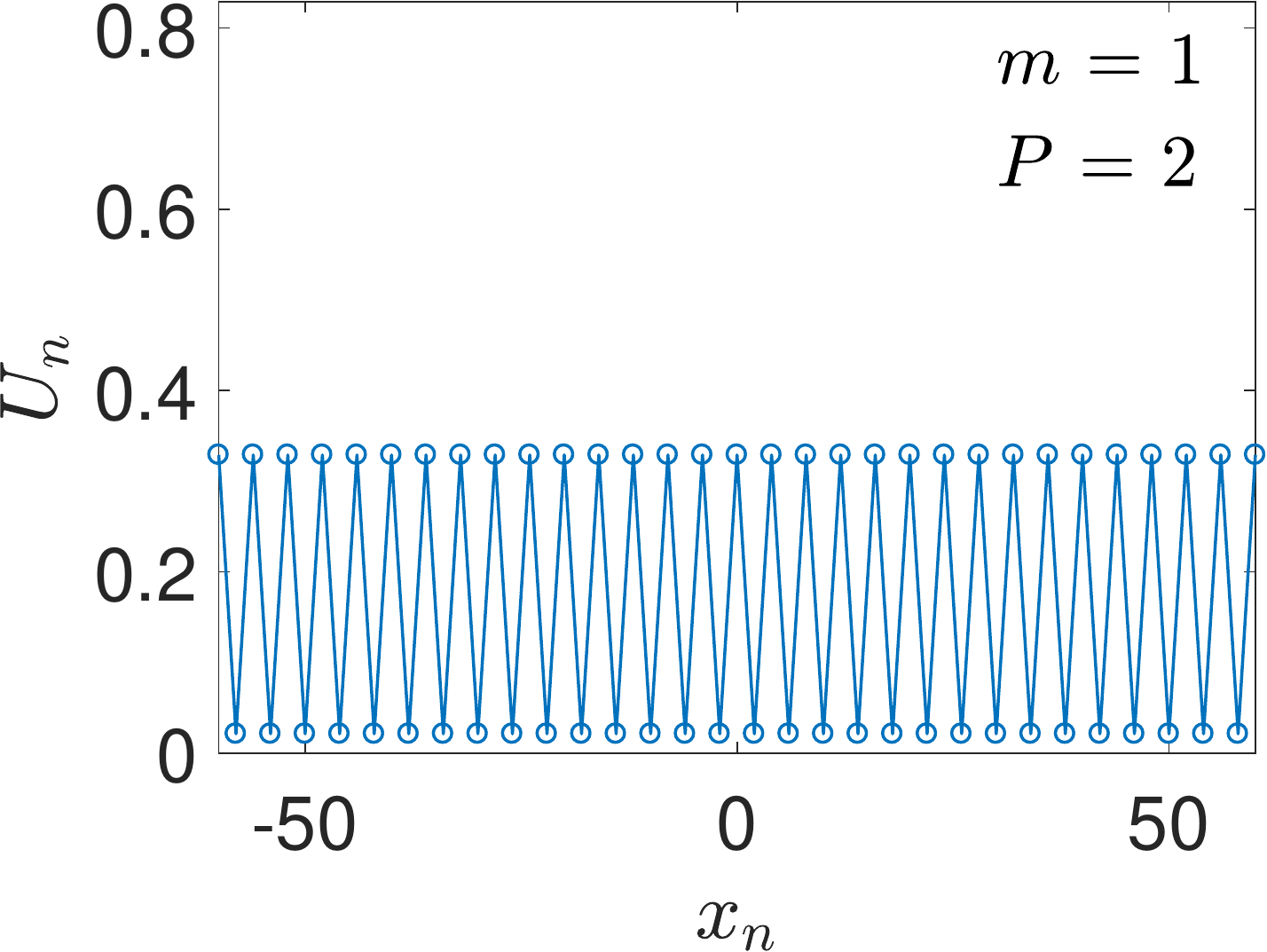}&
		\includegraphics[scale=\figscaleee]{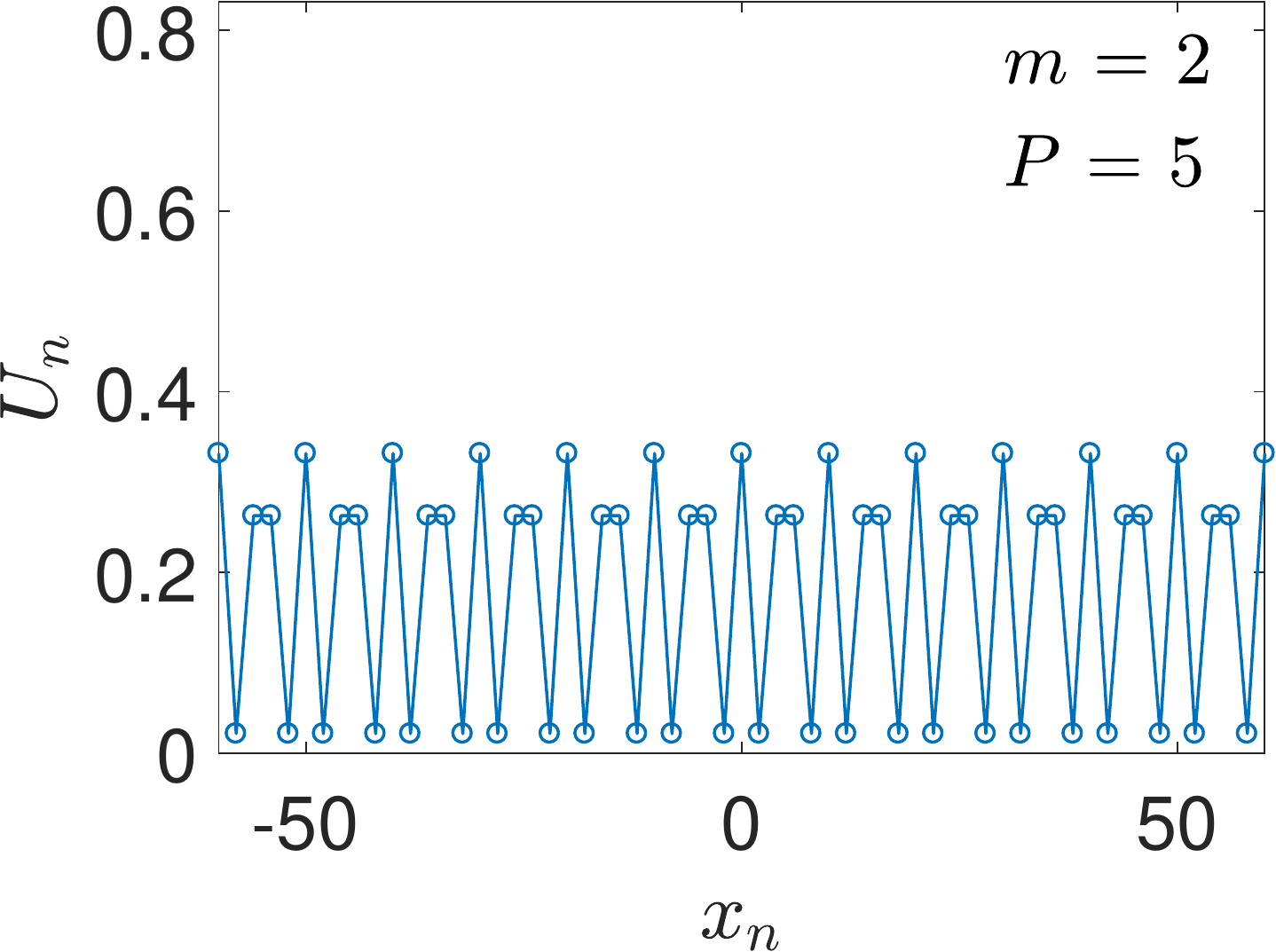}&
		\includegraphics[scale=\figscaleee]{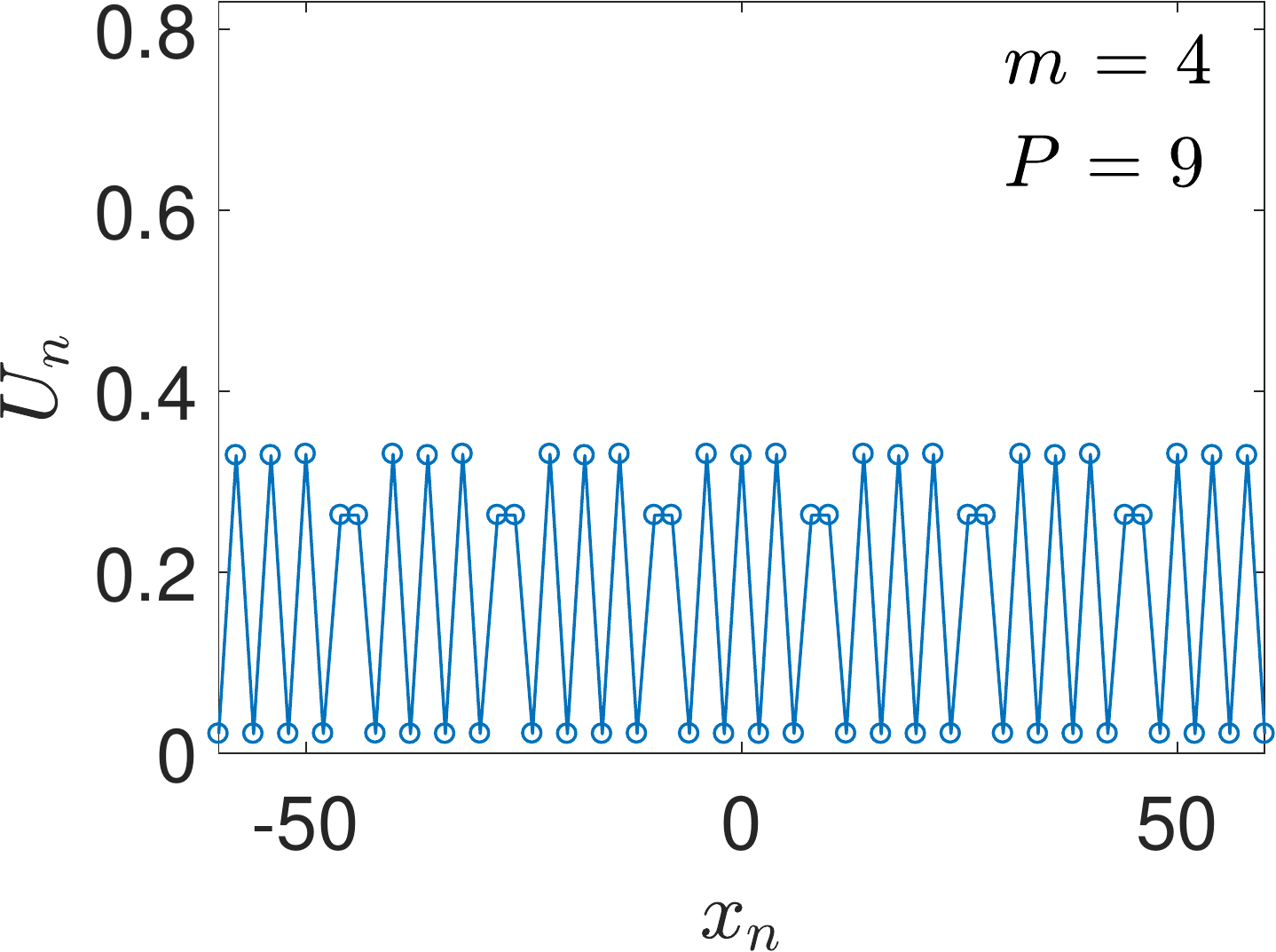}
	\end{tabular}
	\caption{Steady-states for the evolution of the system \eqref{eq1}, for the set of parameters $h=2<h_c=2.236$, $L=60, \gamma_{1}=0.125, \gamma_{2}= 0.5, \gamma_{3}= 0.005, \alpha=0.02, \beta=0.1$\, and the initial condition \eqref{eq:IC2} with $A=B=0.2$ and frequency $k = \frac{2\pi}{h}\frac{m}{P}\in J_0=(1.107,2.034)$ for different choices of $P$ and  $m$.
		}
	\label{fig:A3}
\end{figure}

Let us comment further on the formation mechanism of steady states, motivated by the example of Fig. \ref{fig:A3}. For $P=3$ there is no $m\in\mathbb{Z}$ to form an admissible pair, thus a $3$-spatially periodic equilibrium cannot occur . For $P=4$, the only such an $m=2$. However, since for $(m,P)=(2,4)$, the associated $k=k_1$, the achieved steady-state is the one corresponding to the pair $(m,P)=(1,2)$. Therefore, there are no $4$-periodic equilibria. 
When $P=5$, the pairs $(m,P)=(2,5)$ and $(m,P)=(3,5)$ give rise to the same periodic equilibrium, since $m=3$ and $m'=P-3=2$ are symmetric choices for $m$.

Keeping $\alpha= 0.02$ fixed, a decreasing of $h$ causes the expansion of  $J_0$ (see Figure \ref{fig: Lin_stab}), and thus more admissible $k$-values arise. 
It turns out that additional periodic steady-states can be formed. 
However, when $h$ becomes less than $h_+\approx 1.1547$  the band splits and reduces the possible growing modes. 
Finally, in the spatially continuum regime (i.e. $h\ll1$), the spatial-patterns that appear approach continuous sinusoids. 
It becomes clear from the above, that for intermediate values of  $h<h_c$ in the discrete regime, the system may exhibit richer dynamics than for $h$ in the continuous regime, 
in terms of the spatial structure of the equilibria that it forms.  
Multiple, stable spatially-periodic states seem to exist for fixed values of $h$, which can be patched together forming complex hybrid-states. These states are not necessarily mosaics of periodic patterns, but they may also be localized structures in a non-uniform periodic background.
Various such states can be obtained when testing the dynamics of the system for localized initial conditions. In this case,  invasion phenomena may also emerge. 
\begin{figure}[tbh!]
	\centering 
	\includegraphics[scale=0.36]{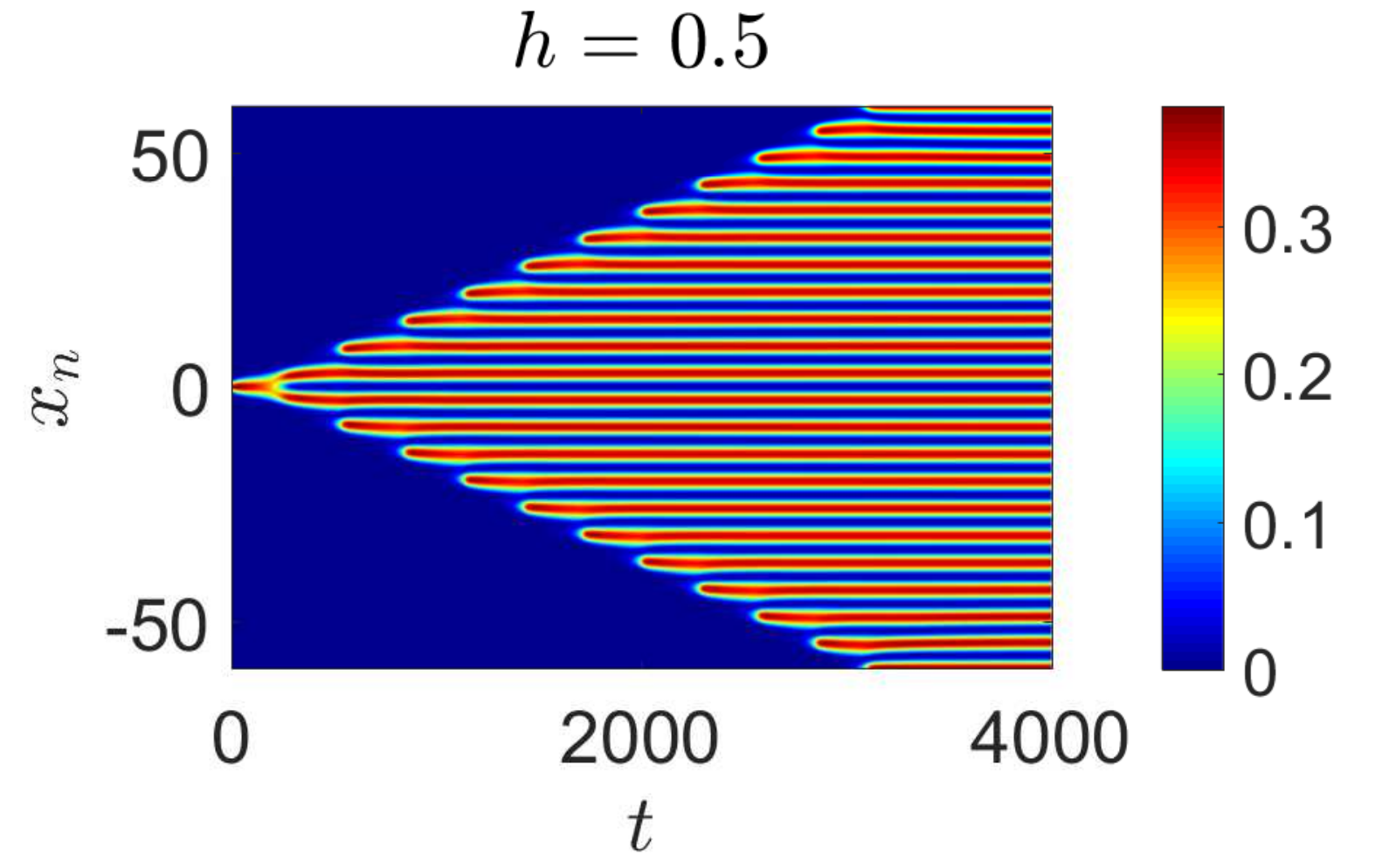}
	\includegraphics[scale=0.36]{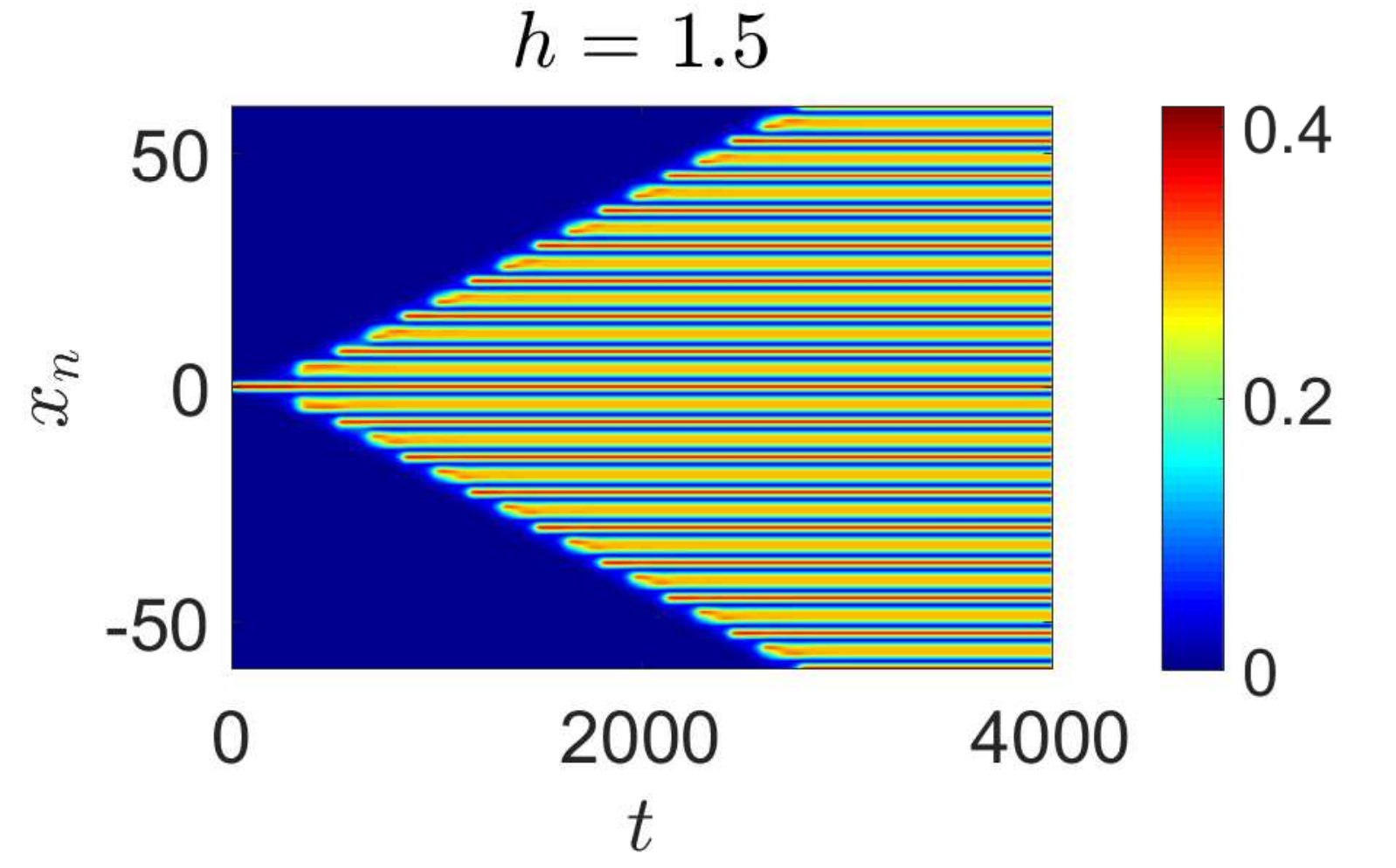}
	\includegraphics[scale=0.36]{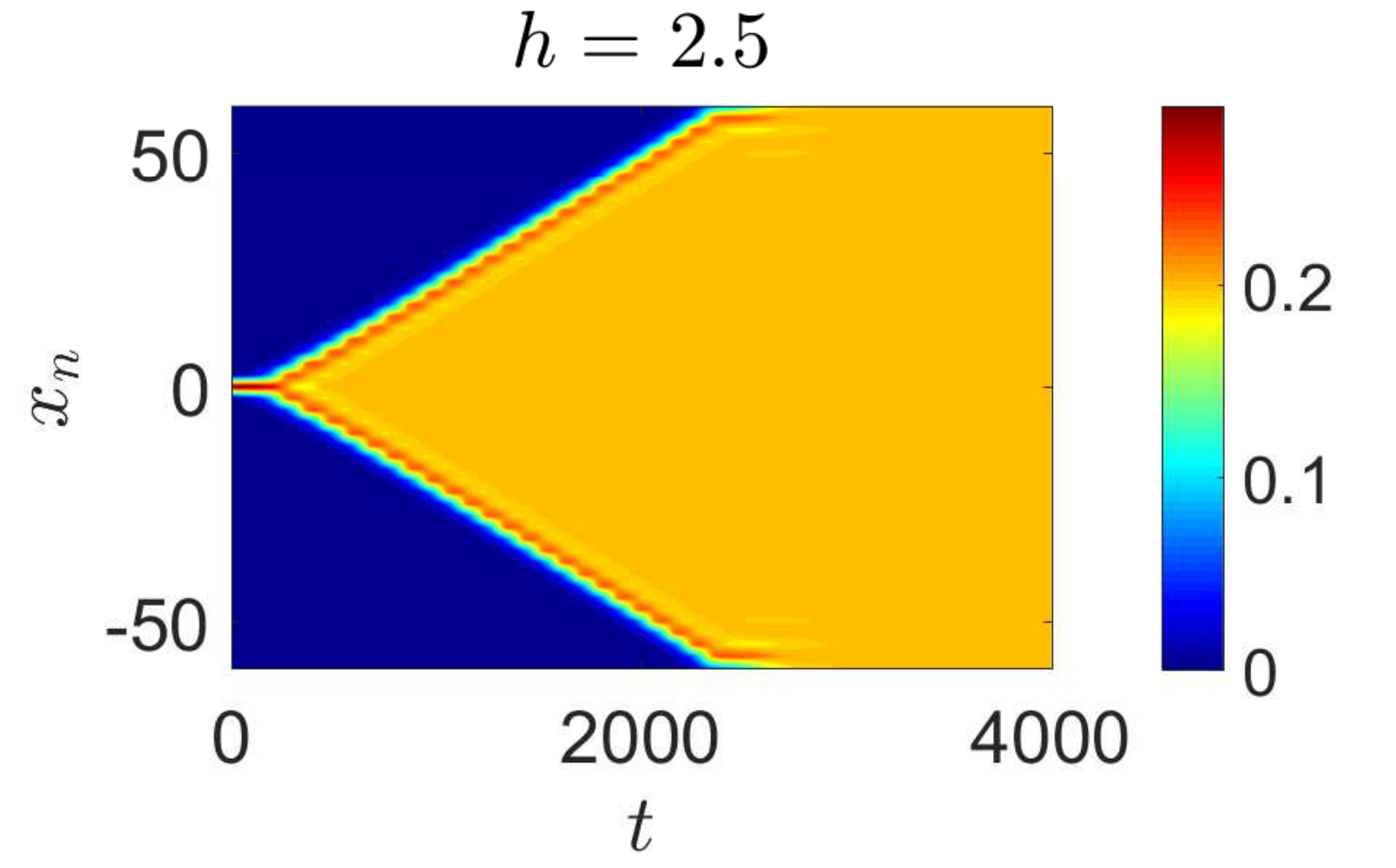}
	\caption{Contour plots of the time-evolution of a  box-profiled initial condition with $w=2$ and amplitude $A=0.28$, for increasing values of $h$. Left panel for $h=0.5<h_c=2.2236$, middle panel for $h=1.5<h_c$ and right panel for $h=2.5>h_c$.  Other parameters: $L=60, \gamma_{1}=0.125, \gamma_{2}= 0.5, \gamma_{3}= 0.005, \alpha=0.02, \beta=0.1$.}
	\label{fig:A2}
\end{figure}	
\begin{figure}[th!]
	\centering 
	\begin{tabular}{ccc}
		(a)&(b)&(c)\\
		\includegraphics[scale=\figscaleee]{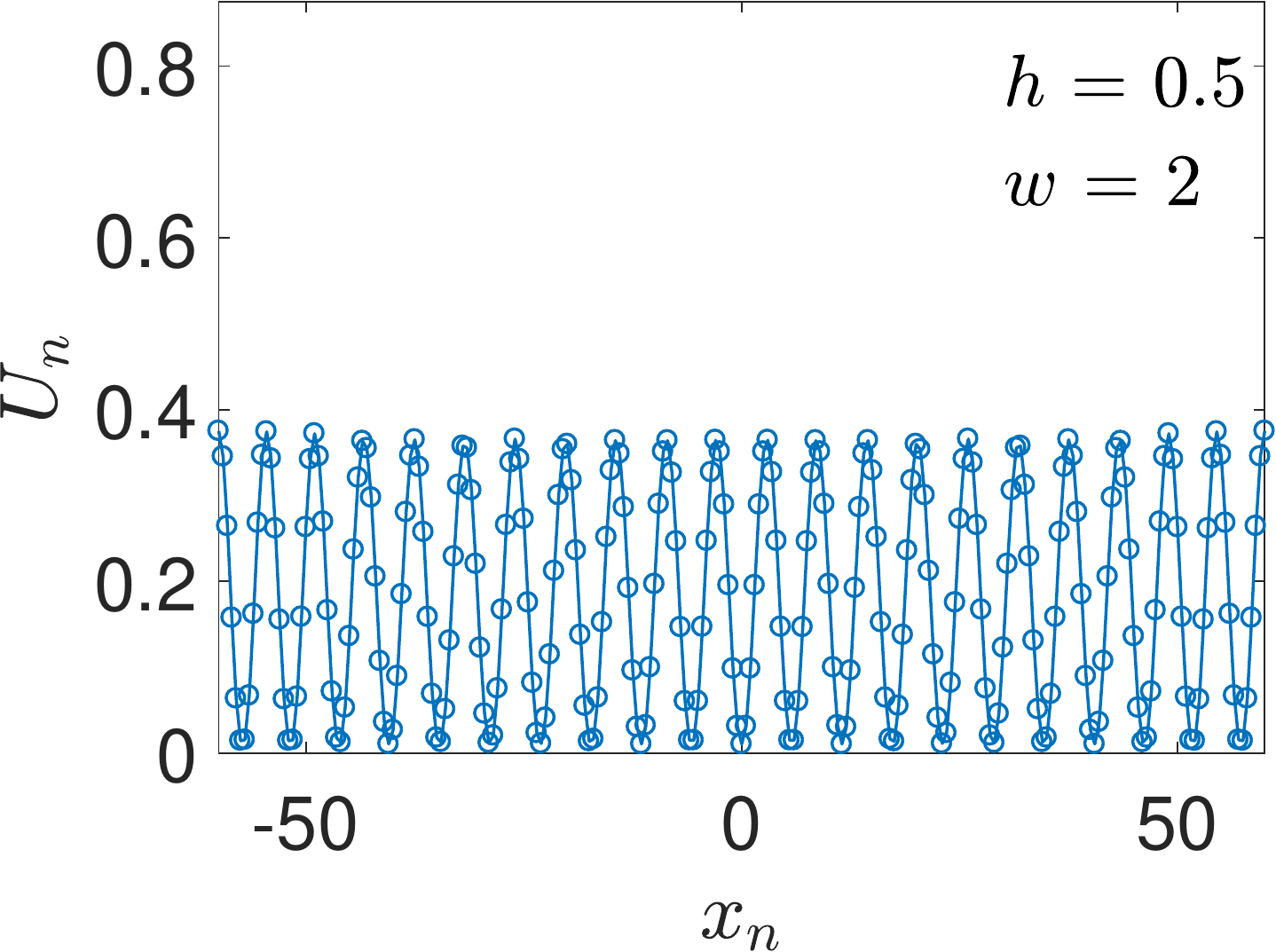}&
		\includegraphics[scale=\figscaleee]{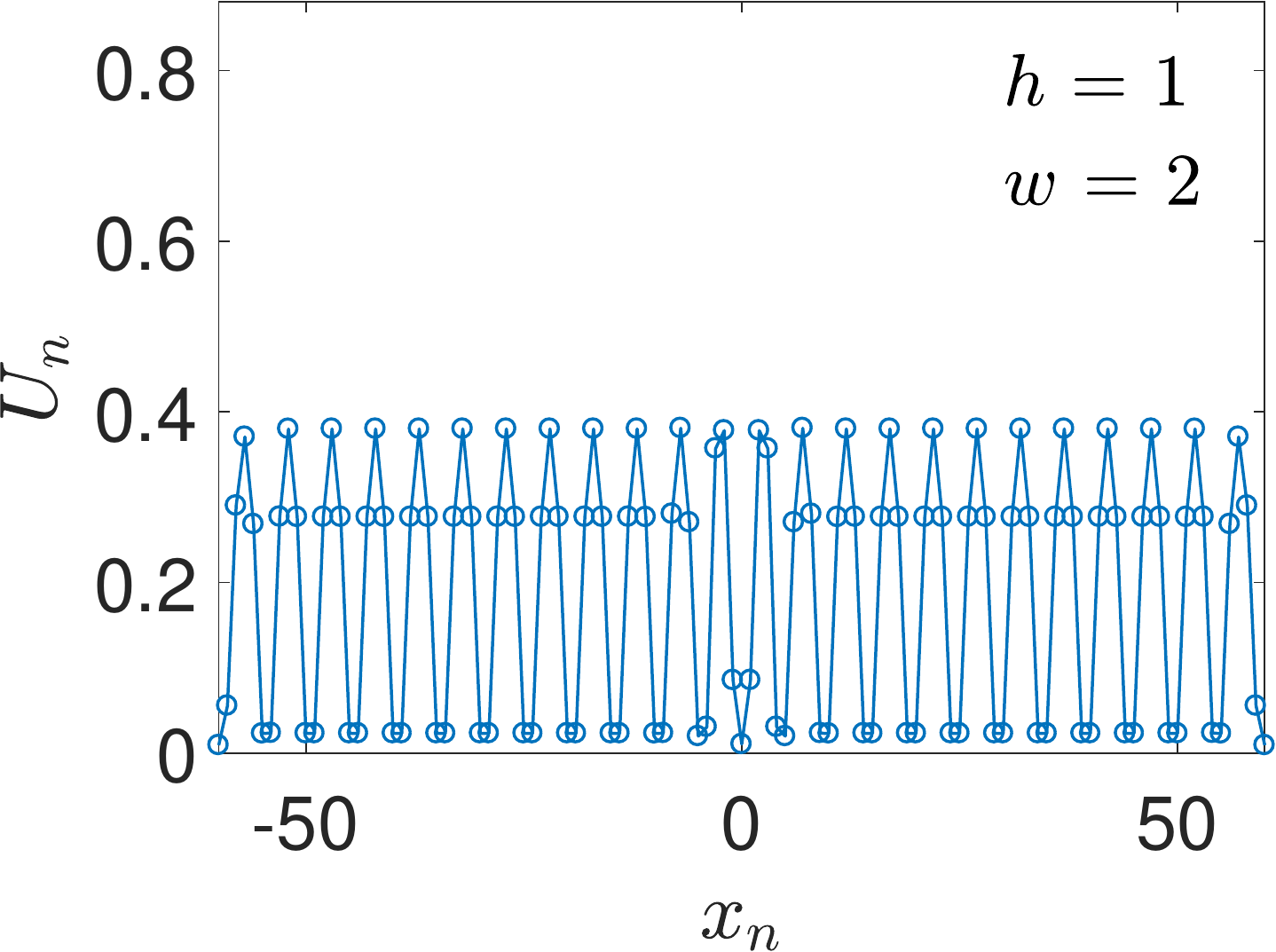}&	
		\includegraphics[scale=\figscaleee]{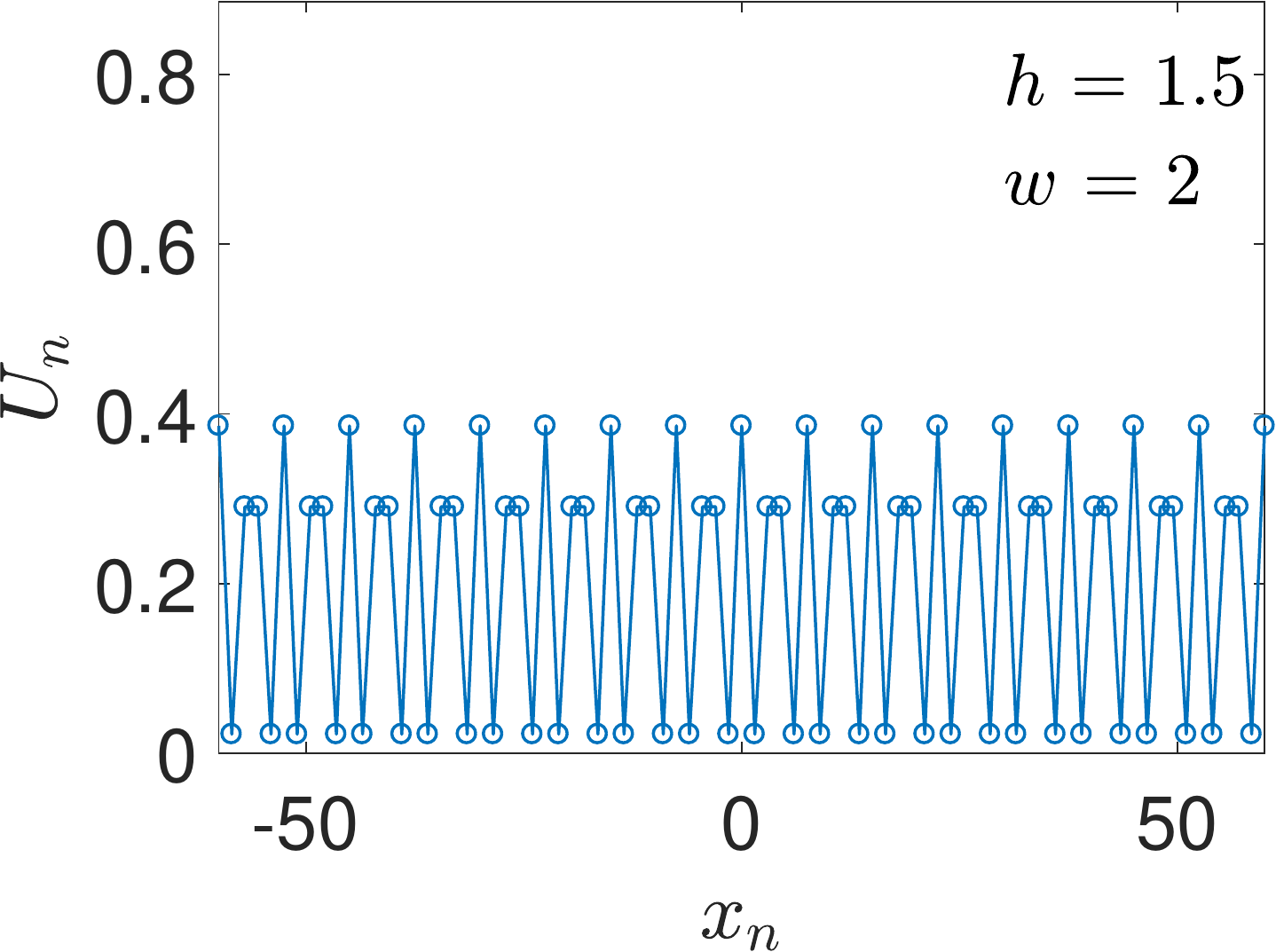}	\\ [4ex] 
		(d)&(e)&(f)\\
		\includegraphics[scale=\figscaleee]{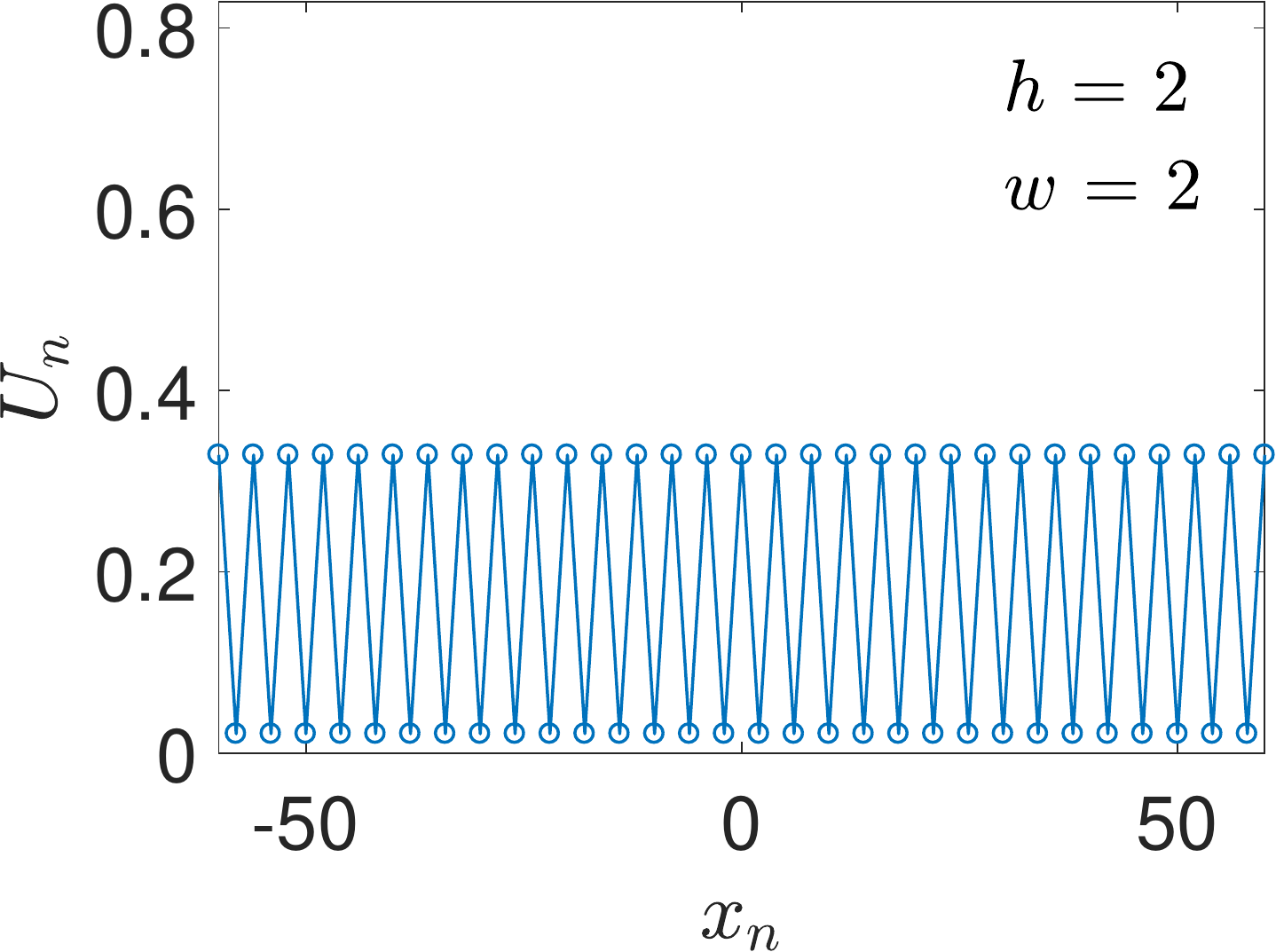}&	
		\includegraphics[scale=\figscaleee]{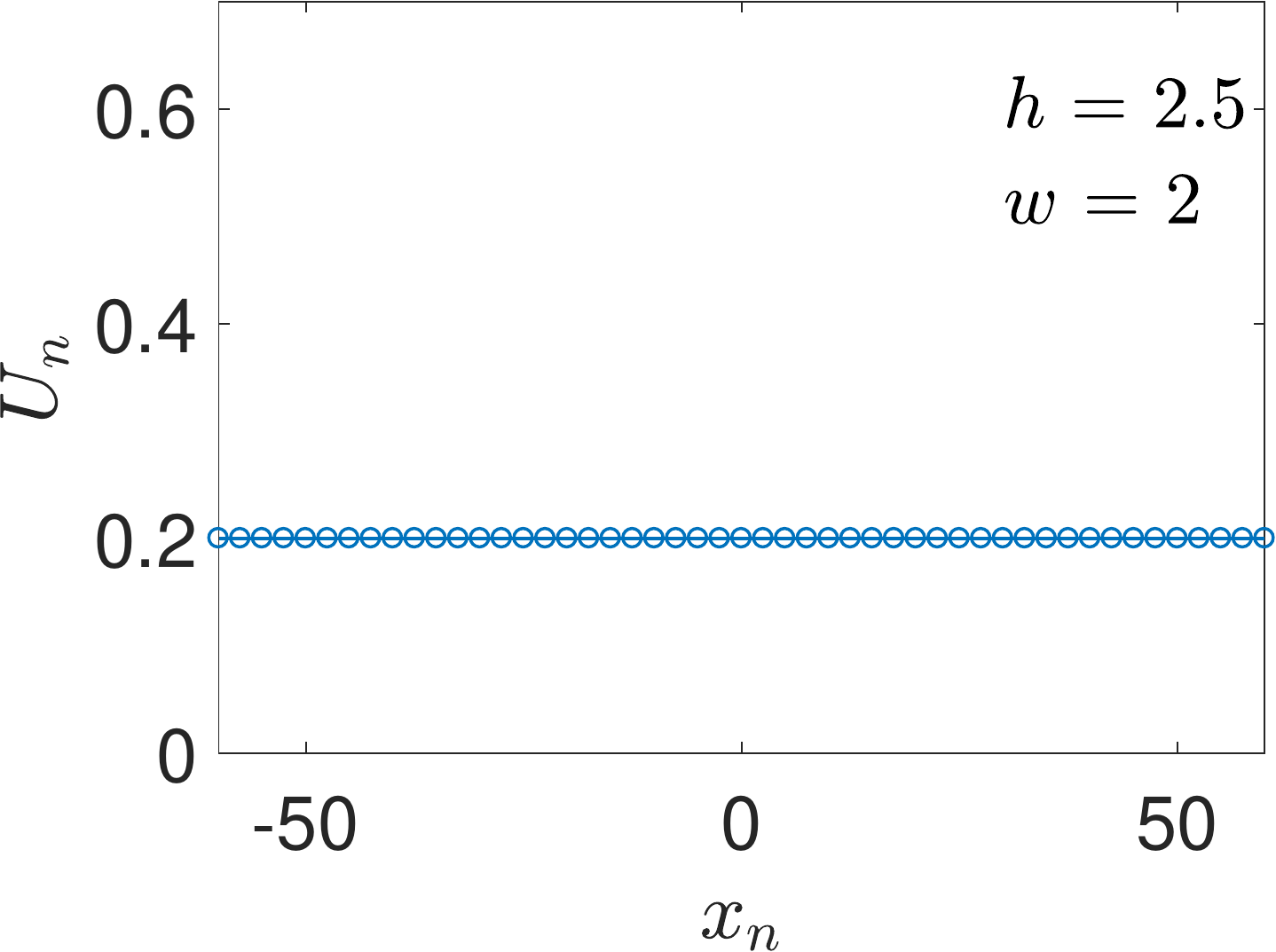}&	
		\includegraphics[scale=\figscaleee]{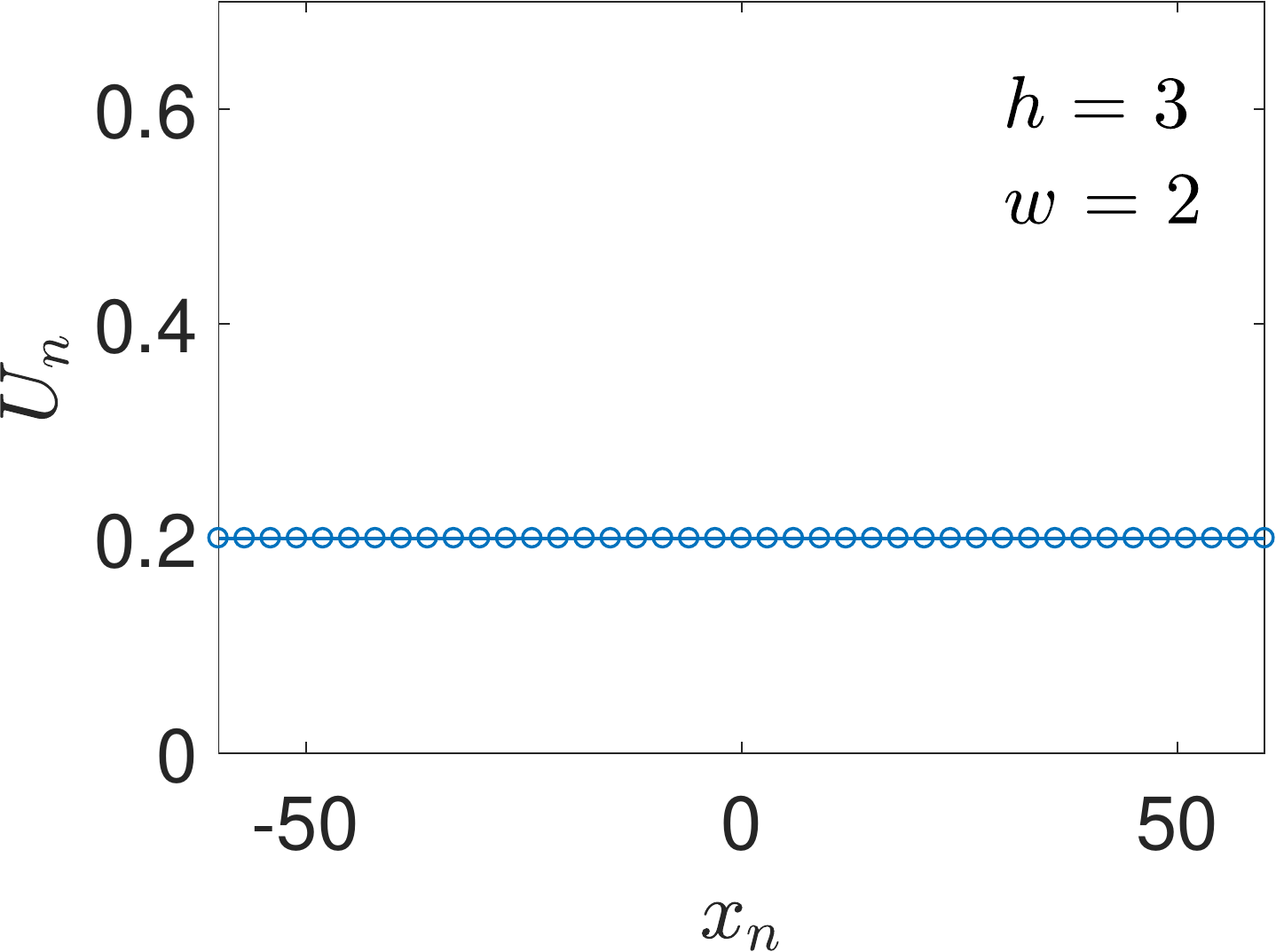}
	\end{tabular}
	\caption{Final equilibrium states of the invasion process captured in the contour-plots of Fig. \ref{fig:A2} (which is emerging from the box-profiled initial condition with a width $w=2$ and amplitude $A=0.28$), for increasing values of $h$.  Other parameters: $L=60, \gamma_{1}=0.125, \gamma_{2}= 0.5, \gamma_{3}= 0.005, \alpha=0.02, \beta=0.1$.}
	\label{fig:A1}
\end{figure}
%

An example for such an invasion process, when $\alpha= 0.02$, is captured in Figure \ref{fig:A2}, showing contour plots of the evolution of $U_n(t)$, emerging from a localized initial condition  (\ref{eq:IC1}) with amplitude $A=0.28$ and width $w=2$. 
The invasion process stops when the system attains the equilibrium solutions with profiles shown in Figure \ref{fig:A1}, for various cases of $h$,  ranging from $h=0.5$ to $h=3$. When $h=0.5$ and $h=1.5$, none of the uniform steady states is stable, since both of these values are less than $h_c=2.236$. For $h=0.5$, the system is closer to the continuous limit regime: the initial localized concentration splits,  developing invading fronts on the top of the trivial steady-state $U_s^0=0$, towards a pattern covering a triangular region (shown in the left panel of Fig. \ref{fig:A2})). The pattern is progressively formed by a  sinusoidal profiled state depicted in the left panel (a) of Fig. \ref{fig:A1} at its final stage. For $h=1$, the system rests in a more complicated equilibrium shown in the middle panel (b); a localized dip is formed at the center, surrounded by symmetric ``tent''-like spiking periodic structures. For $h=1.5$, and $h=2$, the invasion dynamics yet converge to distinct profiles: When $h=1.5$, the $5$-periodic solution shown in panel (c) forms a pattern  consisting of a repeated pair of modes of different amplitudes (shown in the middle panel of Fig. \ref{fig:A2}), attained by a single and a pair of nodes within, respectively. When $h=2$,  the  $2$-periodic state shown in panel (d) is achieved forming a pattern of repeated spiking modes (similar to the one portrayed in Fig. \ref{fig:A3} (c)). Crossing the critical value $h>h_c$, the spatially uniform state $U^+_s=0.2$ becomes asymptotically stable, in full agreement with the analytical predictions; it is the attained equilibrium shown in panels (e) and (f), corresponding to the cases $h=2.5$ and $h=3$, respectively, after the development of the uniform triangular pattern shown in the right panel of Fig. \ref{fig:A2}.
%

\begin{figure}[th!]
	\centering 
	\begin{tabular}{ccc}
		(a)&(b)&(c)\\	
		\includegraphics[scale=\figscaleee]{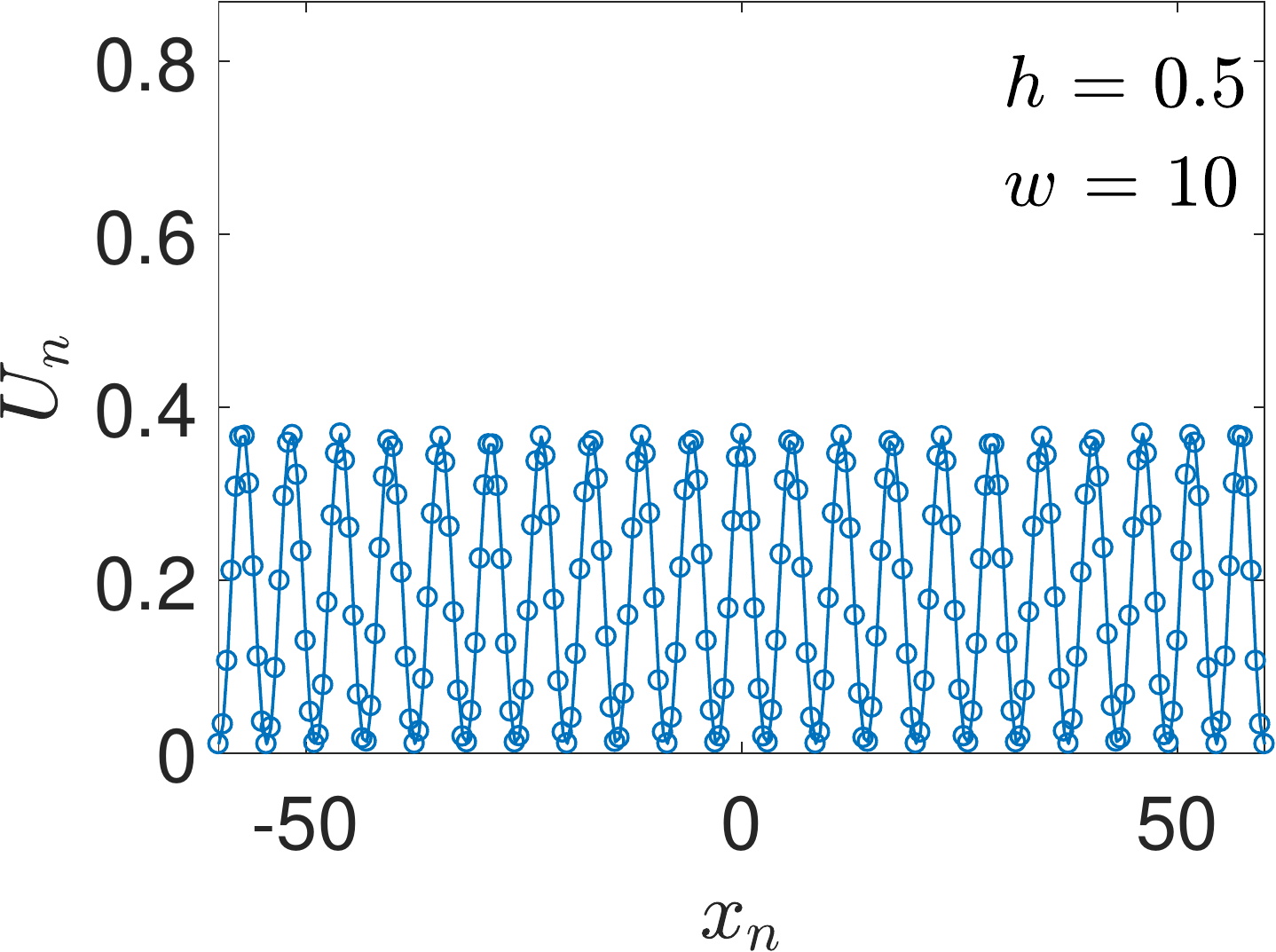}&
		\includegraphics[scale=\figscaleee]{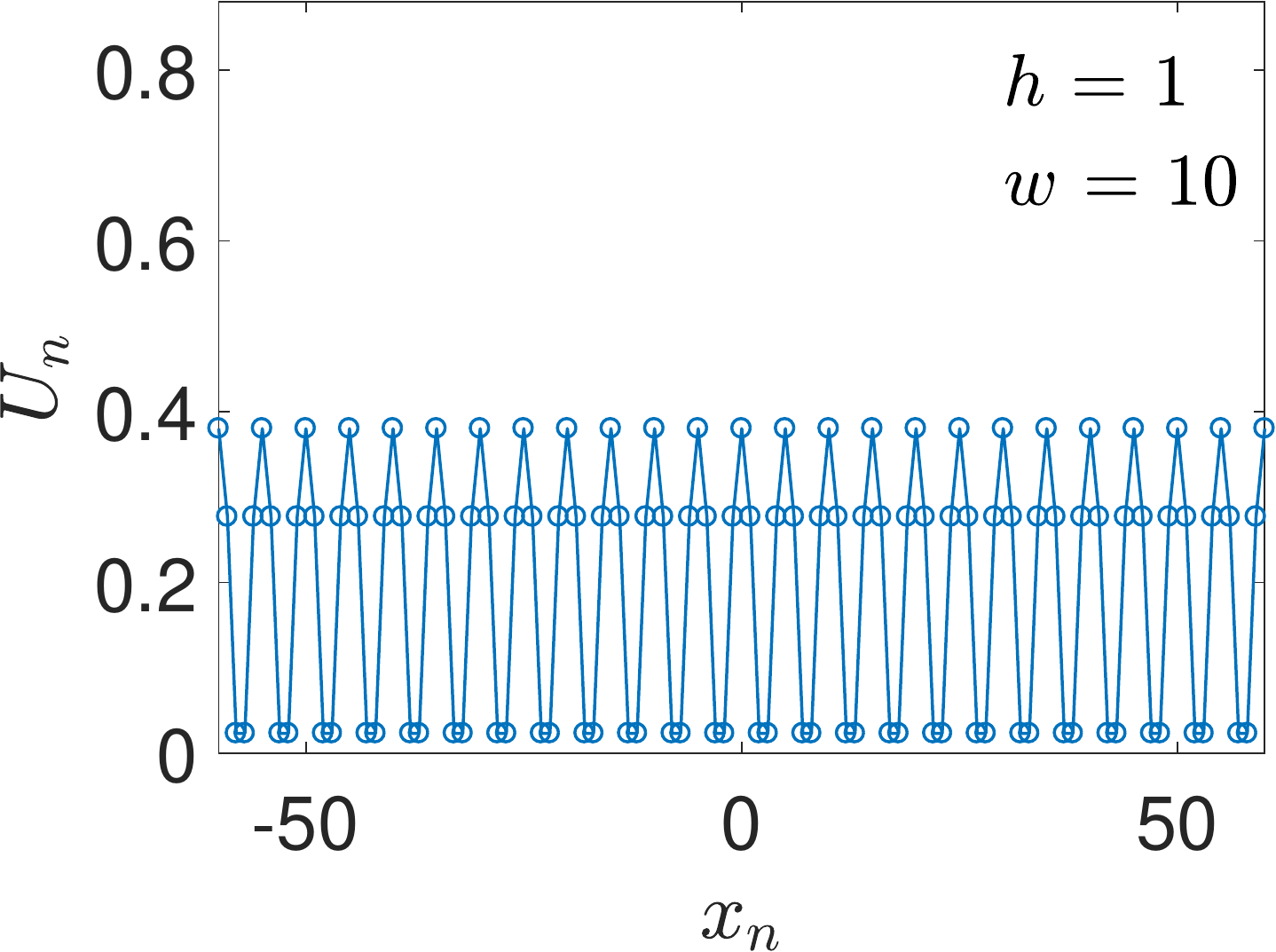}&
		\includegraphics[scale=\figscaleee]{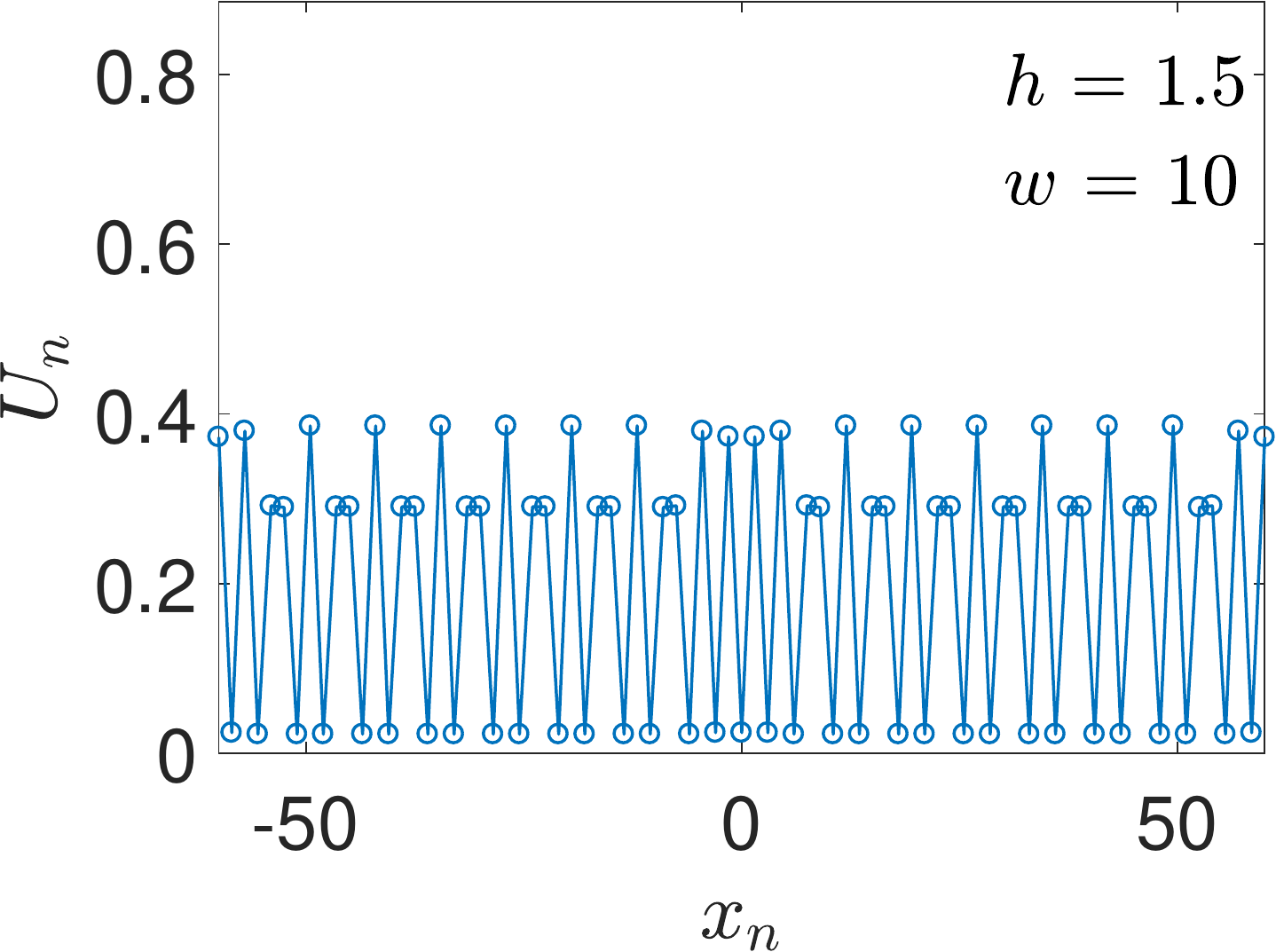}
	\end{tabular}
	\caption{Equilibrium states for the dynamics of the box-profiled initial condition of width $w=10$ and amplitude $A=0.28$, for increasing values of $h$.
		Other parameters: $L=60, \gamma_{1}=0.125, \gamma_{2}= 0.5, \gamma_{3}= 0.005, \alpha=0.02, \beta=0.1$.}
	\label{fig:C}
\end{figure}

Another interesting effect is the impact of the width of  the localized initial data in the resulting dynamics. Figure \ref{fig:C} shows representative final states of the invasion process varying $h$, when the width of the initial condition \eqref{eq:IC1} of the same amplitude $A=0.28$ is increased to $w=10$. For $h=0.5$, the final state shown in panel (a) is a phase-shift of the same sinusoidal-alike equilibrium achieved for $w=2$. Nevertheless, different states are captured for $h<h_c$ in the intermediate discrete regime, as shown in panels (b) and (c), corresponding to $h=1$ and $h=1.5$ respectively. When $h=1$, the equilibrium is made of the ``tent''-alike spikes observed in Fig. \ref{fig:A1} (b), but without the centered dip. When $h=1.5$ the repeated pair of modes observed in Fig. \ref{fig:A1} (c) are separated by the group of the four spiking modes around the centered node $U_0=0$. 
\begin{figure}[h!]
	\centering 
	\begin{tabular}{ccc}
		(a)&(b)&(c)\\	
		\includegraphics[scale=\figscaleee]{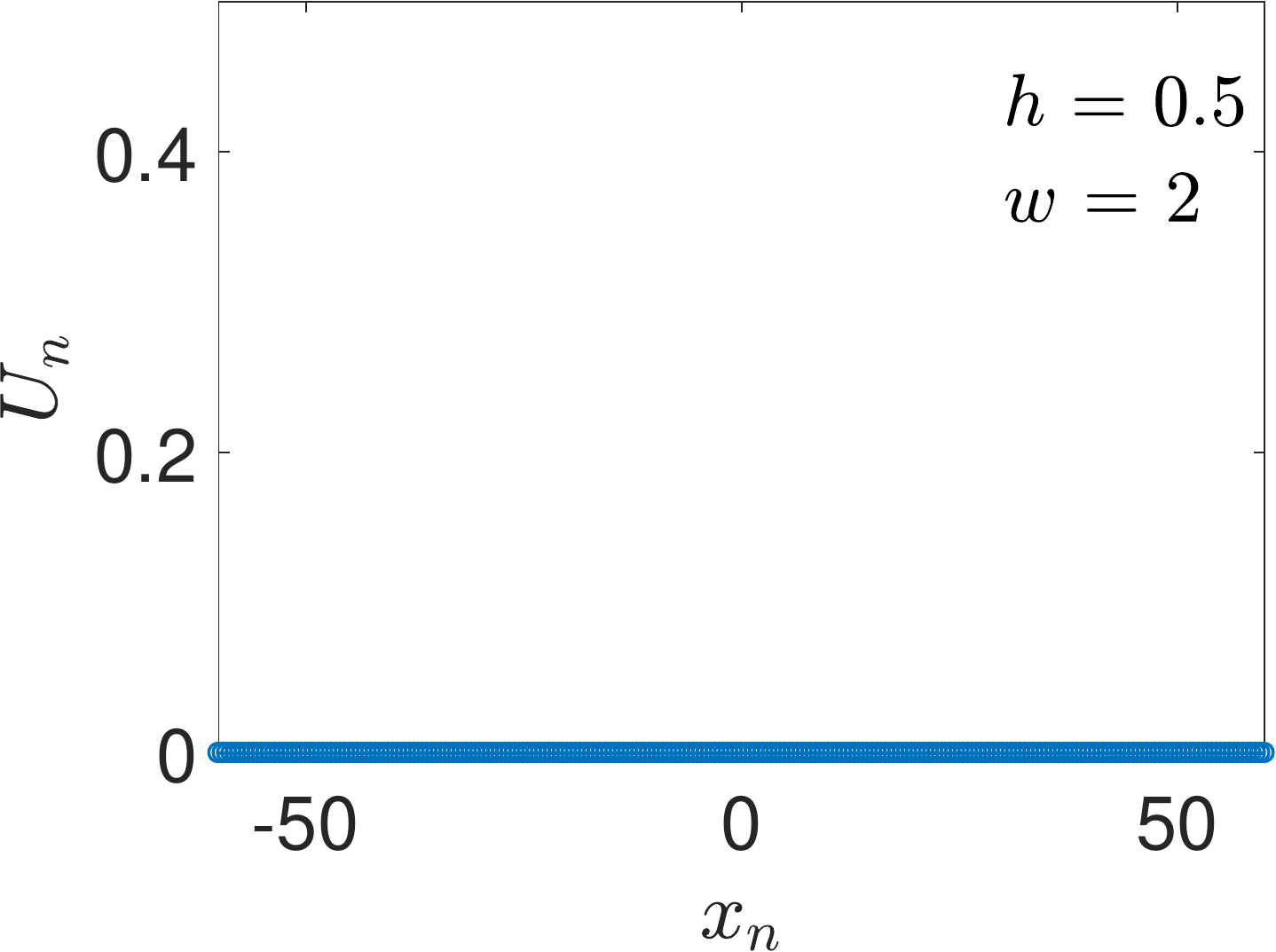}&
		\includegraphics[scale=\figscaleee]{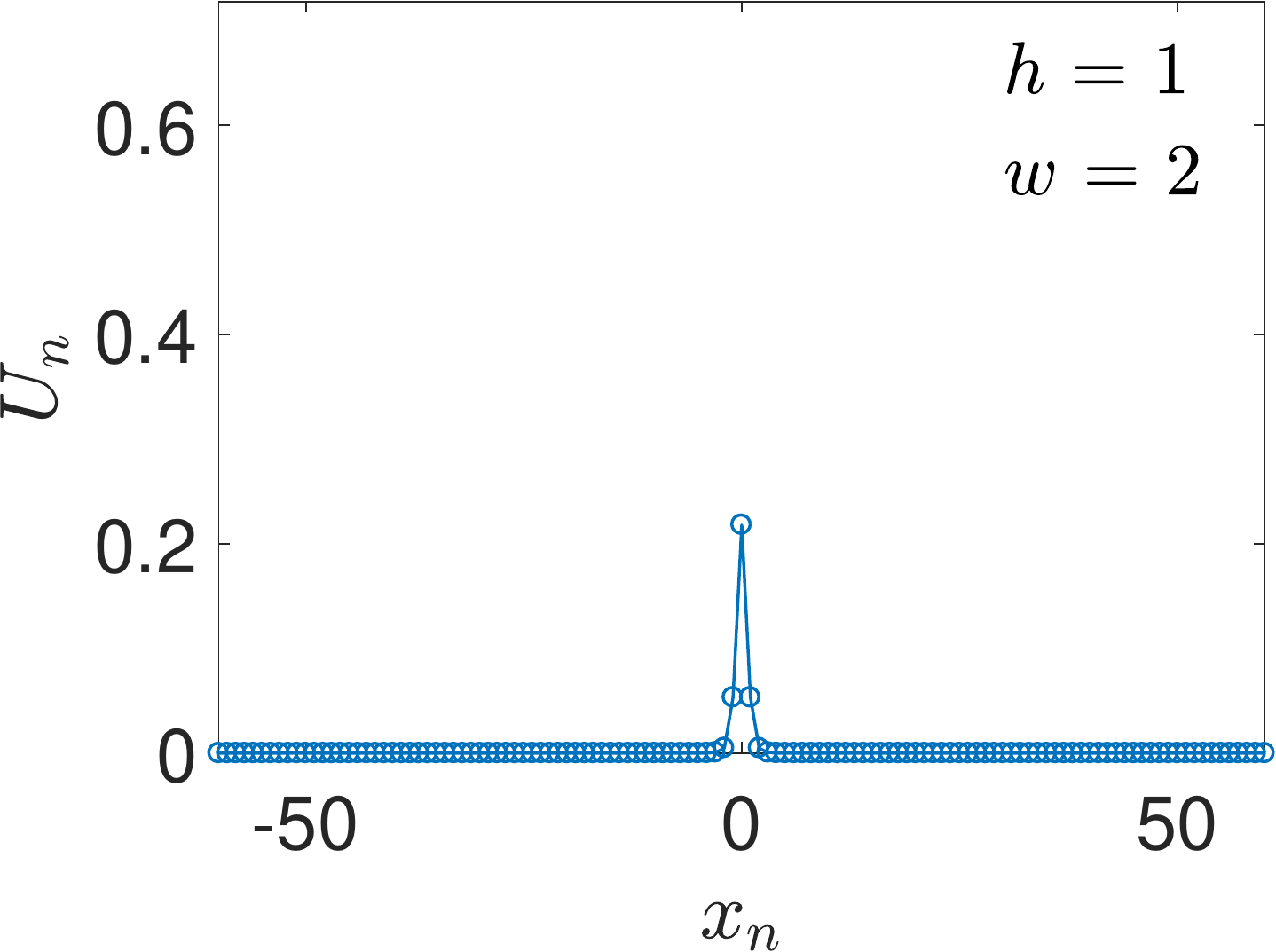}&
		\includegraphics[scale=\figscaleee]{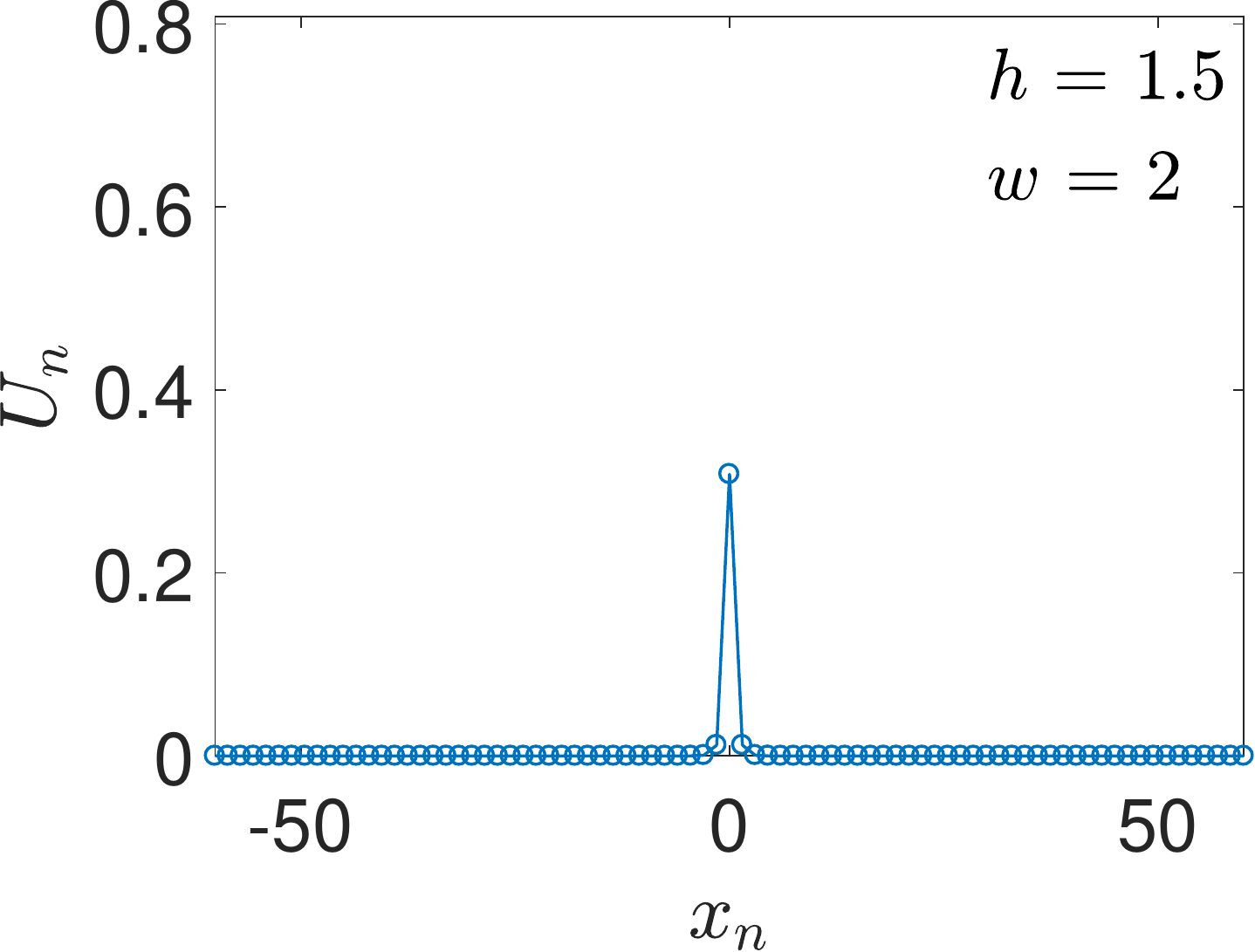}
		\\ [4ex] 
		(d)&(e)&(f)\\		
		\includegraphics[scale=\figscaleee]{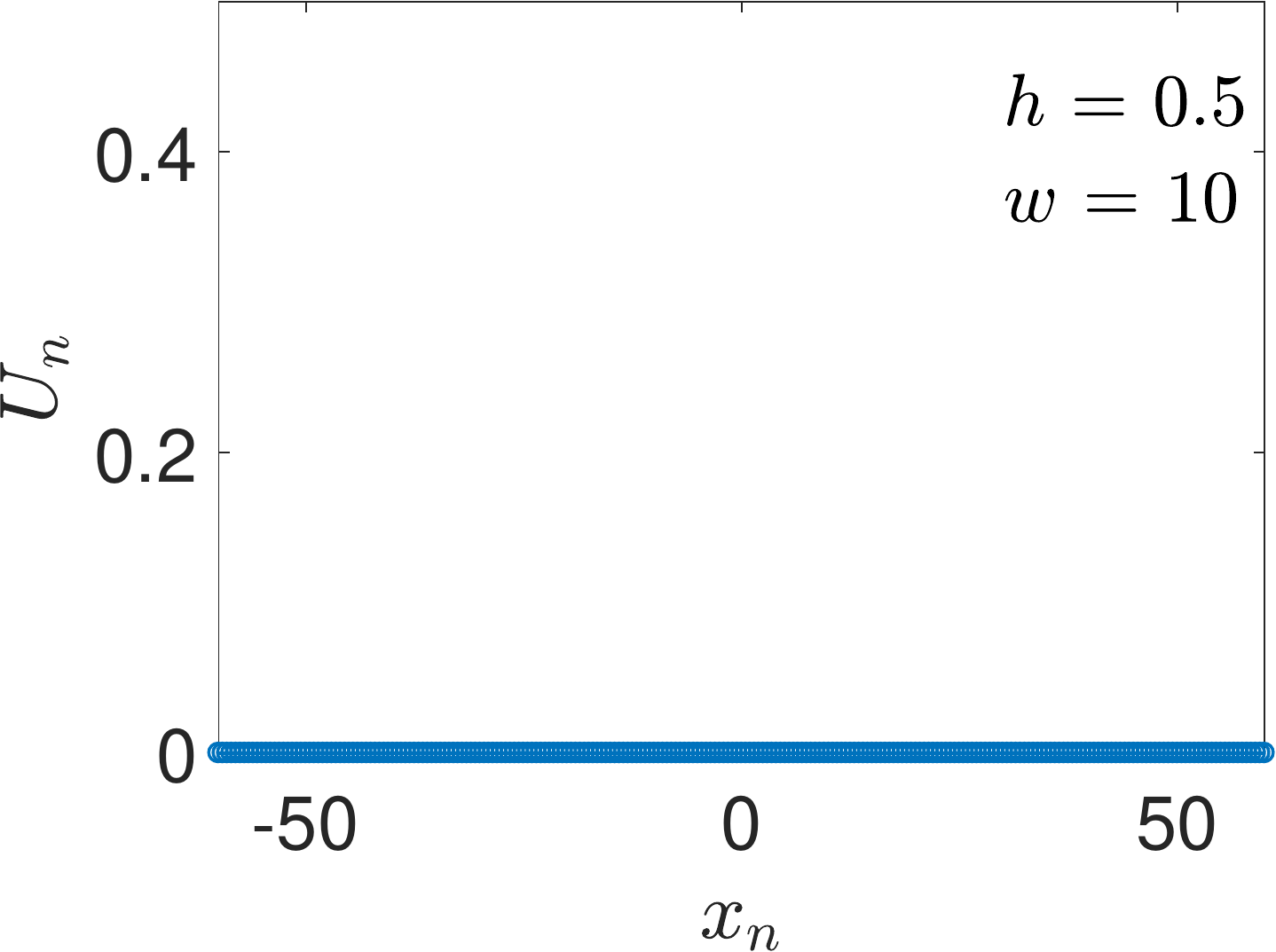}&
		\includegraphics[scale=\figscaleee]{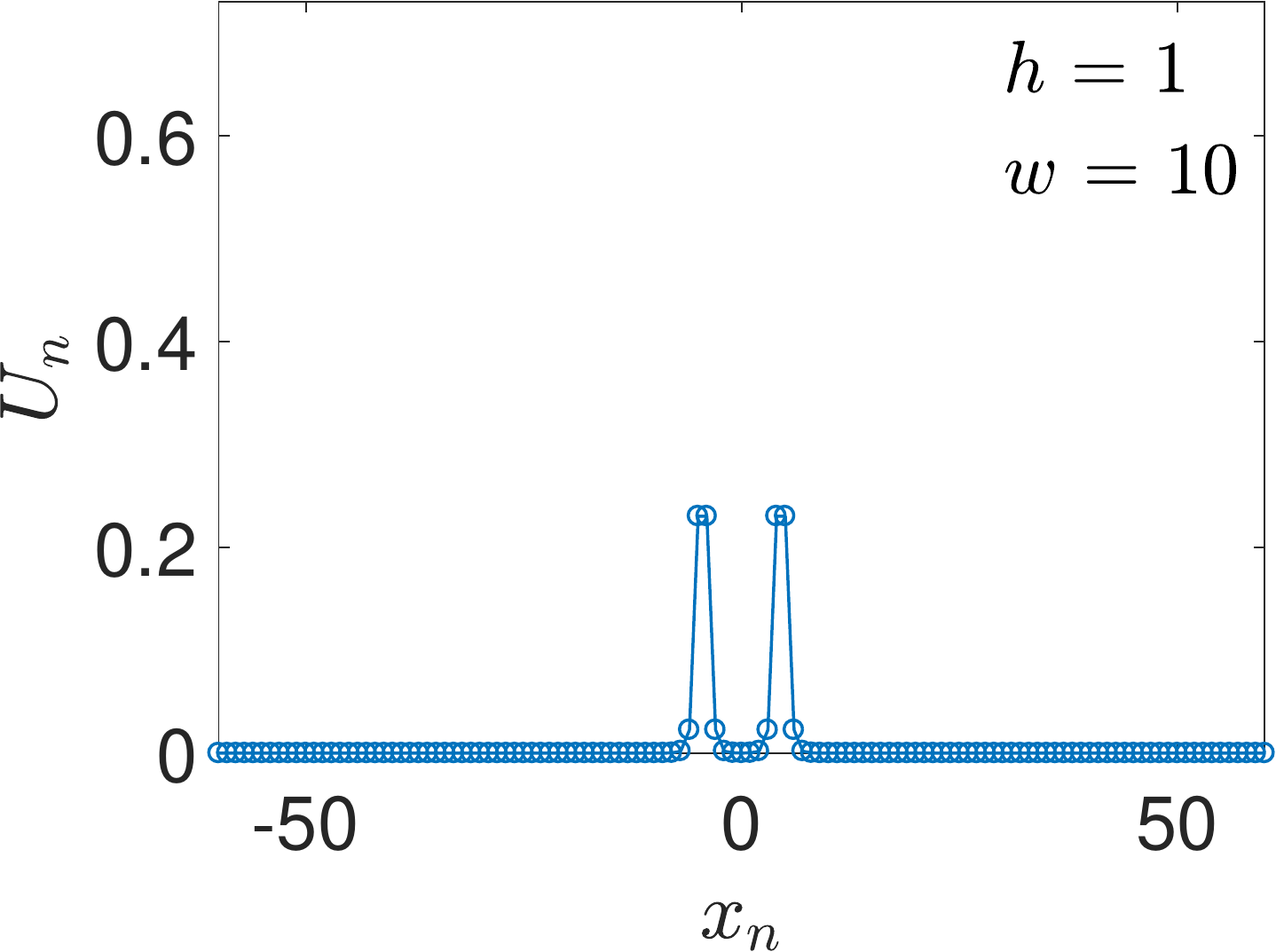}&
		\includegraphics[scale=\figscaleee]{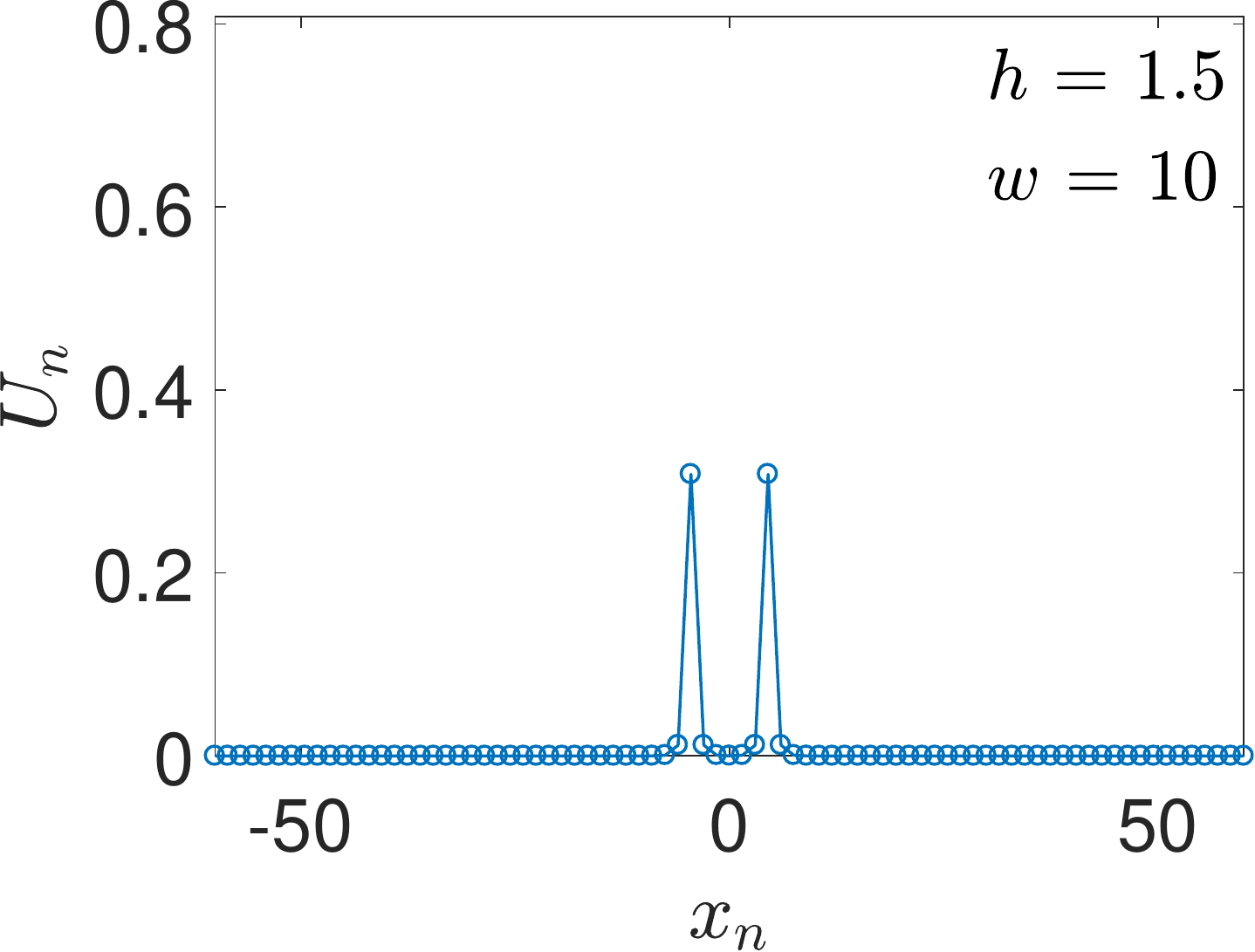}
	\end{tabular}
	\caption{Equilibrium states for the dynamics of box-profiled initial conditions with amplitude $A=0.28$, for increasing values of $h$. Top row: with a width $w=2$. 
		Bottom row: with a width $w=10$. Other parameters for both cases: $L=60, \gamma_{1}=0.125, \gamma_{2}= 0.5, \gamma_{3}= 0.005, \alpha=-0.02, \beta=0.1$.}
	\label{fig:A4}
\end{figure}

We conclude the presentation of the  numerical results with two studies in the so-called, low-productivity regime $\alpha<0$, when in addition $\beta>0$. 
In this regime, the simulations are performed for values of $\alpha$ for which the uncoupled system possesses only the trivial state, 
that is the case for $\alpha <-\beta^2/4:=\alpha_F<0$.   However, when coupling is present the dynamics is much more intrigue: while small values of $h$ still may drive the system towards the trivial  steady-state, larger values of $h$ in the discrete regime, may allow for the formation of spatially localized or periodic equilibria.  
\begin{figure}[tbh!]
	\centering 
	\includegraphics[scale=0.36]{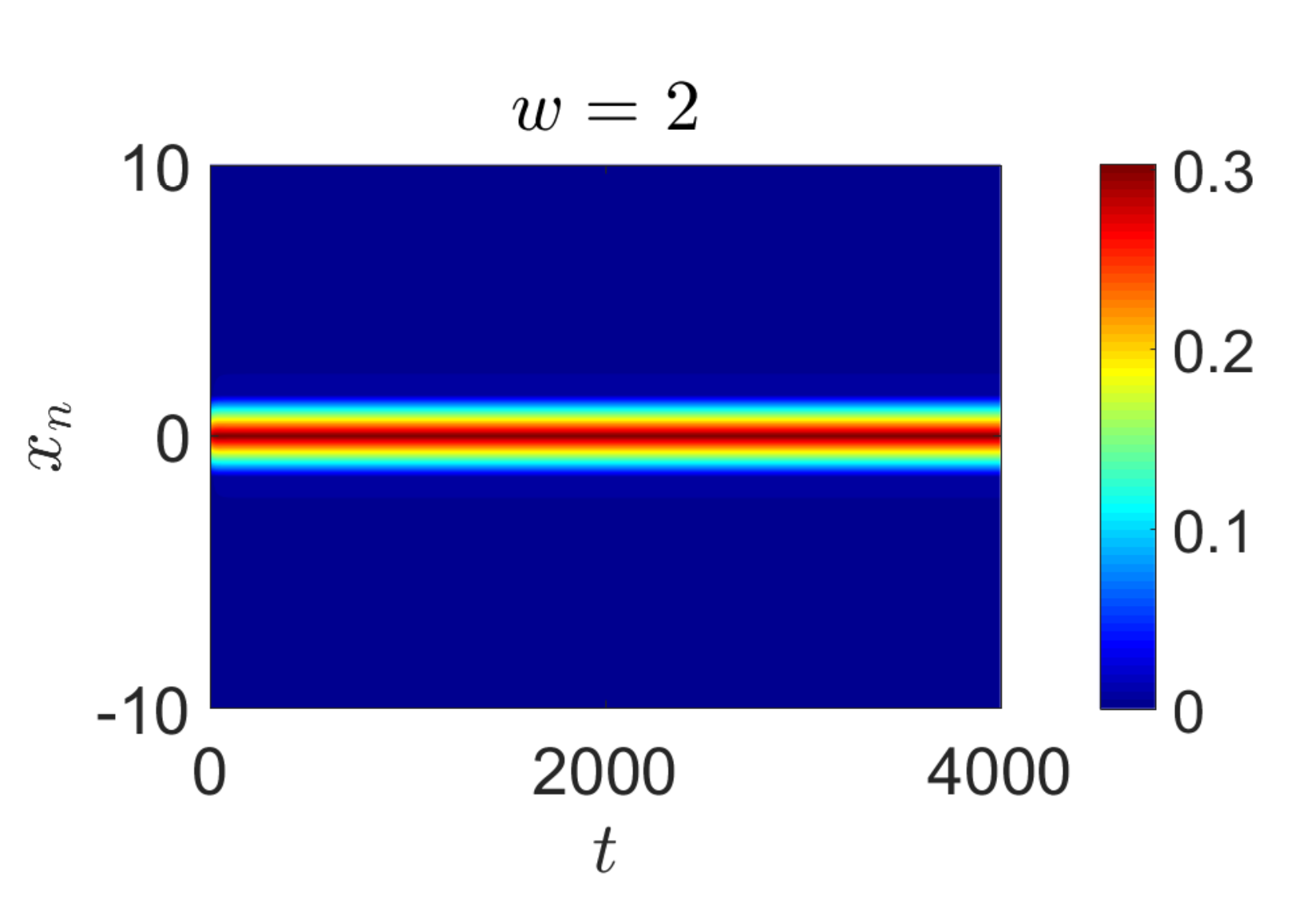}
	\includegraphics[scale=0.36]{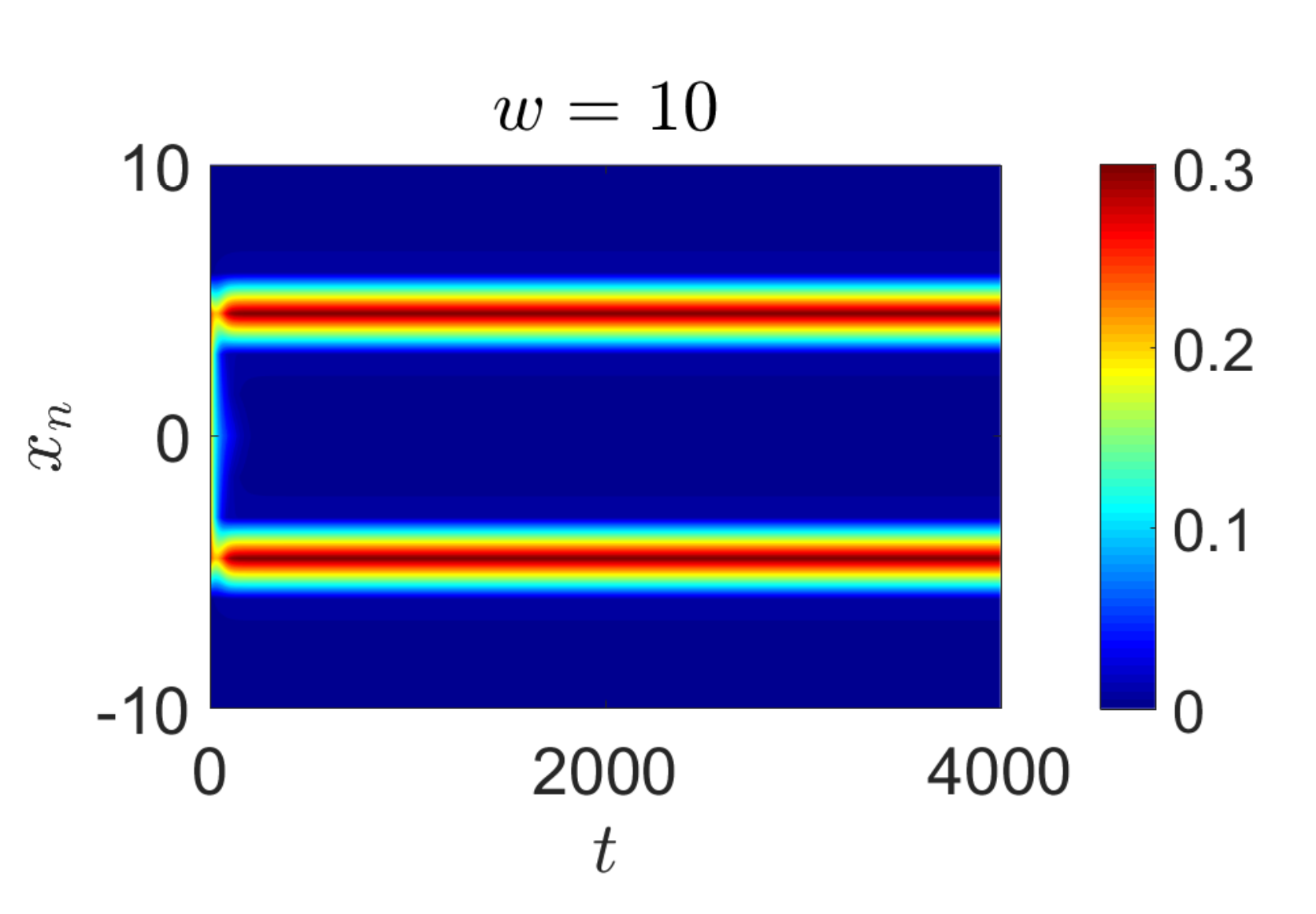}
	\caption{Contour plots of the time-evolution of the box-profiled initial condition of amplitude $A=0.28$. Widths: $w=2$ (left panel) and $w=10$ (right panel). Other parameters: $h=1.5, L=60, \gamma_{1}=0.125, \gamma_{2}= 0.5, \gamma_{3}= 0.005, \alpha=-0.02, \beta=0.1$.}
	\label{fig:a<0}
\end{figure}
 Figure \ref{fig:A4} shows examples of the steady-states, when $\alpha=-0.02$, where the dynamics of box initial conditions of amplitude $A=0.28$ converge for two cases of its width, varying $h$. The upper row of Fig.\ref{fig:A4} shows the steady states when $w=2$. For $h=0.5$ the system converges to the trivial steady state $U^0_s=0$ shown in panel (a).  For larger values of $h$, spatially localized steady-states emerge. These have the form of centered spikes when $h=1$ and $h=1.5$, depicted in panels (b) and (c), respectively; we observe that as $h$ is increasing, their amplitude is  also increasing.  The bottom row of  Fig.\ref{fig:A4} shows the steady states when $w=10$. For $h=0.5$ the system again converges to $U^0_s=0$, but for $h=1$ and $h=1.5$, pairs of twin-spikes are formed; again their amplitude increases with $h$.  The evolution of $U_n(t)$ in both cases of $w$ for $h=1.5$ is portrayed in the contour-plots of Fig. \ref{fig:a<0}. In the case $w=2$, the box localized solution changes its shape to the single-spike (left panel), while in the case $w=10$, it rapidly splits to the twin-spike mode (right panel).
	
\begin{figure}[th!]
	\centering 
	\begin{tabular}{ccc}
		(a)&(b)&(c)\\
		\includegraphics[scale=0.36]{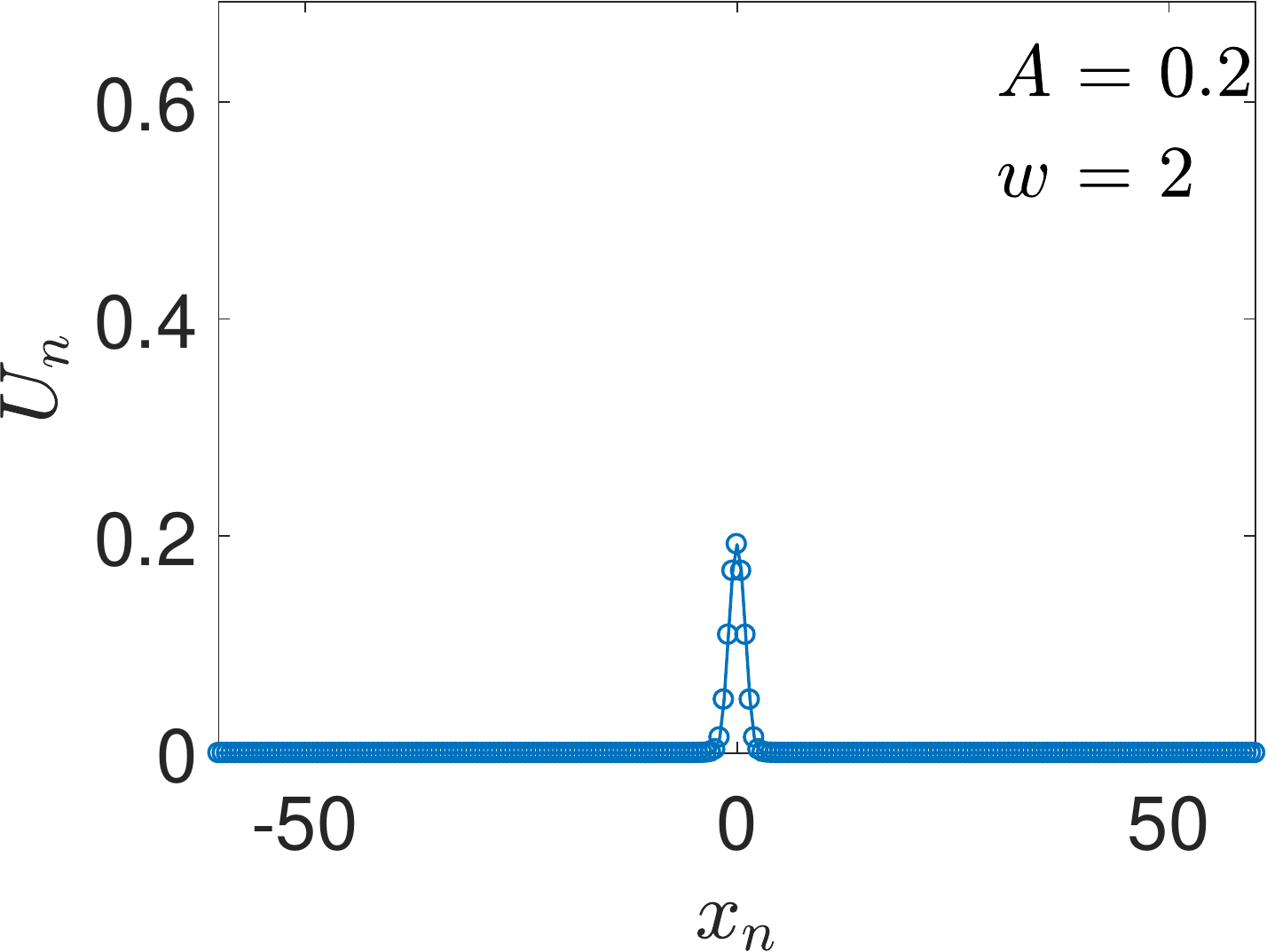}&
		\includegraphics[scale=0.36]{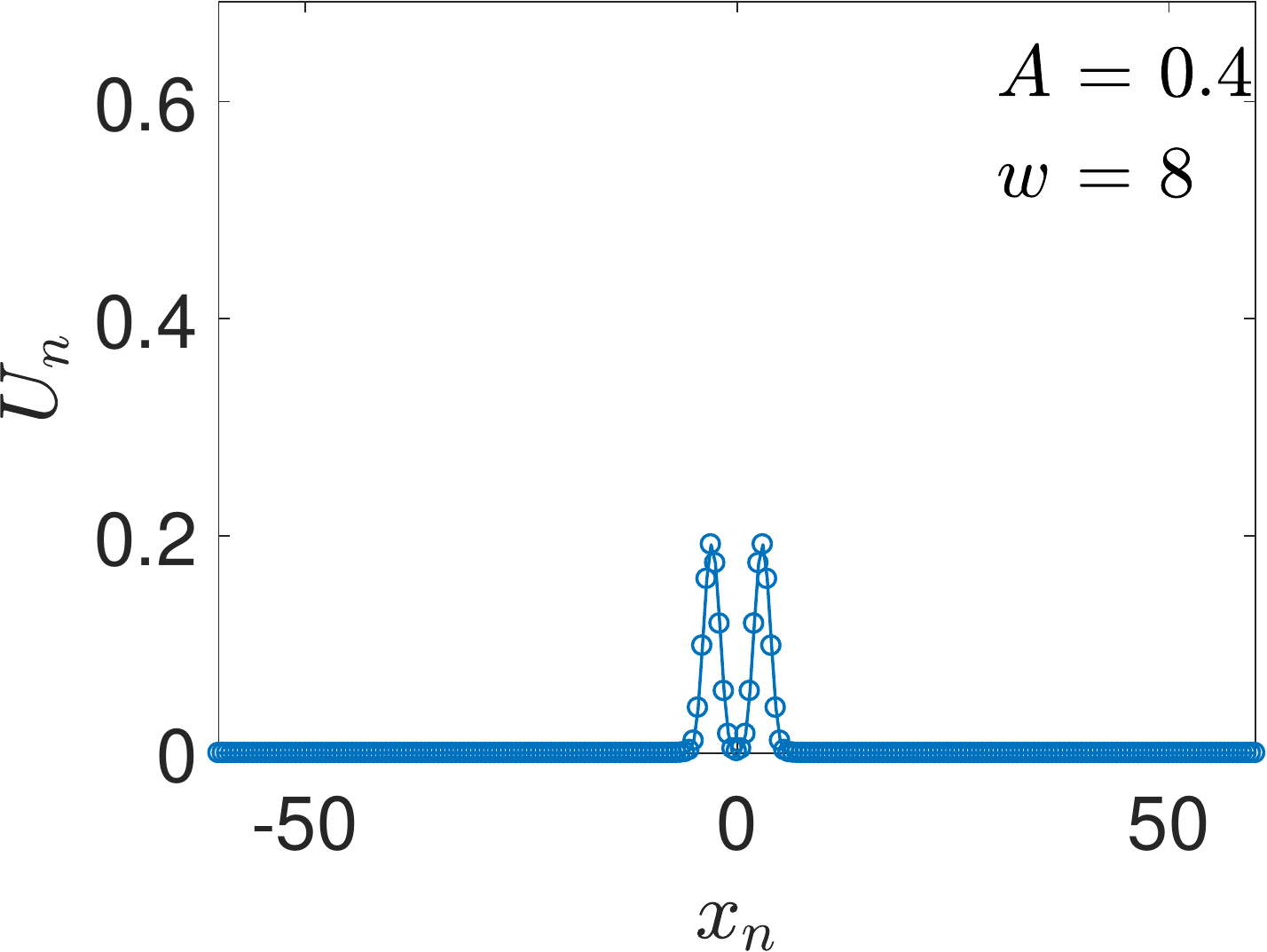}&
		\includegraphics[scale=0.36]{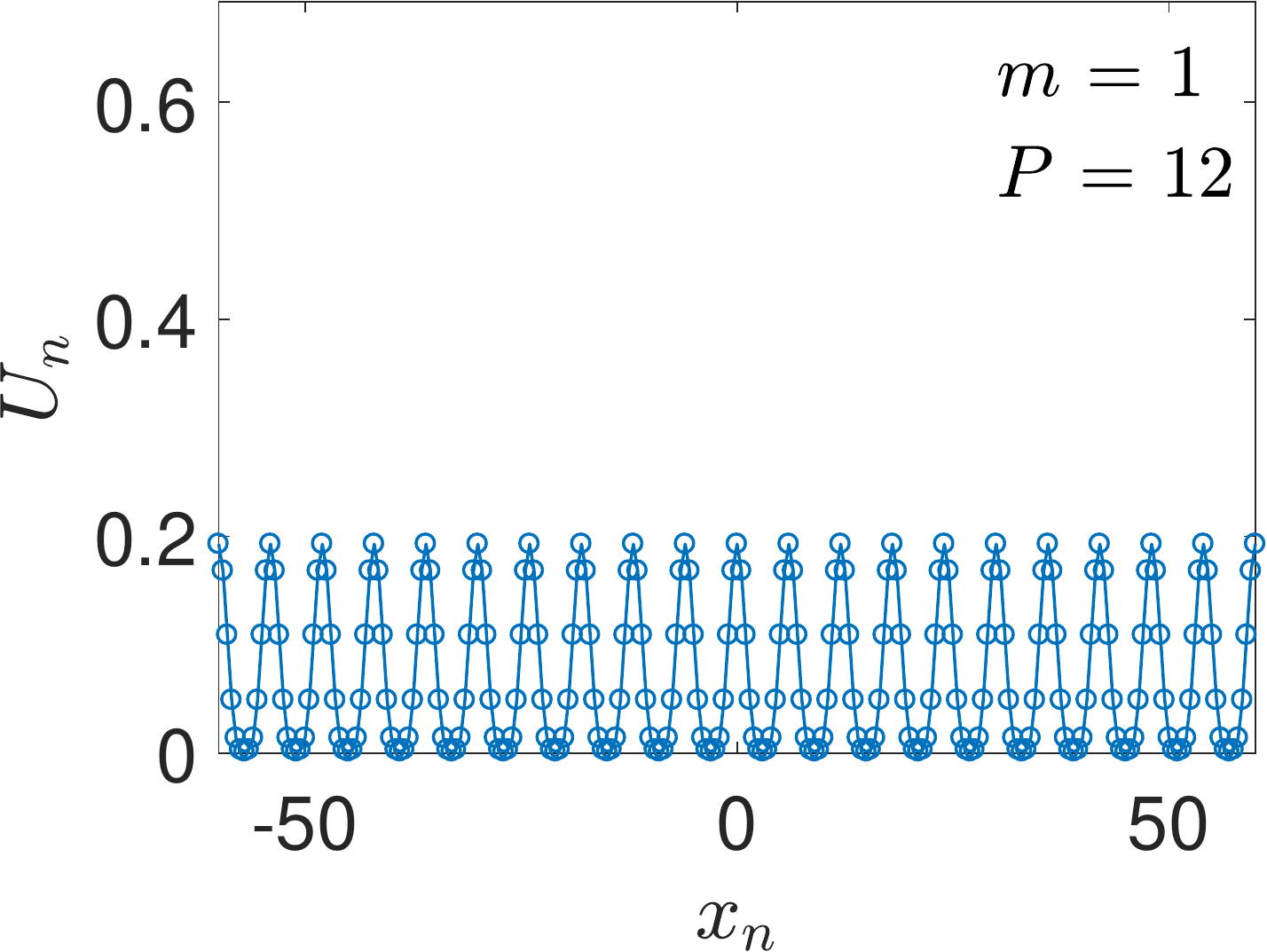}
		\\[4ex]
		(d)&(e)&(f)\\
		\includegraphics[scale=0.36]{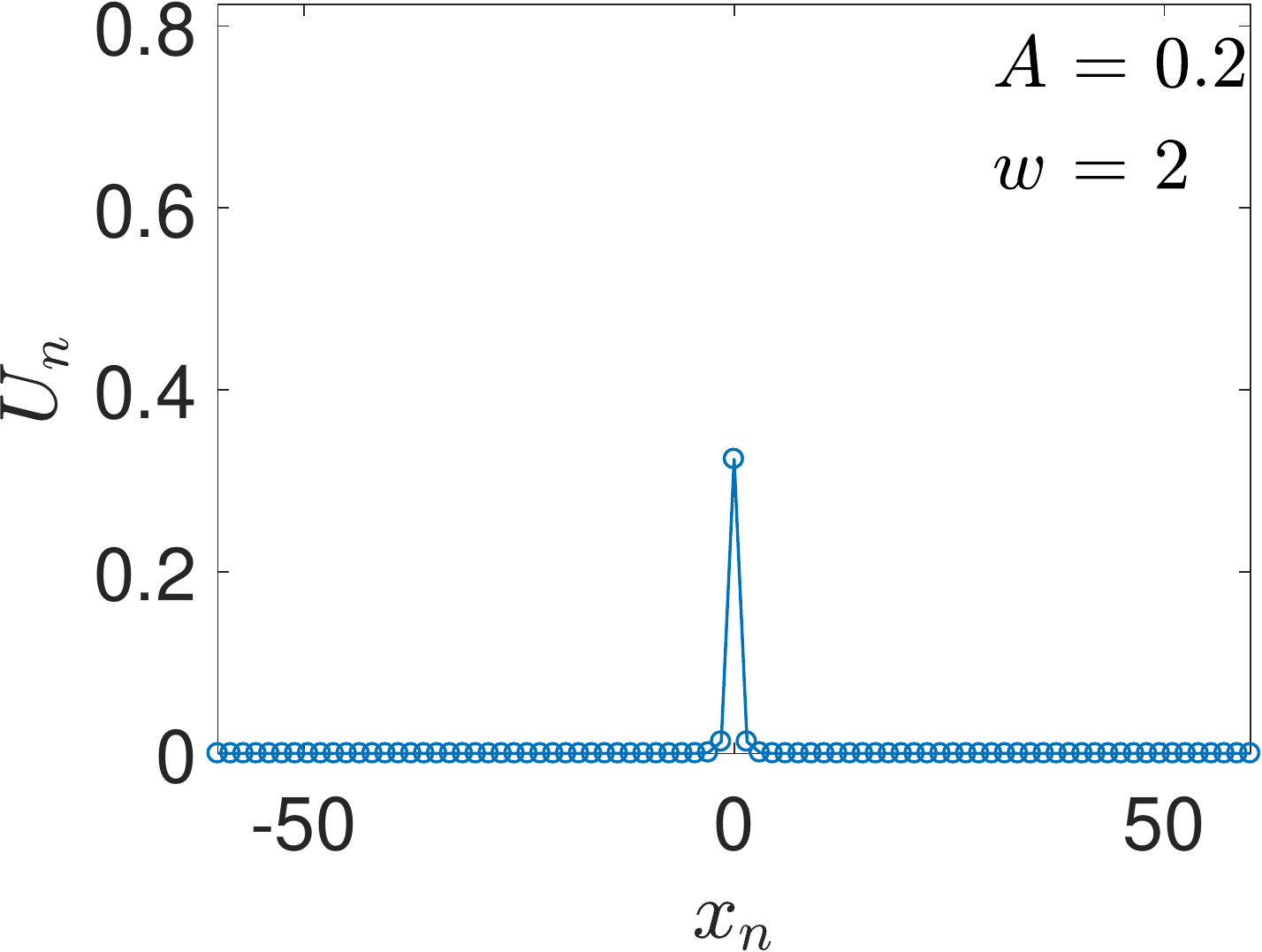}&
		\includegraphics[scale=0.36]{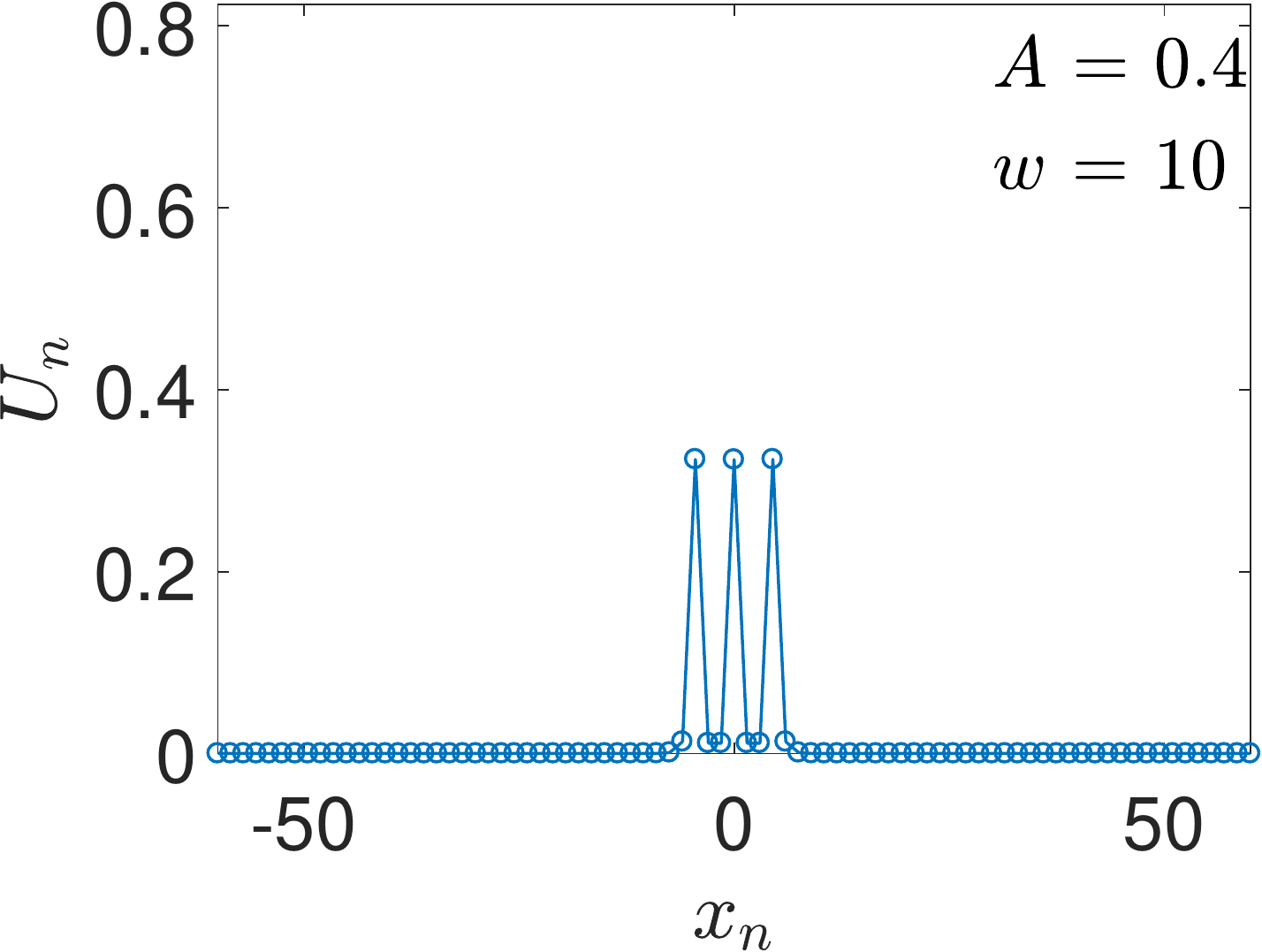}&
		\includegraphics[scale=0.36]{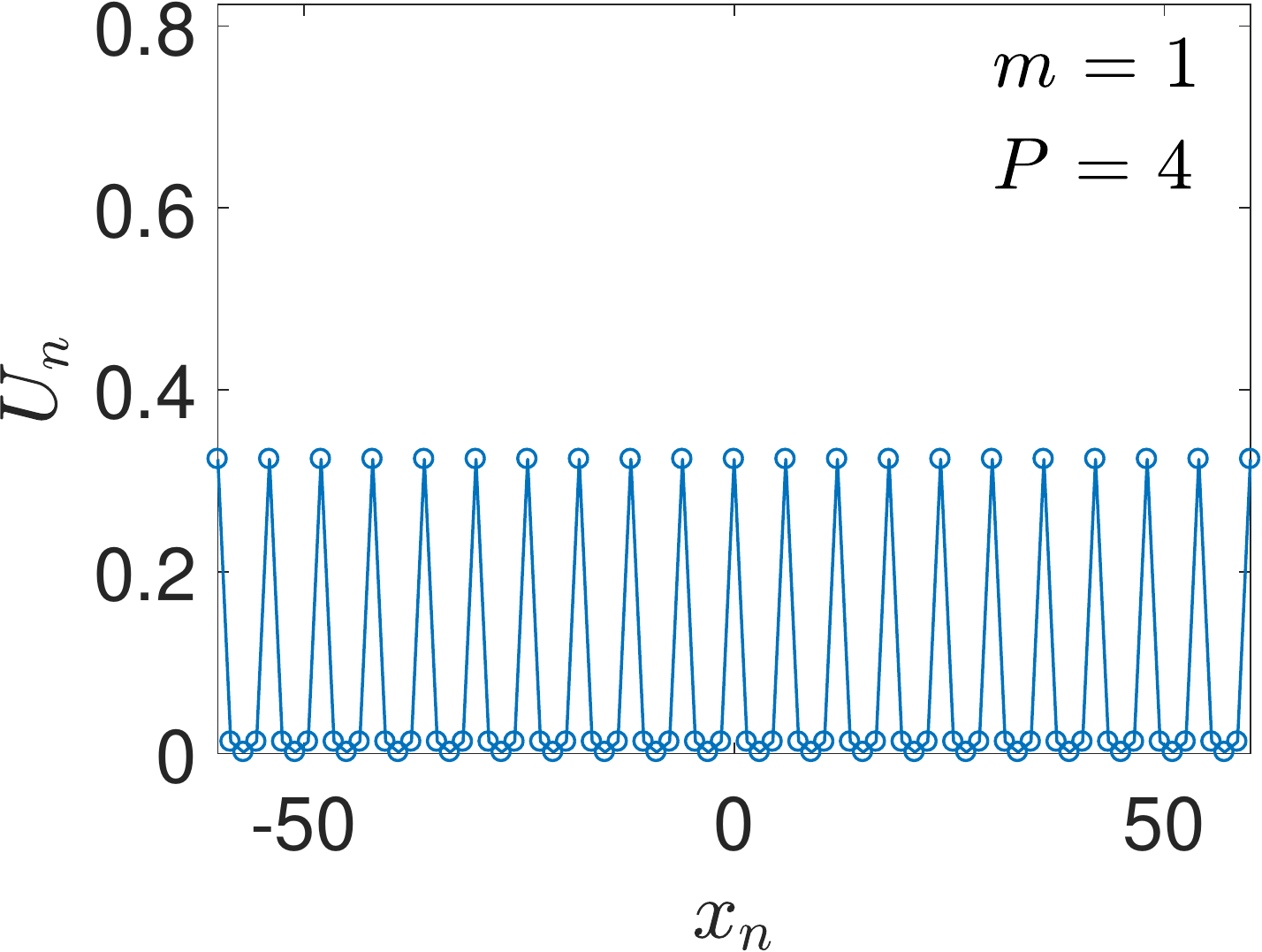}	
		\end{tabular}	
\caption{Top row, $h=0.5$: Equilibrium states for the dynamics of (a) box-profiled initial condition \eqref{eq:IC2} with width $w=2$ and amplitude $A=0.2$ 
(b) box-profiled initial condition with width $w=8$ and amplitude $A=0.4$,  (c) spatially extended initial condition with $A=B=0.2$ and $k=\pi/3$. 
Bottom row, $h=1.5$: Equilibrium states for the dynamics of (a) box-profiled initial condition \eqref{eq:IC2} with width $w=2$ and amplitude $A=0.2$ , 
(b) box-profiled initial condition with width $w=10$ and amplitude $A=0.4$, 
(c) spatially extended initial condition with $A=B=0.2$ and $k=\pi/3$.  
Other parameters: $L=60, \gamma_{1}=0.125, \gamma_{2}= 0.5, \gamma_{3}= 0.005, \alpha=-0.016, \beta=0.1$.}
	\label{fig:B}
\end{figure}
We remark that multi-spikes may also appear, as the width of the localized initial data is further increased. 
Also, spatially periodic solutions can be observed for plane wave initial data of certain frequencies.
Similar behavior is possible when the spatially continuous system is approached (smaller $h$) given that  $\alpha$ is decreased.
Examples are shown in Figure \ref{fig:B} corresponding to the case $\alpha=-0.016$. 
The top row shows the equilibrium states for $h=0.5$. 
The left panel (a) shows a centered localized equilibrium mode attained by the dynamics of a box-profiled initial condition \eqref{eq:IC2} of amplitude $A=0.2$ and width $w=2$. 
The middle panel (b) shows a double spike for  box-profiled initial condition with $A=0.4$ and $w=4$,
while the right panel (c) illustrates  long term dynamics of a spatially periodic initial condition \eqref{eq:IC1} for amplitudes $A=0.2$, $B=0.2$, and frequency $k=\pi/3$.
It is interesting to observe that the periodic steady state of panel (c) consists of multiple copies of the single localized state of panel (a). 
The same effect is verified in the case $h=1.5$. In particular, the equilibria for the  box-profiled initial condition (single spike) and for the spatially periodic initial condition
 are shown in panels (d) and (f), respectively.  
 In Figure \ref{fig:B} panel (e), a localized three-spike solution arises for box-profiled initial condition with $A=0.4$ and $w=10$. 
Similarly to what we observed in Fig. \ref{fig:A4}, for increased $h$, we notice an increase of the amplitude of the equilibrium states. 
Furthermore, in accordance to the linear stability analysis the non-uniform states disappear when $h>h_c$, and the dynamics converge to spatially uniform states.
%

The behavior described above can be roughly understood by the impact $h$ has on the coupling, 
and how this in turn affects the dynamics of the system. 
Since parameter $h$ is negatively associated to the coupling, for high values of $h$ (ultra-discrete system),
both the competition and facilitation mechanisms are negligible and each patch tends to grow independently, 
governed solely by the local dynamics induced by the non-linearity $f$.
On the contrary, approaching the spatially continuous system ($h\to 0$), 
coupling becomes stronger (non-local interactions intensify as the patches come closer to eachother).
In the latter case and for very low productivity, the patches die out as the local growth
is insufficient to support vegetation due to the weak increase in local density. 
This is also the case,  for larger $\alpha<0$ - values (higher productivity) in the absence of competition or strong supremacy of facilitation, which corresponds to large
$\gamma_3$-values. However, increasing competition compared to facilitation, that is $\gamma_3$ is small, can benefit persistence in the low productivity regime and shrink the extinction regime. 
This can be explained by the fact that some individuals compete effectively and win over the underdeveloped individuals 
by exploiting the resources of the surrounding area. 
In the same productivity regime, but for intermediate $h$-values the effect of both mechanisms is moderate,
with competition being weaker than facilitation.  
The balanced resource utilization and spatial development may lead to the expansion of the persistence range along $\alpha<0$. 
Finally, in the high productivity regime ($\alpha>0$),  a weakening of the facilitation or strengthen of competition leads to inhomogeneous distributions
in a parameter regime where the vegetation would be uniform in the absence of coupling. 
Mathematically, this can be attributed to the declined homogenization effect of the discrete-Laplacian.

\section{Discussion and Conclusions} \label{conclusions}
In this work, we considered the discrete counterpart of the Lefever-Lejeune equation describing spatial and temporal dynamics of vegetation in dry-lands, in a $1D$-dimensional lattice. 
Spatial competition and facilitation mechanisms are present in the lattice system through the coupling terms involving the discrete bi-harmonic and Laplacian operators.
First, by taking the advantage of the properties of the discrete phase space, we were able to prove decay estimates, or uniform in time bounds for the solutions, for certain parametric regimes. Both cases are of significant physical importance: the former is associated with the extinction of vegetation densities, while the latter  to the existence of attracting sets.
Since it is unclear if the system possesses a gradient structure, we performed a linear stability analysis of spatial homogeneous equillibria, to get insight on its time-asymptotic states. We identified thresholds on the discretization parameter which separate convergence to spatial homogeneous steady-states, from  the convergence to inhomogeneous ones when the discretization parameter is below this threshold. Remarkably, in the case of periodic boundary conditions, the linear stability analysis around  uniform states, apart from predicting the occurrence of periodic patterns, provided analytical criteria for determining their possible period and shape, in terms of the lattice parameters.

The analytical results corroborated with numerical simulations, revealed-in full agreement with the analytical predictions-that the facilitation effect, has a crucial impact on the extinction and pattern formation dynamics.  In the case of a finite lattice imposed by Dirichlet boundary conditions (vanishing vegetation density on the boundary), we found that extinction can even occur in the high productivity regime, when the facilitation strength is sufficiently large. On the contrary, when the facilitation strength is sufficiently small, the system may exhibit various spatial patterns as asymptotic states.  

Regarding the equilibrium set, it was revealed that it possesses a much richer structure  when the coupling parameter assumes intermediate values below its threshold value, than when it is approaching the continuous limit: space-periodic equilibria due to a Turing-like destabilization of a uniform state and  localized structures such as single spikes or multi-spikes may emerge in the low productivity regime. In the high productivity regime,  hybrid equilibria consisting of mosaics of periodic states, and localized states in a spatial periodic background, may occur as asymptotic states for the system.   Moreover, the spatially inhomogeneous and localized solutions arise for smaller
values of the productivity parameter in the lattice, than those considered in the continuous limit, shrinking the extinction regime with respect to the productivity parameter, and 
forming a wider bi-stability region between inhomogeneous in space solutions and the trivial steady-state, along the productivity gradient.  

The above results constitute a starting point for further explorations. An important, yet incomplete issue, concerns a detailed bifurcation analysis for the emergence of the aforementioned spatial-inhomogeneous states and the potential detection of snaking effects \cite{snake1,snake2,snake3}.   Of exceptional interest may be the study of the higher-dimensional lattice, as exciting pattern formation dynamics may emerge  due to the interplay of discreteness and higher-dimensionality, see \cite{hiD2,hiD1}. Another interesting direction is leading to the study of relevant coupled lattice systems \cite{sys0,hiD2,sys1,sys2,sys3}.  Studies revolving around the above themes are in progress and will be reported in future works.
\section*{Acknowledgments}
The authors gratefully acknowledge the support of the grant MIS 5004244 under the action $\mathrm{E\Delta BM34}$, funded by European Social Fund (ESF) and Hellenic General Secretariat of Research and Technology (GSRT).\\ 
\appendix
\label{App}
\section{Phase spaces and continuity properties of discrete linear operators.}
The standard sequence spaces
\begin{eqnarray}
\label{eq6}
{\ell}^p:=\left\{U=(U_n)_{n\in\mathbb{Z}}\in\mathbb{R}:\quad
\|U\|_{\ell^p}:=\left(\sum_{n\in\mathbb{Z}}|U_n|^p\right)^{\frac{1}{p}}<\infty\right\},
\end{eqnarray}
for $1\leq p\leq\infty$,  come into play in the case of the infinite lattice supplemented with the vanishing boundary conditions \eqref{vanv}. They  posses the inclusion relation 
\begin{eqnarray}
\label{eq7}
\ell^q\subset\ell^p,\quad \|U\|_{\ell^p}\leq \|U\|_{\ell^q}, \quad 1\leq q\leq p\leq\infty,
\end{eqnarray}	
which is one of the key properties for the manipulation of the higher-order nonlinearities. It is always useful to highlight that it is a reverse order inclusion relation with respect to the ordering of the exponents $q\leq p$, if compared with inclusion relation $L^p(\Omega)\subset L^q(\Omega)$ of the continuous spaces $L^p(\Omega)$ of measurable functions when the $\Omega\subseteq\mathbb{R}^N$ has finite measure.

In the case of the periodic boundary conditions, the system is considered in the spaces of periodic sequences
\begin{eqnarray}
\label{eq20}
{\ell}^p_{\mathrm{per}}:=\left\{U=(U_n)_{n\in\mathbb{Z}}\in\mathbb{R}:\quad U_n=U_{n+N},\quad
\|U\|_{\ell^p_{\mathrm{per}}}:=\left(h\sum_{n=0}^{N-1}|U_n|^p\right)^{\frac{1}{p}}<\infty\right\}, \quad 1\leq p\leq\infty,
\end{eqnarray}
while in the case of the Dirichlet boundary conditions, the system is considered in the finite dimensional subspaces of $\ell^p$
\begin{eqnarray}
\label{eq21}
\ell^p_0=\left\{U\in\ell^p\;:\;U_0=U_{N}=0\right\},
\end{eqnarray}
endowed with the same norm given as \eqref{eq20}.  Interested in a potential approximation of the continuum limit, we defined the norm ${\ell}^p_{\mathrm{per}}$, by following the standard 
numerical approximation of the $L^p(\Omega)$-norm. In this setting,
the case $h=O(1)$ corresponds to the discrete regime of the system, while, as noted above, when $h\rightarrow0$, approximates the continuous LL-counterpart \eqref{eq5}. Evidently, both cases of boundary conditions give rise to a finite dimensional dynamical system (for the periodic lattice we may restrict the dynamics on the fundamental interval $\Omega$), on $\mathbb{R}^{N+1}$. 
In both cases of the finite dimensional subspaces, instead of the inclusion relation \eqref{eq7} we shall use the equivalence of norms in $\mathbb{R}^{N+1}$:
\begin{eqnarray}
\label{eq22}
||U||_{\ell^q}\leq ||U||_{\ell^p}\leq N^{\frac{(q-p)}{qp}}||U||_{\ell^q},\;\;1\leq p\leq q<\infty,
\end{eqnarray}
using fro brevity the same symbol $||\cdot||_{\ell^p}$ for the norms in finite dimesnional cases.
\begin{lemma}
	\label{LA1}
	Let $Z={\ell^p_{\mathrm{per}}}, \ell^p_{0}$, $1\leq p\leq\infty$. The discrete Laplacian $\Delta_d:H\rightarrow \ell^p$ satisfies the inequality $\|\Delta_dU\|_{\ell^p}\leq C\|U\|_Z$, for all $U\in Z$.
\end{lemma}
\begin{proof} By the definition of $\Delta_d$ (see \eqref{eq3}) and the definition of $\ell^p_{\mathrm{per}}$ (see \eqref{eq20}), we have:
	\begin{equation}\label{ineq:(a)}
	\begin{aligned}
	\|\Delta_dU\|_{\ell^p}^p=&h \sum_{n=0}^{N-1}\left|U_{n+1}-2 U_n+U_{n-1}  \right|^p \\
	\leq& C\left(\sum_{n=0}^{N-1}|U_{n+1}|^p+h\sum_{n=0}^{N-1}|U_n|^p+h\sum_{n=0}^{N-1}|U_{n-1}|^p\right)\\
	=& C(hS_1+hS_2+hS_3),
	\end{aligned}
	\end{equation}
	for some positive constant $C>0$ depending on $p$. For the sums $S_1$ and $S_2$ , we perform the following change of variables:
	Let $j=n+1$. Then $n=0$ implies $j=1$ and for $n=N-1$, $j=(N-1)+1=N$.
	Thus, $S_1=\sum_{j=1}^{N}|u_{j}|^2$.  Similarly, for $S_3$: let $j=n-1$. Then for $n=0$, $j=-1$ and for $n=N-1$, $j=(N-1)-1=N-2$.
	Therefore, (\ref{ineq:(a)}) is rewritten as
	\begin{equation}\label{ineq:(b)}
	\begin{aligned}
	\|\Delta_dU\|_{\ell^p}^p\leq &C\left(h\sum_{j=1}^{N}|U_{j}|^p+h\sum_{j=0}^{N-1}|U_{j}|^p+h\sum_{j=-1}^{N-2}|U_{j}|^p\right)\\
	=&C[h\left(|U_1|^p+|U_2|^p+\dots+|U_N|^p\right)+h\left(|U_0|^p+|U_1|^p \dots+|U_{N-1}|^p \right) \\
	&+  h\left(|U_{-1}|^p+|U_0|^p+\dots+|U_{N-2}|^p\right)].
	\end{aligned}
	\end{equation}
	However, due to the periodic boundary conditions $U_{-1} = U_{N-1}$ (also $U_0=U_N$),  the inequality (\ref{ineq:(b)}) becomes: 
	\begin{equation}\label{ineq:(c)}
	\begin{aligned}
	\|\Delta_dU\|_{\ell^p}^p\leq &C[h\left(|U_0|^p+|U_1|^p+|U_2|^p+\dots+|U_{N-1}|^p\right )+h\left(|U_0|^p+|U_1|^p+\dots|U_{N-1}|^p\right)\\
	&+h\left(|U_0|^p+\dots+|U_{N-2}|^p+|U_{N-1}|^p\right)]\leq Ch\sum_{j=1}^{N-1}|U_j|^p=C\|U\|_{\ell^p_{\mathrm{per}}}^p.
	\end{aligned}
	\end{equation}
	and the lemma is proved, in the case where $U=\ell^2_{\mathrm{per}}$.  In the case of $U=\ell^p_0$, we note the following: when considering the standard Dirichlet boundary conditions, 
	for the first term of (\ref{ineq:(b)}), since $u_N=0$ and $u_0=0$ we have that  $\sum_{n=1}^{N-1}|U_n|^p=\sum_0^{N} |U_n|^p= \|U\|_{\ell^p}^p$,
	while, for the second term of (\ref{ineq:(b)}), since $u_0=0$, we have $\sum_{n=1}^{N-1}| U_n|^p=\| U \|_{\ell^p}^p$.
	For the third term, we note that  $U_{-1}=0$ by definition.
	Then, the sum in the third term becomes $\sum_{n=0}^{N-2}|U_n|^p\leq \sum_{ n=0}^{N-1 }|U_n|^p$. Inserting all the above in (\ref{ineq:(a)}), we derive the claimed inequality. When the Dirichlet boundary conditions of the second kind are considered, without pre-assuming zero values for the nodes $U_{-1}$ and $U_{N+1}$, then due to the antisymmetric boundary conditions \eqref{D1}-\eqref{D2}, we have obviously only a  modification of the generic constant $C$. 
\end{proof}
\begin{lemma}
	\label{LA2}
	Let $Z={\ell^p_{\mathrm{per}}}, \ell^p_{0}$. The discrete biharmonic $\Delta_d^2:H\rightarrow \ell^p$ satisfies the inequality $\|\Delta_d^2U\|_{\ell^p}\leq C\|U\|_Z$, for all $U\in Z$.
\end{lemma}
\begin{proof}
	By the definition of the biharmonic operator \eqref{eq4} and of $\ell^2_{\mathrm{per}}$, we have:
	\begin{equation}\label{eq:I}
	\begin{aligned}
	\|\Delta_d^2U \|_{\ell^p}^p\leq C\left(h\sum_{n=0}^{N-1}|U_{n+2}|^p+h\sum_{n=0}^{N-1}|U_{n+1}|^p+h\sum_{n=0}^{N-1}|U_n|^2
	+h\sum_{n=0}^{N-1}|U_{n-1}|^p+h\sum_{n=0}^{N-1}|U_{n-2}|^p\right).
	\end{aligned}												
	\end{equation}
	We only need to consider the first and the last term of the right hand side of the above inequality. The rest can be treated as in the previous Lemma \ref{LA1}.
	Again we use change of variables. For the sum of the first term we have:
	\begin{equation}
	\label{aux1}
	\begin{aligned}
	\sum_{n=0}^{N-1}|U_{n+2}|^p= \sum_{j=2}^{N+1}|U_j|^p
	=\left(|U_2|^p+|U_3|^p+\dots+|U_N|^p+|U_{N+1}|^p\right)
	=\sum_{n=0}^{N-1}|U_n|^p,
	\end{aligned}
	\end{equation}
	where we have set $j=n+2$  (so that for $n=0$, we get $j=-2$ and for $n=N-1$, we get $j=N+1$), and used the periodic boundary conditions.
	Similarly, for  the sum of the last term we have:
	\begin{equation}
	\label{aux2}
	\begin{aligned}
	\sum_{n=0}^{N-1}|U_{n-2}|^p=&\sum_{j=-2}^{N-3} |U_j|^p=\left(|U_{-2}|^p +|U_{-1}|^p+|U_{0}|^p+\dots+|U_{N-4}|^p+|U_{N-3}|^p\right)\\
	=& \left( |U_{N-2}|^p+|U_{N-1}|^p+|U_0|^p+\dots |U_{N-4}|^p+|U_{N-3}|^p\right) \\
	=& \left(|U_0|^p+|U_1|^p+\dots |U_{N-4}|^p+|U_{N-3}|^p+|U_{N-2}|^p+|U_{N-1}|^p\right) \\
	=&\sum _{n=0}^{N-1}|U_n|^p.
	\end{aligned}
	\end{equation}
	Then,   by inserting equations \eqref{aux1} and \eqref{aux2} in (\ref{eq:I}), we conclude with the proof, in the case of $U=\ell^p_{\mathrm{per}}$.  In the case of Dirichlet boundary conditions where $U\in\ell^2_0$, instead of \eqref{aux2}, we have the inequality:
	\begin{equation}
	\begin{aligned}
	\sum_{n=0}^{N-1}|U_{n-2}|^p=&\sum_{j=-2}^{N-3} |U_j|^p=\left(|U_{-2}|^p +|U_{-1}|^p+|U_{0}|^p+\dots+|U_{N-4}|^p+|U_{N-3}|^p\right)\\
	&\leq \left( |U_0|^p+|U_1|^p+\dots |U_{N-3}|^p+|U_{N-2}|^p+|U_{N-1}|^p\right) \\
	&=\sum _{n=0}^{N-1}|U_n|^p,
	\end{aligned}
	\end{equation}
	and similarly we manipulate the first term of \eqref{eq:I}. When the generic antisymmetric boundary conditions \eqref{D1}-\eqref{D2} are considered, still the proof is modified up to the generic constant $C$.
\end{proof}


\end{document}